\documentclass{article} % IF YOU WANT TO COMPILE - UNCOMMENT
\usepackage{graphicx}
\usepackage{clipboard}
\usepackage{tikz}
\usepackage{appendix}
\usetikzlibrary{arrows,automata,patterns}
\usepackage{enumerate}
\usepackage{makecell}

\usepackage{silence}
\usepackage{hyperref}

\usepackage[linesnumbered,ruled,vlined,boxed,algo2e]{algorithm2e}
\usepackage{nicefrac}
\usepackage{amsmath}
\usepackage{float}
\usepackage{mathtools}
\usepackage{amsthm}
\usepackage{multirow}
\usepackage{fullpage}
\usepackage{amsfonts}

\newtheorem{theorem}{Theorem}

\newtheorem{lemma}{Lemma}

\newtheorem{claim}{Claim}[lemma]

\newtheorem{property}{Property}[lemma]

\newtheorem{corollary}{Corollary}[theorem]

\newtheorem{problem}{Problem}

\begin{document}

% General
\newcommand{\defeq}{:=}
\newcommand{\eps}{\varepsilon}
% Pseudocode
\newcommand{\ReturnCode}{\textbf{return}}

% Comments
\newcommand{\sidford}[1]{{\color{ForestGreen} \textbf{Aaron}: #1}} 
\newcommand{\liam}[1]{{\color{blue} \textbf{Liam}: #1}} 
\newcommand{\avi}[1]{{\color{purple} \textbf{Avi}: #1}} 
\newcommand{\vvw}[1]{{\color{Plum} \textbf{Virgi}: #1}} 
\newcommand{\uri}[1]{{\color{red} \textbf{Uri}: #1}} 
\newcommand{\blue}[1]{{\color{blue}#1}}

% Pseudocode
\newcommand{\codestyle}[1]{\texttt{#1}}
\newcommand{\Initialize}{\mbox{\codestyle{Initialize}}}
\newcommand{\Cycle}{\mbox{\codestyle{Cycle}}}
\newcommand{\CycleOdd}{\codestyle{CycleOdd}}
\newcommand{\BallOrCycle}{\codestyle{BallOrCycle}}
\newcommand{\ClusterOrCycleBounded}{\codestyle{ClusterOrCycleBounded}}
\newcommand{\ClusterOrCycle}{\codestyle{ClusterOrCycle}}
\newcommand{\SimpleCycle}{\codestyle{SimpleCycle}}
\newcommand{\Next}{\codestyle{Next}}
\newcommand{\Sample}{\codestyle{Sample}}
\newcommand{\Dijkstra}{\codestyle{Dijkstra}}
\newcommand{\Preprocess}{\codestyle{Preprocess}}
\newcommand{\HashTable}{\codestyle{HashTable}}
\newcommand{\Heap}{\codestyle{Heap}}
\newcommand{\RelaxNext}{\codestyle{RelaxNext}}
\newcommand{\Reminder}[1]{

\vspace{0.5em}
\noindent\textbf{Reminder of~\autoref{#1}.} \textit{\Paste{#1}}
\vspace{0.5em}

% \end{rreminder}
}

\newcommand{\ODD}{_{\normalfont\textsc{odd}}}
\newcommand{\PreprocessGraph}{\codestyle{Initialize}}
\newcommand{\Route}{\codestyle{Route}}
\newcommand{\TreeRoute}{\codestyle{TreeRoute}}
\newcommand{\N}{\mathbb{N}}
\newcommand{\MinCycle}{\codestyle{MinCycle}}
\newcommand{\Query}{\codestyle{Query}}
\newcommand{\TZQuery}{\codestyle{TZQuery}}
\newcommand{\MTZQuery}{\codestyle{MTZQuery}}
\newcommand{\Ball}{\codestyle{Ball}}
\newcommand{\DistanceOracle}{\codestyle{MTZQuery}}
\newcommand{\SparseOrCycle}{\codestyle{SparseOrCycle}}
\newcommand{\Intersection}{\codestyle{Intersection}}
\newcommand{\CycleAdditive}{\codestyle{CycleAdditive}}
\newcommand{\GenerateSi}{\codestyle{ComputeS}}
\newcommand{\ApproximateANSC}{\codestyle{ApproximateANSC}}
\newcommand{\ApproximateANSCUnweighted}{\codestyle{ApproximateANSCUnweighted}}

\newcommand{\codeNull}{\codestyle{null}}
\newcommand{\codeYes}{\codestyle{Yes}}
\newcommand{\codeNo}{\codestyle{No}}
\newcommand{\codeAnd}{~ \mathrm{and} ~}
\newcommand{\codeOr}{~ \mathrm{or} ~}
\newcommand{\wt}{\ell}
\newcommand{\Cl}{CL}
\newcommand{\CL}{CL}
\newcommand{\cl}{c\ell}

\newcommand{\LE}{\;\le\;}
\newcommand{\EQ}{\;=\;}
\newcommand{\GE}{\;\ge\;}
\newcommand{\Ot}{\tilde{O}}

\newcommand{\EE}{\mathbb{E}}
\newcommand{\RR}{\mathbb{R}}
\DeclarePairedDelimiter{\ceil}{\lceil}{\rceil}
\DeclarePairedDelimiter{\pair}{\langle}{\rangle}
\DeclarePairedDelimiter{\floor}{\lfloor}{\rfloor}

\ActivateWarningFilters[pdftoc]
\title{New approximate distance oracles and their applications}
\date{}
\author{Avi Kadria\thanks{Department of Computer Science, Bar Ilan University, Ramat Gan 5290002, Israel. E-mail {\tt avi.kadria3@gmail.com}.} \and Liam Roditty\thanks{Department of Computer Science, Bar Ilan University, Ramat Gan 5290002, Israel. E-mail {\tt liam.roditty@biu.ac.il}. Supported in part by BSF grants 2016365 and 2020356.}}
% \author{}
\maketitle

\thispagestyle{empty}

\begin{abstract}
Let $G = (V, E)$ be an undirected graph with $n$ vertices and $m$ edges, and let $\mu = m/n$.
A \emph{distance oracle} is a data structure designed to answer approximate distance queries, with the goal of achieving low stretch, efficient space usage, and fast query time.
While much of the prior work focused on distance oracles with constant query time, this paper presents a comprehensive study of distance oracles with non-constant query time.
We explore the tradeoffs between space, stretch, and query time of distance oracles in various regimes. Specifically, we consider both weighted and unweighted graphs in the regimes of stretch $< 2$ and stretch $\ge 2$.
In addition, we demonstrate several applications of our new distance oracles to the $n$-Pairs Shortest Paths ($n$-PSP) problem and the All Nodes Shortest Cycles ($ANSC$) problem.
Our main contributions are:

\begin{itemize}
    \item \textbf{Weighted graphs:} We present a new three-way trade-off between stretch, space, and query time, offering a natural extension of the classical Thorup–Zwick distance oracle [STOC’01 and JACM’05] to regimes with larger query time. Specifically, for any $0 < r < 1/2$ and integer $k \ge 1$, we construct a $(2k(1 - 2r) - 1)$-stretch distance oracle with $\tilde{O}(m + n^{1 + 1/k})$ space and $\tilde{O}(\mu n^r)$ query time.
    This construction provides an asymptotic improvement over the classical $(2k - 1)$-stretch and $O(n^{1 + 1/k})$-space tradeoff of Thorup and Zwick in sparse graphs, at the cost of increased query time.
    We also improve upon a result of Dalirrooyfard et al.~[FOCS’22], who presented a $(2k - 2)$-stretch distance oracle with $O(m + n^{1 + 1/k})$ space and $O(\mu n^{1/k})$ query time.
    In our oracle we reduce the stretch from $(2k - 2)$ to $(2k - 5)$ while preserving the same space and query time.
  \item \textbf{Unweighted graphs:} 
  We present a
  $(2k - 5, 4 + 2\ODD)$-approximation\footnote{$2\ODD=2\cdot (d(u,v)\mod 2)$.} distance oracle  with $O(n^{1 + 1/k})$ space and $O(n^{1/k})$ query time. This improves upon 
  a $(2k - 2, 2\ODD)$-approximation distance oracle of Dalirrooyfard et al.~[FOCS’22] while maintaining the same space and query time. 
    We also present a distance oracle that given $u,v\in V$ returns an estimate $\hat{d}(u,v) \le d(u,v) + 2\lceil d(u,v) / 3 \rceil + 2$, using $O(n^{4/3 + 2\eps})$ space and $O(n^{1 - 3\eps})$ query time. To the best of our knowledge, this is the first distance oracle that simultaneously achieves a multiplicative stretch $<2$, and a space complexity $O(n^{1.5 - \alpha})$, for some $\alpha > 0$.
  
  \item \textbf{Applications for $n$-PSP and $ANSC$:} 
  We present an $\tilde{O}(m^{1 - 1/(k+1)} n)$-time algorithm for the $n$-PSP problem, that for every input pair $\langle s_i,t_i \rangle$, where $i\in [n]$, returns an estimate $\hat{d}(s_i, t_i)$ such that $\hat{d}(s_i,t_i) \le d(s_i,t_i) + 2\lceil d(s_i,t_i) / 2k \rceil$. By allowing a small additive error, this result circumvents the conditional running time lower bound of $\Omega\left(m^{2 - \frac{2}{k+1}} \cdot n^{\frac{1}{k+1} - o(1)}\right)$, established by Dalirrooyfard et al.~[FOCS’22] for achieving $(1 + 1/k)$-stretch.
  Additionally, we present an $\Ot(mn^{1 - 1/k})$-time algorithm for the $ANSC$ problem that computes, for every $u \in V$, an estimate $\hat{c}_u$ such that $\hat{c}_u \le SC(u) + 2\lceil SC(u)/2(k - 1) \rceil$, where $SC(u)$ denotes the length of the shortest cycle containing $u$. This improves upon the $\Ot(m^{2 - 2/k}n^{1/k})$-time algorithm of Dalirrooyfard et al.~[FOCS'22], while achieving the same approximation guarantee.
\end{itemize}

We obtain our results by developing several new techniques, among them are the \emph{borderline vertices} technique and the \emph{middle vertex} technique,  which may be of independent interest.
\end{abstract}
\clearpage
\pagenumbering{arabic} 
\newpage
\section{Introduction}\label{S-Int}
Let $G=(V,E)$ be an undirected graph with $n=|V|$ vertices and $m=|E|$ edges with non-negative real edge weights, and let $\mu=m/n$ be the average degree. 
The distance $d_G(u,v)$ between $u$ and $v$ is the length of the shortest path\footnote{We omit $G$ when it is clear from the context.} between $u$ and $v$ in $G$.
An estimation $\hat{d}(u,v)$ of $d(u,v)$ is an  
$\alpha$-stretch if $d(u,v)\leq \hat{d}(u,v) \leq \alpha \cdot d(u,v)$.

In their seminal work, Thorup and Zwick~\cite{DBLP:journals/jacm/ThorupZ05} presented a data structure for distance approximations called a \textit{distance oracle}.
A distance oracle is a data structure with $o(n^2)$ space (otherwise we can trivially store the distance matrix), that supports distance queries between vertices.
For any integer $k\geq 1$, they~\cite{DBLP:journals/jacm/ThorupZ05} showed that in $O(kmn^{1/k})$ expected time, it is possible to preprocess a weighted undirected graph 
and create a distance oracle of size $O(kn^{1+1/k})$. For every $u,v \in V$ the query of the distance oracle returns a $(2k-1)$-stretch for $d(u,v)$ in $O(k)$ time.
Different aspects of distance oracles, such as construction time, query time, and stretch, have been studied since their introduction, more than two decades ago. For more details see for example~\cite{DBLP:journals/siamcomp/BaswanaK10, DBLP:conf/focs/MendelN06,Nilsen12,DBLP:conf/stoc/Chechik14,DBLP:conf/stoc/Chechik15,ElkinNW16,ElkinP16,DBLP:conf/focs/ElkinS23, DBLP:conf/icalp/Chechik024a, DBLP:conf/stoc/AbboudBKZ22,DBLP:conf/stoc/AbboudBF23, DBLP:conf/stoc/BiloCCC0KS23, DBLP:journals/corr/abs-2307-11677Bilo23ADO, bilo2024improved}.

Thorup and Zwick~\cite{DBLP:journals/jacm/ThorupZ05} employed Erd\H{o}s' girth conjecture\footnote{The girth of a graph is the length of its shortest cycle.} to establish a lower bound on the space/stretch tradeoff for distance oracles. The conjecture asserts the existence of graphs with $\Omega(n^{1+1/k})$ edges and girth $\ge 2k+2$.
Under this assumption, they proved that any distance oracle with stretch $t \leq 2k-1$ must use $\Omega(n^{1+1/k})$ bits on some input, for all $k \geq 1$. However, these lower bounds apply to graphs with $m = \Omega(n^{1+1/k})$ edges.

This motivates studying the $(2k-1)$-stretch/$O(n^{1+1/k})$-space tradeoff in sparser graphs, where $m = o(n^{1+1/k})$. \footnote{We remark that a distance oracle in such graphs must use $\Omega(m)$ space, because of similar arguments to the lower bound of~\cite{DBLP:journals/jacm/ThorupZ05}.}
% \footnote{Assuming the girth conjecture, any such result must still incur $\Omega(m)$ space, since subgraphs of graphs with girth $\ge 2k+2$ preserve the girth.}.
This approach has been applied in two different stretch regimes. Distance oracles with stretch $\ge 2$ (see, for example, \cite{DBLP:journals/corr/abs-1201-2703, DBLP:conf/esa/Agarwal14, DBLP:conf/infocom/AgarwalGH11, PatrascuRT12, DBLP:conf/wdag/AbrahamG11, ElkinP16}) and distance oracles with stretch $<2$ (see, for example, \cite{DBLP:conf/esa/Agarwal14, DBLP:conf/soda/AgarwalG13, DBLP:conf/icalp/Chechik024a, DBLP:conf/icalp/KopelowitzKR24, DBLP:journals/algorithmica/PoratR13}).

\subsection{Distance oracles with stretch $\ge 2$}

Our main result for stretch~$\ge 2$ establishes a three-way tradeoff among the key parameters of distance oracles: stretch, space, and query time.  
This new tradeoff improves upon previously known results for specific values of stretch, space, and query time.  
We prove the following:

\begin{theorem}[\autoref{S-2k-1-4c-Weighted}] \label{T-DO-2k-1-4c-Weighted}  \Copy{T-DO-2k-1-4c-Weighted}{
    Let $k \ge 5$ and let $0 < c < \frac{k}{2}-1$ be an integer.
    There is an $\Ot(m+n^{1+1/k})$ space distance oracle that given two query vertices $u,v\in V$ computes in  $O(\mu n^{c/k})$-time 
    a distance estimation $\hat{d}(u,v)$ that satisfies
    $d(u,v)\le \hat{d}(u,v) \le (2k-1-4c)\cdot d(u,v)$. The distance oracle is constructed in $\Ot(mn^{\frac{c+2}{k}})$ expected time. 
    }
\end{theorem}
An interesting tradeoff that follows from~\autoref{T-DO-2k-1-4c-Weighted} is a better space/stretch tradeoff than the classical Thorup and Zwick~\cite{DBLP:journals/jacm/ThorupZ05} $(2k-1)$-stretch / $\Ot(n^{1+1/k})$ space tradeoff, at a cost of a larger query time. 
Specifically, by setting $r=\frac ck$, we get a $(2k(1-2r)-1)$-stretch distance oracle that uses $\Ot(m+n^{1+1/k})$ space, at the price of $\Ot(\mu n^{r})$ query time.
% For example, if we allow a  $\Ot(\mu n^{\frac{1}{3}})$ query time we get a $(\frac{2}{3}k-1)$-stretch distance oracle that uses $\Ot(m+n^{1+1/k})$ space.
By setting $c=1$ in~\autoref{T-DO-2k-1-4c-Weighted}, we improve upon a recent result of  Dalirrooyfard, Jin, V. Williams, and Wein~\cite{DBLP:conf/focs/DalirrooyfardJW22}. They presented a $(2k-2)$-stretch distance oracle that uses $\Ot(m+n^{1+1/k})$ space and $\Ot(\mu n^{1/k})$ query time.\footnote{We remark that even more recently, Chechik, Hoch, and Lifshitz~\cite{DBLP:conf/soda/ChechikHL25} presented a $(2k-3)$-stretch $n$-PSP algorithm, however, they did not present a new distance oracle.
% and they used the fact the queries are known in advance.
% but without constructing a querying a distance oracle.
}
We improve the stretch from $(2k-2)$ to $(2k-5)$  while using the same space and query time.

To obtain~\autoref{T-DO-2k-1-4c-Weighted}, we develop a new technique called the \textit{borderline vertices} technique. Given $u,v\in V$, the borderline vertices 
$\tau_u$ and $\tau_v$ are two vertices on the shortest path between $u$ and $v$ that satisfy special properties that allow us to better exploit the structure of the  
Thorup and Zwick distance oracle. 
Using the borderline vertices technique, we also manage to achieve the following distance oracle, which improves upon the construction time of~\autoref{T-DO-2k-1-4c-Weighted} at the cost of increasing the stretch.
\begin{theorem}[\autoref{S-2k-4c-weighted}]\label{T-2k-4c-Weighted}\Copy{T-2k-4c-Weighted}{
    Let $k \ge 4$ and let $0 < c < \frac{k}{2}-1$ be an integer.
    There is an $\Ot(m+n^{1+1/k})$ space distance oracle that given  two query vertices $u,v\in V$ computes in $\Ot(\mu n^{c/k})$-time 
    a distance estimation $\hat{d}(u,v)$ that satisfies
    $d(u,v)\le \hat{d}(u,v) \le (2k-4c)\cdot d(u,v)$. The distance oracle is constructed in $\Ot(mn^{\frac{c+1}{k}})$ expected time. }
\end{theorem}

Agarwal, Godfrey, and Har{-}Peled~\cite{DBLP:conf/infocom/AgarwalGH11}  considered also small stretches, and presented a $2$-stretch distance oracle that uses $\Ot(m+n^{3/2})$ space and has $\Ot(\mu n^{1/2})$ query time, that is constructed in $\Ot(mn^{1/2})$ time, and a $3$-stretch distance oracle that uses $\Ot(m+n^{4/3})$ space and has $\Ot(\mu n^{1/3})$ query time, and is constructed in $\Ot(mn^{2/3})$ time.
These results suggest that the following general stretch/space/query time tradeoff may exist for every $k$.
\begin{problem}\label{Plm-1}
    For which values of $k$ it is possible to construct a $k$-stretch distance oracle that uses $\Ot(m+n^{1+1/k})$ space and has $\Ot(\mu n^{1/k})$ query time.
\end{problem}

Agarwal, Godfrey, and Har{-}Peled~\cite{DBLP:conf/infocom/AgarwalGH11} solved~\autoref{Plm-1} for $k=2,3$. 
By setting $c=1$ in~\autoref{T-2k-4c-Weighted} and~\autoref{T-DO-2k-1-4c-Weighted}, respectively, we 
obtain a $4$-stretch distance oracle that uses $\Ot(m+n^{5/4})$ space and has $\Ot(\mu n^{1/4})$ query time, and a $5$-stretch distance oracle that uses $\Ot(m+n^{6/5})$ space and has $\Ot(\mu n^{1/5})$ query time. Thus, we solve~\autoref{Plm-1} for $k=4,5$.
% 
% Furthermore, the borderline vertices technique enables us to achieve the following distance oracle.
% \begin{theorem}[\autoref{S-2k-4c-weighted}]\label{T-2k-3-Weighted} \Copy{T-2k-3-Weighted} 
% {
%     Let $k \ge 3$ be an integer.
%     There is an $\Ot(m+n^{1+1/k})$ space distance oracle that given  two query vertices $u,v\in V$ computes in $O(\mu n^{1/k})$-time a distance estimation $\hat{d}(u,v)$ that satisfies
%     $d(u,v)\le \hat{d}(u,v) \le (2k-3)\cdot d(u,v)$. The distance oracle is constructed in $\Ot(mn^{1/k})$ expected time. 
% }
% \end{theorem}
% By setting $k=3$ in~\autoref{T-2k-3-Weighted} we improve the construction time of the $3$-stretch distance oracle of Agarwal, Godfrey and Har{-}Peled~\cite{DBLP:conf/infocom/AgarwalGH11} from $\Ot(mn^{2/3})$ to $\Ot(mn^{1/3})$.
% Recently, Chechik, Hoch, and Lifshitz~\cite{DBLP:conf/soda/ChechikHL25} considered the problem of approximating $n$ pairs shortest paths ($n$-PSP) and presented a $(2k-3)$-stretch in $O(mn^{1/k})$ time, which can be also achieved by constructing the distance oracle of~\autoref{T-2k-3-Weighted} and querying it $n$ times. However, in their result, they do not create a distance oracle that ensures worst-case query time.

In addition, using the borderline vertices technique we obtain $(2k-3)$-stretch distance oracle that uses  $\Ot(m+n^{1+1/k})$ space, has $\Ot(\mu n^{1/k})$ query time, and is constructed in $\Ot(m n^{1/k})$ time (see~\autoref{T-2k-3-Weighted}).
For $k=3$ this improves the construction time of~\cite{DBLP:conf/infocom/AgarwalGH11} for a $3$-stretch distance oracle from  $\Ot(mn^{2/3})$ to $\Ot(mn^{1/3})$. Our results for weighted graphs are summarized in~\autoref{tab:weighted-new}.

Next, we turn our focus to \textit{unweighted} graphs. 
An estimation $\hat{d}(u,v)$ of $d(u,v)$ is an  
$(\alpha, \beta)$-approximation if $d(u,v)\leq \hat{d}(u,v) \leq \alpha d(u,v)+\beta$.
$(\alpha,\beta)$-approximations were extensively studied in the context of distance oracles, graph spanners, and emulators. (For more details see for example~\cite{DBLP:journals/siamcomp/ElkinP04, DBLP:conf/soda/ThorupZ06,DBLP:journals/talg/BaswanaKMP10,DBLP:conf/soda/Chechik13,VassilevskaSp15,DBLP:conf/icalp/Parter14,DBLP:journals/jacm/AbboudB17,DBLP:journals/siamcomp/AbboudBP18,DBLP:conf/infocom/AgarwalGH11,DBLP:journals/corr/abs-1201-2703, DBLP:journals/corr/abs-2307-11677Bilo23ADO}).

From the girth conjecture, it follows that in unweighted graphs, for every $(\alpha,\beta)$-distance oracle that uses $\Ot(n^{1+1/k})$ space, it must hold that $\alpha+\beta \ge 2k-1$.
However, this inequality must hold only for adjacent vertices. For non-adjacent vertex pairs, we prove the following simple lower bound.
\begin{theorem}[\autoref{S-LB-appendix}]\label{T-LB-2} \Copy{T-LB-2} {
    Assuming the Erd\H{o}s girth conjecture, any distance oracle that uses $o(n^{1+1/k})$ space must have an input graph $G=(V,E)$ and two vertices $u,v\in V$ such that $d(u,v)=2$ and $\hat{d}(u,v) \ge 2k+2$.
}
\end{theorem}

A \textit{spanner} is a subgraph $H\subseteq G$ that approximates distances \textit{without} supporting distance queries. Baswana, Kavitha, Mehlhorn, and Pettie~\cite{DBLP:journals/talg/BaswanaKMP10} and Parter~\cite{DBLP:conf/icalp/Parter14} presented a $k$-spanner for every $u,v\in V$ such that $d(u,v)=2$, and matched the lower bound of~\autoref{T-LB-2} for spanners. 
An interesting question is whether this lower bound can be matched by distance oracles with efficient distance queries as well.
% (that not only preserve distances but also support distance queries). 
Thus, we formulate the following problem.
\begin{problem}\label{Plm-2}
    For which values of $k$  it is possible to construct a distance oracle with $\Ot(n^{1+1/k})$ space that uses $\Ot(n^{1/k})$ query time,
    such that $d(u,v) \le \hat{d}(u,v) \le k\cdot d(u,v)$, for every $u,v\in V$ that satisfy $d(u,v)=2$.
\end{problem}
Agarwal, Godfrey, and Har{-}Peled~\cite{DBLP:conf/infocom/AgarwalGH11} presented 
for unweighted graphs a $(2,1)$-approximation distance oracle that uses $\Ot(n^{3/2})$ space, has $\Ot(n^{1/2})$ query time, and is constructed in $\Ot(mn^{1/2})$ time. They also presented a $(3,2)$-approximation distance oracle that uses $\Ot(n^{4/3})$ space, has $\Ot(n^{1/3})$ query time, that is constructed in $\Ot(mn^{2/3})$ time.

By analyzing more carefully the $(2,1)$-approximation distance oracle of~\cite{DBLP:conf/infocom/AgarwalGH11} it is possible to show that it is actually a $(2,1\ODD)$-approximation\footnote{$x\ODD$ is defined as $x\cdot 1\ODD$, where $1\ODD$ is $d(u,v) \mod 2$} distance oracle, and therefore solves~\autoref{Plm-2} for $k=2$.\footnote{In their work, the additive error comes from the fact that if $B(u)\cap B(v)=\emptyset$ then $\min(h(u),h(v))\le d(u,v)/2+1$. However, this can be improved by showing that if $B(u)\cap B(v)=\emptyset$ then $\min(h(u),h(v))\le d(u,v)/2+1\ODD$, see~\autoref{L-NO-Intersect}}
In the following two theorems, we show that it is possible to solve~\autoref{Plm-2} for $k=3$ and $k=4$ as well.
% The following two theorems, solves~\autoref{Plm-2} for $k=3,4$.
\begin{theorem}[\autoref{S-4-2}] \label{T-DO-3-Unweighted}\Copy{T-DO-3-Unweighted}{
There is a $\Ot(n^{4/3})$ space distance oracle that given two query vertices $u,v\in V$ computes in  $\Ot(n^{1/3})$-time 
a distance estimation $\hat{d}(u,v)$ that satisfies
$d(u,v)\le \hat{d}(u,v) \le 3d(u,v) + 2\ODD$. The distance oracle is constructed in $\Ot(mn^{1/3})$ expected time. }
\end{theorem}

\begin{theorem}[\autoref{S-DO-4-Unweighted}] \label{T-DO-4-Unweighted} \Copy{T-DO-4-Unweighted}{
There is a $\Ot(n^{5/4})$ space distance oracle that given two query vertices $u,v\in V$ computes in  $\Ot(n^{1/4})$-time 
a distance estimation $\hat{d}(u,v)$ that satisfies
$d(u,v)\le \hat{d}(u,v) \le 4d(u,v) + 3\ODD$. The distance oracle is constructed in $\Ot(mn^{1/2})$ expected time. }
\end{theorem}

In~\autoref{T-DO-3-Unweighted} we improve the distance oracle of~\cite{DBLP:conf/infocom/AgarwalGH11} from more than a decade ago. In particular we improve the construction time from $\Ot(mn^{2/3})$ to $\Ot(mn^{1/3})$, and the approximation from $(3,2)$ to $(3,2\ODD)$.
% The distance oracle of~\autoref{T-DO-3-Unweighted} improves the construction time from $\Ot(mn^{2/3})$ to $\Ot(mn^{1/3})$, and the approximation from $(3,2)$ to $(3,2\ODD)$ of the $(3,2)$-approximation distance oracle of~\cite{DBLP:conf/infocom/AgarwalGH11} from more than a decade ago. 
Dalirrooyfard, Jin, V. Williams, and Wein~\cite{DBLP:conf/focs/DalirrooyfardJW22} presented a $(2k-2,2\ODD)$-stretch distance oracle that uses $\Ot(n^{1+1/k})$ space and $\Ot(n^{1/k})$ query time.
By setting $k=4$ in the $(2k-2,2\ODD)$-stretch distance oracle of~\cite{DBLP:conf/focs/DalirrooyfardJW22} they obtain a $(6,2\ODD)$-stretch distance oracle. In~\autoref{T-DO-4-Unweighted} we improve the approximation from $(6, 2\ODD)$ to $(4,3\ODD)$ while maintaining the same space and query time.

To obtain our new distance oracles for unweighted graphs, presented in~\autoref{T-DO-3-Unweighted} and~\autoref{T-DO-4-Unweighted}, we develop a new technique for using the middle vertex in the path, called the \textit{middle vertex} technique. 
Let $u,v \in V$, the middle vertex $\tau$ is a vertex on the shortest path between $u$ and $v$ that satisfies $d(u,\tau)=d(u,v)/2+0.5\ODD$ and $d(v,\tau)=d(u,v)/2-0.5\ODD$. 
As with the borderline vertices technique for \textit{sparse weighted} graphs, the middle vertex technique enables us to better exploit the structure of the distance oracles of Thorup and Zwick in \textit{dense unweighted} graphs.

% Using the properties of $\tau$ and $B_0(\tau)$, we bound $h_2(u)$ by $d(u,v)+1\ODD$ rather than $1.5d(u,v)+1\ODD$ as in~\cite{DBLP:conf/focs/DalirrooyfardJW22} and~\cite{DBLP:conf/infocom/AgarwalGH11}.
% This new technique together with our new borderline vertices technique (adapted for unweighted graphs), we obtain all of our new distance oracles for unweighted graphs.

By combining the middle vertex technique with the borderline vertices technique, we manage to achieve two more distance oracles for dense unweighted graphs. 
The first distance oracle improves the approximation of the distance oracle of~\autoref{T-DO-4-Unweighted} from $(4,3\ODD)$ to $(3,2+2\ODD)$ at the cost of increasing the query time to $\Ot(n^{1/2})$. 
\begin{theorem}[\autoref{S-3,2+2ODD-4}] \label{T-DO-3-2-Unweighted} \Copy{T-DO-3-2-Unweighted} {
There is a $\Ot(n^{5/4})$ space distance oracle that given  two query vertices $u,v\in V$ computes in  $\Ot(n^{1/2})$-time 
a distance estimation $\hat{d}(u,v)$ that satisfies
$d(u,v)\le \hat{d}(u,v) \le 3d(u,v) + 2+2\ODD$. The distance oracle is constructed in $\Ot(mn^{1/2})$ expected time. }
\end{theorem}

The second distance oracle improves the $(2k-2, 2\ODD)$-approximation of~\cite{DBLP:conf/focs/DalirrooyfardJW22} to $(2k-5, 4+2\ODD)$-approximation,  while using the same space and query time. %, for every value of $k \ge 5$,  
\begin{theorem}[\autoref{S-2k-5+6-unweighted}] \label{T-DO-2k-5-unweighted} \Copy{T-DO-2k-5-unweighted} {
    Let $k \ge 5$.
    There is an $\Ot(n^{1+1/k})$ space distance oracle that given  two query vertices $u,v\in V$ computes in  $\Ot(n^{1/k})$-time 
    a distance estimation $\hat{d}(u,v)$ that satisfies
    $d(u,v)\le \hat{d}(u,v) \le (2k-5)d(u,v)+ 4+2\ODD$. The distance oracle is constructed in $\Ot(mn^{\frac{3}{k}})$ expected time. }
\end{theorem}
Our results for unweighted graphs are summarized in~\autoref{tab:unweighted}.

\subsection{Distance oracles with stretch $<2$}

For distance oracles with stretch $<2$ in weighted graphs, we present the following new three-way tradeoff between: stretch, space, and query time.  
 We prove:
% This new tradeoff improves upon   previously known results for specific values of stretch, space, and query time.  
% We prove the following:

% For distance oracles with stretch less than $2$ in weighted graphs, we present a new space/query-time tradeoff for achieving a given stretch. This tradeoff improves upon the space/stretch bound of Agarwal~\cite{DBLP:conf/esa/Agarwal14}, at the cost of increased query time.
\begin{theorem}[\autoref{S-5-3-weighted}]\label{T-Weighted-5/3}\Copy{T-Weighted-5/3}{
Let $0 < c < 1/3$ be a real constant.
There is a $\Ot(m+n^{2-2c})$-space distance oracle that given  two query vertices $u,v\in V$ computes in  $\Ot(\mu^tn^{tc})$-time 
a distance estimation $\hat{d}(u,v)$ that satisfies 
$d(u,v) \leq \hat{d}(u,v) \le (1+2/t)d(u,v)$. The distance oracle is constructed in $\Ot(mn^{1-c})$ time. }
\end{theorem}

Using this tradeoff, we obtain the first distance oracle with stretch $< 2$, sublinear query time, and $\Ot(m + n^{1.5-\alpha})$ space, for $\alpha > 0$.   
In particular, by setting $t = 3$ and $c = 1/3 - \eps$, we obtain a $5/3$-stretch distance oracle with  $\Ot(m + n^{4/3 + 2\eps})$-space and $O(\mu^3 n^{1 - 3\eps})$-query time.
For comparison, Agarwal~\cite{DBLP:conf/esa/Agarwal14} presented a distance oracle with stretch $1 + \frac{1}{t + 0.5}$, using $\Ot(m + n^{2 - c})$ space and $\Ot(\mu^t n^{(t+1)c})$ query time.  
Setting $t = 1$ and $c = 1/2 - \eps$  in the construction of~\cite{DBLP:conf/esa/Agarwal14} yields a $5/3$-stretch oracle with $\Ot(m + n^{1.5 + \eps})$-space and $\Ot(\mu n^{0.5 - \eps})$-query time. (See~\autoref{tab:weighted-new}.)

For unweighted graphs, we present two new tradeoffs.  
The first improves upon the tradeoff of the distance oracle of Bil{\`{o}}, Chechik, Choudhary, Cohen, Friedrich, and Schirneck~\cite{DBLP:journals/corr/abs-2307-11677Bilo23ADO}, while preserving the same space and query time.  
Specifically, they presented a distance oracle for unweighted graphs that uses $\Ot(n^{2 - c})$ space, has $\Ot(n^{tc})$ query time, and returns an estimate $\hat{d}(u,v) \le (1 + 1/t)d(u,v)+2$.  
In the following theorem, we slightly improve the stretch bound to $\hat{d}(u,v) \le d(u,v) + 2\ceil*{ \frac{d(u,v)}{2t}}$, while using the same space and query time.
\begin{theorem}[\autoref{S-5-t}]\label{T-Unweighted-1+1/t}\Copy{T-Unweighted-1+1/t}{
Let $t\geq 1$ be an integer and let $0 < c < 1/2$ be a real constant.
There is a $\Ot(n^{2-c})$-space distance oracle that given  two query vertices $u,v\in V$ computes in  $\Ot(n^{c\cdot t})$-time 
a distance estimation $\hat{d}(u,v)$ that satisfies 
$d(u,v) \leq \hat{d}(u,v) \le d(u,v) + 2\ceil*{ \frac{d(u,v)}{2t}}$. The distance oracle is constructed in $\Ot(mn^{1-c})$  time. }
\end{theorem}
Next,  by adapting methods of \autoref{T-Weighted-5/3} to unweighted graphs, we obtain another tradeoff that uses less space than the tradeoff of \autoref{T-Unweighted-1+1/t} at the cost of a larger query time.  
We prove:
\begin{theorem}[\autoref{S-5-3}]\label{T-Unweighted-5/3}\Copy{T-Unweighted-5/3}{
Let $t\geq 1$. Let $0 < c < 1/3$ be a real constant.
There is a $\Ot(n^{2-2c})$-space distance oracle that given  two query vertices $u,v\in V$ computes in  $\Ot(n^{tc})$-time 
a distance estimation $\hat{d}(u,v)$ that satisfies 
$d(u,v) \leq \hat{d}(u,v) \le d(u,v)  + 2\ceil{d(u,v) / t} + 2$. The distance oracle is constructed in $\Ot(mn^{1-c})$  time. }
\end{theorem}
In particular, using~\autoref{T-Unweighted-5/3} we obtain  the first distance oracle with $o(n^{1.5-\alpha})$-space, for $\alpha >0$, that achieves a multiplicative error strictly better than~$2$.  
Specifically, by setting $t = 3$ and $c = 1/3 - \eps$, we get a distance oracle with $\Ot(n^{4/3 + 2\eps})$ space and $\Ot(n^{1 - 3\eps})$ query time.
These results are summarized in~\autoref{tab:unweighted}.
\begin{table}[t]
    \centering
    \begin{tabular}{|c|c|c|c|}
     \hline
      Space & Time  & Old Approximation & New Approximation \\
     \hline\hline
     $\Ot(kn^{1+1/k})$ & $\Ot(n^{1/k})$  & $(2k-2, 2\ODD)$ \cite{DBLP:conf/focs/DalirrooyfardJW22} & $(2k-5, 4 + 2\ODD)$ (Th\ref{T-DO-2k-5-unweighted}) \\ \hline
    $\Ot(n^{4/3})$ & $\Ot(n^{1/3}$ ) & $(3,2)$~\cite{DBLP:journals/corr/abs-1201-2703} & $(3,2\ODD)$ (Th\ref{T-DO-3-Unweighted}) \\ \hline
    % $\Ot(n^{4/3})$ & $\Ot(n^{1/3})$  & $(3,2\ODD)$ & Theorem~\ref{T-DO-3-Unweighted}\\ \hline 
    $\Ot(n^{5/4})$ & $\Ot(n^{1/4})$ & $(6, 2\ODD)$ \cite{DBLP:conf/focs/DalirrooyfardJW22} & $(4, 3\ODD)$ (Th\ref{T-DO-4-Unweighted})  \\ \hline
    % $\Ot(n^{5/4})$ & $\Ot(n^{1/4})$  & $(4, 3\ODD)$ & Theorem\ref{T-DO-4-Unweighted}\\ \hline
    $\Ot(n^{5/4})$ & $\Ot(n^{1/2})$  & - & $(3, 2+2\ODD)$ (Th\ref{T-DO-3-2-Unweighted})\\ \hline 
    $\Ot(n^{6/5})$ & $\Ot(n^{1/5})$  & $(8,2\ODD)$~\cite{DBLP:conf/focs/DalirrooyfardJW22} & $(5,4+ 2\ODD)$ (Cor\ref{C-DO-5-unweighted}) \\ \hline
    % $kn^{1+1/k}$ & $\Ot(n^{1/k})$  & $(2k-5, 4 + 2\ODD)$& Theorem~\ref{T-DO-2k-5-unweighted} \\ \hline
    \Xhline{3\arrayrulewidth}
    $\Ot(n^{2-c})$ & $\Ot(n^{tc})$  & $\delta+\delta/t+2$ \cite{DBLP:journals/corr/abs-2307-11677Bilo23ADO} & $\delta+2\lceil \delta/2t\rceil$  (Th\ref{T-Unweighted-1+1/t}) \\ \hline
    % & $\delta+2\lceil \delta/2t\rceil$ & Theorem~\ref{T-Unweighted-1+1/t} \\ \hline 
    $\Ot(n^{3/2+\epsilon})$ & $\Ot(n^{1-2\epsilon})$ & $1.5\delta+2$ \cite{DBLP:journals/corr/abs-2307-11677Bilo23ADO} & $\delta+2\lceil \delta/4\rceil$ (Th\ref{T-Unweighted-1+1/t}) \\ \hline 
    $\Ot(n^{2-2c)}$ & $\Ot(n^{tc})$  & - & $\delta+2\lceil \delta/t\rceil +2$ (Th\ref{T-Unweighted-5/3}) \\ \hline 
    $\Ot(n^{4/3+2\epsilon})$ & $\Ot(n^{1-3\epsilon})$ & - & $\delta+2\lceil \delta/3\rceil+2$ (Th\ref{T-Unweighted-5/3}) \\ \hline 
    % $\Ot(n^{4/3})$ & $\Ot(n^{2/3})$  & - & $\delta+2\lceil \delta/2\rceil+2$ (Th\ref{T-Unweighted-5/3}) \\ \hline 
    \end{tabular}
    \caption{Distance oracles in unweighted graphs. $x\ODD$ is $x \cdot (d(u,v) \mod 2)$}
    \label{tab:unweighted}
\end{table}

\subsection{Applications for $n$-PSP and $ANSC$}
Recently, Dalirrooyfard, Jin, V. Williams, and Wein~\cite{DBLP:conf/focs/DalirrooyfardJW22} (followed by Chechik, Hoch and Lifshitz~\cite{DBLP:conf/soda/ChechikHL25}) studied two fundamental problems in graphs, the $n$-pairs shortest paths ($n$-PSP) problem and the all-nodes shortest cycles ($ANSC$) problem.
The $n$-PSP problem is defined as follows. Given a set of vertex pairs $(s_i,t_i)$, for $1\leq i \leq O(n)$, compute the distance $d(s_i,t_i)$. 
The $ANSC$ problem is defined as follows. Compute the length of the shortest simple cycle that contains $v$, for every $v\in V$. 

A straightforward approach for solving the $n$-PSP problem is to first construct a distance oracle and then query it for every pair in the input set.  
The total runtime of this approach is the sum of the construction time and $n$ times the query time of the distance oracle.
While this method naturally yields an $n$-PSP algorithm, a distance oracle must satisfy additional requirements that an $n$-PSP algorithm does not.
In particular, a distance oracle must use limited space and support queries between any of the $\Omega(n^2)$ pairs of vertices.

% A straightforward approach for solving the $n$-PSP problem is as follows. First, construct a distance oracle and then query it for every pair from the input set. 
% The runtime of an algorithm that follows this approach is the sum of the construction time and $n$ times the query time of the distance oracle. 
% This approach allows us to create $n$-PSP algorithm; however, a distance oracle has more requirements; it needs to use less space, its queries are unknown in the construction time, and the queries must have worst-case guarantee instead of an average over all the $n$ queries. 

% % However, this is one type of algorithm for the $n$-PSP
% Therefore, an efficient construction time of the distance oracle is a prerequisite for applying this approach to obtain an efficient $n$-PSP algorithm.

% We remark that this is one type of algorithm for the $n$-PSP problem, but a distance oracle has more requirements; it needs to use less space, its queries are unknown in advance, and it should have a worst-case query guarantee (instead of average).
% \avi{Notice that the $n$-PSP can be solved by constructing a distance oracle and querying it $n$ times, but it can also use the vertex pairs in advance and thereby is an easier problem.}

Using our new distance oracles, we obtain improved time/stretch tradeoffs for both the $n$-PSP problem and the $ANSC$-problem, improving some of the results of~\cite{DBLP:conf/focs/DalirrooyfardJW22, DBLP:conf/soda/ChechikHL25}.

In~\cite{DBLP:conf/focs/DalirrooyfardJW22}, a lower bound based on the combinatorial $4k$ clique hypothesis for the $n$-PSP problem is presented. They showed that 
$\Omega(m^{2 - \frac{2}{k+1}} \cdot n^{\frac{1}{k+1} - o(1)})$ time is required to achieve $(1+1/k-\eps)$ stretch. As mentioned in \cite{DBLP:conf/focs/DalirrooyfardJW22} using 
the $(1+1/k)$-distance oracle of~\cite{DBLP:conf/esa/Agarwal14} it is possible to solve $n$-PSP with  $(1+1/k)$-stretch in $O(m^{2-\frac{2}{k+1}} n^{\frac{1}{k+1}})$ time, which almost matches the lower bound of~\cite{DBLP:conf/focs/DalirrooyfardJW22}. 
In this paper, we show that it is possible to significantly improve the running time of~\cite{DBLP:conf/esa/Agarwal14} and to bypass the lower bound of \cite{DBLP:conf/focs/DalirrooyfardJW22} at the cost of a small additive error (of up to $2$).
We obtain an $O(m^{1-\frac{1}{k+1}}n)$ time algorithm that returns $\hat{d}(s_i,t_i) \le d(s_i,t_i) + 2\ceil*{\frac{d(s_i,t_i)}{2k}} \le (1+1/k)d(u,v)+2$, as presented in the following theorem.

\begin{theorem}[\autoref{S-apps-npsp}]\label{T-NPSP-1+1/k-Unweighted}\Copy{T-NPSP-1+1/k-Unweighted}{
    Let $k\ge2$. Given an unweighted graph $G$ and vertex pairs $(s_i,t_i)$ for $1\leq i\leq O(n)$, there is an algorithm for $n-PSP$ that computes an estimation $\hat{d}(s_i,t_i)$ such that $d(s_i,t_i) \le \hat{d}(s_i,t_i) \le d(s_i,t_i) + 2\ceil*{\frac{d(s_i,t_i)}{2k}}$ in $\Ot(m^{1-\frac{1}{k+1}}n)$ time.}
\end{theorem}

In~\cite{DBLP:conf/focs/DalirrooyfardJW22} they presented an $O(m+n^{2/k})$ time  $n$-PSP algorithm that returns $(2k-2)(2k-1)$-stretch. In the following theorem we improve the stretch from $(2k-1)(2k-2)$ to $(2k-1)(2k-3)$ while using the same running time.
% In addition, we present an $n$-PSP algorithm, which improves upon the $(2k-2)(2k-1)$-stretch of~\cite{DBLP:conf/focs/DalirrooyfardJW22} that runs in the same time.
\begin{theorem}[\autoref{S-apps-npsp}]\label{T-NPSP-2k-1-2k-3}\Copy{T-NPSP-2k-1-2k-3}{
    Let $k\ge 3$. Given a weighted graph $G$ and vertex pairs $(s_i,t_i)$ for $1\leq i\leq O(n)$, there is an algorithm for $n-PSP$ that computes an estimation $\hat{d}(s_i,t_i)$ such that $d(s_i,t_i) \le \hat{d}(s_i,t_i) \le (2k-1)(2k-3)$ in $\Ot(m+n^{2/k})$ time.}
\end{theorem}

For the $ANSC$ problem, \cite{DBLP:conf/focs/DalirrooyfardJW22} presented an $O(m^{2-2/k}n^{1/k})$ running time algorithm that for every $u\in V$ returns estimate $\hat{c}(u)$ such that $\hat{c}_u \le SC(u)+2\ceil*{\frac{SC(u)}{2(k-1)}}$, where $SC(u)$ is the length of the shortest simple cycle that contains $u$. In the following theorem, we improve the running time from $O(m^{2-2/k}n^{1/k})$ of \cite{DBLP:conf/focs/DalirrooyfardJW22} to $\Ot(mn^{1-\frac{1}{k}})$.

% we present an $O(mn^{1-1/k})$ time algorithm that returns $\hat{c}_u \le SC(u)+2\ceil*{\frac{SC(u)}{2(k-1)}}$, where $SC(u)$ is the length of the shortest simple cycle that contains $u$ and $\hat{c}_u$ is the estimation, as presented in the following theorem. 
\begin{theorem}[\autoref{S-ANSC}] \label{T-ANSC-1+1/k-approx-unweighted} \Copy{T-ANSC-1+1/k-approx-unweighted}{
    Given an undirected unweighted graph $G$, let $k$ be a positive integer. There is a randomized algorithm for $ANSC$ that computes for every $u\in V$ an estimation $\hat{c}_u$ such that $SC(u) \le \hat{c}_u \le 1+2\ceil*{\frac{SC(u)}{2(k-1)}}$ in $\Ot(mn^{1-\frac{1}{k}})$-time.}
\end{theorem}
% This algorithm improves the $O(m^{2-2/k}n^{1/k})$ running time algorithm of~\cite{DBLP:conf/focs/DalirrooyfardJW22} that achieves the same stretch. 
In addition, we present the first, to the best of our knowledge, stretch $<2$ algorithm for the $ANSC$ problem in \textit{weighted graphs}. Specifically, we present an $O(m^{2-1/k})$ time algorithm for weighted graphs that for every $u\in V$ returns an estimate $\hat{c}$ such that $\hat{c}_u \le 1+\frac{1}{k-1}SC(u)$, as presented in the following theorem.
\begin{theorem}[\autoref{S-ANSC}] \label{T-ANSC-1+1/k-approx-weighted}\Copy{T-ANSC-1+1/k-approx-weighted}{
    Given an undirected weighted graph $G$, let $k$ be a positive integer.
    There is a randomized algorithm for $ANSC$ that computes for every $u\in V$ an estimation $\hat{c}_u$ such that $SC(u) \le \hat{c}_u \le (1+\frac{1}{k-1})SC(u)$ in $\Ot(m^{2-\frac{1}{k}})$-time.
    % Let $G=(V, E)$ be an undirected weighted graph with real non-negative edge weights. Let $|V|=n$, let $|E|=m$, let $k>0$ be an integer and let $\mu=m/n$.
    
    % Let $k$ be a positive integer. There is a randomized algorithm for $ANSC$ that for computes for every $u\in V$ an estimation $\hat{c}_u$ such that $SC(u) \le \hat{c}_u \le (1+\frac{1}{k-1})SC(u)$ in $\Ot(m^{2-\frac{1}{k}})$-time.
}
\end{theorem}
Our results for the $ANSC$ problem are summarized in~\autoref{tab:ansc}.

The rest of this paper is organized as follows. In the next section, we present some necessary preliminaries. 
In~\autoref{S-Over} we present a technical overview of our main techniques and some of our new distance oracles. 
In~\autoref{S-New-Weighted} we present our borderline vertices technique and our main tradeoffs between space, stretch, and query time.
In~\autoref{S-more-3} we present our new middle vertex technique and four
new distance oracles obtained using this technique for unweighted graphs.
In~\autoref{S-at-most-2} we present distance oracles achieving stretch strictly less than 2 in unweighted graphs. In~\autoref{S-at-most-2weighted} we extend the results from~\autoref{S-at-most-2} to weighted graphs with non-negative real edge weights. 
In~\autoref{S-apps} we present our new algorithms for $n$-PSP and $ANSC$ that are obtained with our new distance oracles.
In~\autoref{S-LB-appendix} we prove the lower bound of~\autoref{T-LB-2}.
For reference, in~\autoref{S-Tables} summarizes all our main theorems in tables.
% In~\autoref{S-Tables} we present our $\frac43(2k-1)$-stretch distance oracle that uses $\Ot(m+n^{(1-c)(1+1/k)})$ space.

\section{Preliminaries}\label{S-Prel}
In this section, we present several definitions and existing tools that we use to develop our new techniques and ideas which we then apply in order to obtain new distance oracles. 
Let $G=(V,E)$ be an undirected graph with $n=|V|$ vertices and $m=|E|$ edges. Throughout the paper, we consider both unweighted graphs and weighted graphs with non-negative real edge weights.

Let $u,v\in V$. The distance $d(u,v)$ between $u$ and $v$ is the length of a shortest path between $u$ and $v$. 
Let $P(u,v)$ be the shortest path between $u$ and $v$. 
Let $N(u)$ be the vertices that are neighbors of $u$ and let $deg(u)=|N(u)|$ be the degree of $u$.
Let $X\subseteq V$. Let $N(X)$ be the vertices that are neighbors of some vertex $u\in X$, i.e, $N(X)=\{w\in V | \exists x\in X:(w,x)\in E\}$. 
The distance $d(u,X)$ between $u$ and $X$ is the distance between $u$ and the closest vertex to $u$ from $X$, that is, 
$d(u,X)= \min_{x\in X}(d(u,x))$. Let $p(u, X)=\arg \min_{x\in X}(d(u,x))$ (ties are broken in favor of the vertex with a smaller identifier). 

Next, we define bunches and clusters as in~\cite{DBLP:journals/jacm/ThorupZ05}. 
Let $u\in V$ and let $X,Y\subseteq V$. 
Let $B(u,X,Y)=\{ v\in X \mid d(u,v)<d(u,Y)\} $ be the bunch of $u$ with respect to $X$ and $Y$. 
Let $C(u,Y)=\{ v\in V \mid d(u,v)<d(v,Y)\} $ be the cluster of $u$ with respect to $Y$. 

The starting point in many algorithms and data structures for distance approximation, and in particular in Thorup and Zwick\cite{DBLP:journals/jacm/ThorupZ05}'s distance oracles, is a hierarchy of vertex sets  $A_0, A_1,\ldots, A_k$, where $A_0=V$, $A_k=\emptyset$ and for $0\leq i \leq k-1$, $A_{i+1}\subseteq A_i$ and $|A_i|=n^{1-i/k}$. For every $u\in V$, let $p_i(u)=p(u,A_{i})$ and let $h_i(u)=d(u,A_i)$.
Using this hierarchy, Thorup and Zwick defined 
$k$ bunches for every vertex $u\in V$ as follows: For every $0\leq i \leq k-1$, the bunch $B_i(u)$ is defined as $B_i(u)=B(u,A_i,A_{i+1})$. 
They also defined a cluster for every vertex $w\in A_i\setminus A_{i+1}$ as follows: $C(w)=C(w,A_{i+1})$. 
Throughout the paper when we save $B_i(u)$, for every  $u\in V$  and $0\leq i \leq k-1$ (or $C(u)=C(u,A_{i+1})$, for every $u\in A_i\setminus A_{i+1}$), 
we mean that we can check whether $x\in B_i(u)$ (or $x\in C(u)$) in constant time and if $x\in B_i(u)$ (or $x\in C(u)$) we can retrieve $d(u,x)$ in constant time as well.

The simplest instance of this hierarchy is when $k=2$ and the hierarchy is composed of the sets $A_0,A_1,A_2$, where $A_0=V$, $A_1\subset V$, and $A_2=\emptyset$. 
Throughout this paper, in such a case,
we will omit subscripts and will refer to $A_1$ simply as $A$. Moreover, since  in this case $B_0(u)=B(u,V,A)$ and $B_1(u)
=B(u,A,\emptyset)=A$, for every $u\in V$, we use $B(u)$ to denote $B_0(u)$ and $A$ to denote $B_1(u)$ and since we have only $p_0(u)$ and $p_1(u)$, where $p_0(u)$ is $u$,   we use  $p(u)$ to denote $p_1(u)$ and $h(u)$ to denote $h_1(u)$. 

Thorup and Zwick~\cite{DBLP:journals/jacm/ThorupZ05} construct the sets $A_1,\ldots, A_{k-1}$,  by 
 adding to $A_{i+1}$, where $i\in [0,k-2]$, every vertex of $A_i$, independently at random with probability $p$.
Constructing the sets in such a way allows Thorup and Zwick~\cite{DBLP:journals/jacm/ThorupZ05} to prove the following:

\begin{lemma}[\cite{DBLP:journals/jacm/ThorupZ05}]\label{L-TZ-Size}
Given an integer parameter $k\geq 2$, we can compute in $\Ot(n)$ time sets $A_1,\ldots, A_{k-1}$, such that $|A_i|=O(n^{1-i/k})$ for every $i\in [1,k-1]$ with high probability, and for every $i\in [0,k-1]$ the size of $B_i(u)$ is $O(n^{1/k})$, with high probability (w.h.p).
The cost of computing $B_i(u)$, for every $u\in V$, is $\Ot(mn^{1/k})$  expected time.
\end{lemma}

\begin{lemma}[\cite{DBLP:conf/spaa/ThorupZ01}]\label{L-A-center}
Given a parameter $p$, we can compute a set $A$ of size $\Ot(np)$ in $\Ot(mp^{-1})$ expected time
such that, $|C(w,A)|=O(1/p)$, for every vertex $w\in V \setminus A$, and $|B(v,V,A)|=O(1/p)$ for every $v\in V$. 
\end{lemma}

We remark that both~\autoref{L-TZ-Size} and~\autoref{L-A-center} have also slower deterministic constructions~\cite{RodittyTZ05, DBLP:conf/spaa/ThorupZ01}. 

Next, we present procedure $\Intersection$. 
The input to $\Intersection$ is two vertices $u,w\in V$ and two sets $U,W\subseteq V$, such that, $u\in U$ and $w\in W$. For every $u'\in U$ and every $w'\in W$, we assume that distance estimations $d'(u,u')$ and $d'(w,w')$ are known. The output of $\Intersection$
is $\min_{x\in U\cap W} d'(u,x) + d'(w,x)$, if $U\cap W\neq \emptyset$ and $\infty$, otherwise. (See Algorithm~\ref{Algorithm-Intersection}.)
In most of our uses of $\Intersection$ we have  $d'(u,u')=d(u,u')$ and $d'(w,w')=d(w,w')$. When this is not the case, we explicitly describe  $d'$ and its relation to $d$. 
\begin{algorithm2e}[t] 
\caption{$\Intersection(u, w, U, W
)$}\label{Algorithm-Intersection}
\If {$U\cap W\neq \emptyset$} {
    \Return $\min_{x\in U\cap W} d'(u,x) + d'(x,w)$;
}
\Return $\infty$;
\end{algorithm2e}
The following lemma regarding the running time of $\Intersection$ is straightforward. 
\begin{lemma}
    $\Intersection(u,w,U,W)$ runs in $O(\min(|U|,|W|))$.
\end{lemma}\label{Intersection-Runtime}
\begin{proof}
    Assume, without loss of generality (wlog), that $|U| \ge |W|$.
    By checking for every vertex $y\in W$ in $O(1)$ time whether $y\in U$ we compute $U\cap W$.  
    For each vertex $x\in U\cap W$, we compute in $O(1)$ time the value $d'(u,x) + d'(x,w)$. 
    Thus, the total running time is 
    $\min(|U|,|W|)\cdot O(1)=O(\min(|U|,|W|))$, as required.
\end{proof}

The following property of $\Intersection(u,w,U,W)$ is the main feature of $\Intersection$, and will be used throughout the paper. 
\begin{property}
\label{P-Intersection-With-P-Returns-d(u,v)}
Let $P=P(u,w)$. If $P\cap (U\cap W)\neq \emptyset$, and there exists $x\in P\cap (U\cap W)$ such that $d'(u,x)=d(u,x)$ and $d'(x,v)=d(x,v)$ then $\Intersection(u,w,U,W)$
returns $d(u,w)$. 
\end{property}

Let $\TZQuery$ be the query of the distance oracle presented by Thorup and Zwick~\cite{DBLP:journals/jacm/ThorupZ05}. We let $\MTZQuery(u,v)$ be $\min(\TZQuery(u,v), \TZQuery(v,u))$. Pseudocode for $\TZQuery$ exists in~\autoref{A-ADO-Query}.
\begin{algorithm2e}[t] 
\caption{$\TZQuery(u,v)$} \label{A-ADO-Query}
\For{$i\in [k]$} {
    \lIf{$p_{i-1}(u)\in B_i(v)$}{\Return $h_{i-1}(u)+d(p_{i-1}(u),v)$}
    $(u,v) \gets (v,u)$
}
\end{algorithm2e}

The following lemma appears explicitly in~\cite{DBLP:journals/jacm/ThorupZ05} regarding the correctness and running time of $\DistanceOracle$.
\begin{lemma}[\cite{DBLP:journals/jacm/ThorupZ05}]\label{L-THZ05-Q-Correctness}
    The running time of $\DistanceOracle$ is $O(k)$ and it holds that $$\DistanceOracle(u,v) \le (2k-1)d(u,v)$$
\end{lemma}

The next lemma follows from~\cite{DBLP:journals/jacm/ThorupZ05}, and we prove it here for completeness.
\begin{lemma}\label{L-ADO.Query-Correctness}
    Let $\hat{d}(u,v)$ be $\MTZQuery(u,v)$.
    For every integer $1 \le i \le k$, one of the following holds.
    \begin{itemize}
        \item $\min(h_i(u),h_i(v)) \le \min(h_{i-1}(u),h_{i-1}(v)) + d(u,v)$
        \item $\hat{d}(u,v) \le 2\min(h_{i-1}(u),h_{i-1}(v)) + d(u,v)$
    \end{itemize}
\end{lemma}
\begin{proof}
    Wlog, assume that $h_{i-1}(u) \le h_{i-1}(v)$. 
    We divide the proof into two cases.
    The case that $p_{i-1}(u)\in B_i(v)$ and the case that $p_{i-1}(u)\notin B_i(v)$.
    Consider the case that $p_{i-1}(u)\in B_i(v)$. In this case 
    the value of $d(p_{i-1}(u),v)$ is saved in the distance oracle. From the triangle inequality it follows that $d(p_{i-1}(u),v)\le h_{i-1}(u) + d(u,v)$. Therefore, we have that $\hat{d}(u,v) \le h_{i-1}(u)+d(p_{i-1}(u),v) \le 2h_{i-1}(u)+d(u,v)$.
    % , where the last inequality follows from the triangle inequality. 
    Since $h_{i-1}(u) \le h_{i-1}(v)$ we get that $\hat{d}(u,v) \le 2\min(h_{i-1}(u),h_{i-1}(v)) + d(u,v)$,
    as required.
    
    Consider now the case that $p_{i-1}(u)\notin B_i(v)$. From the definition of $B_i(v)$ it follows that $h_i(v) \le d(v,p_{i-1}(u)) \le d(u,v) + h_{i-1}(u)$. Since $h_{i-1}(u) \le h_{i-1}(v)$, we get that $h_i(v) \le \min(h_{i-1}(u),h_{i-1}(v)) + d(u,v)$, and thus $\min(h_i(u),h_i(v)) \le \min(h_{i-1}(u),h_{i-1}(v)) + d(u,v)$, as required.
\end{proof}
The following property follows by applying~\autoref{L-ADO.Query-Correctness} inductively.
\begin{property}\label{P-ADO.Query-Correctness}
    Let $1 \le i \le k$ be an integer, and let $1\le j<i$.
    Then either:
    \begin{itemize}
        \item $\min(h_i(u),h_i(v)) \le \min(h_{i-1}(u),h_{i-1}(v)) + d(u,v)$
        \item $\hat{d}(u,v) \le 2\min(h_{j}(u),h_{j}(v)) + (i-j)d(u,v)$
    \end{itemize}
\end{property}
We remark that all our distance oracles return an estimation that is the length of a path in $G$ and therefore $d(u,v)\leq \hat{d}(u,v)$.

\section{Technical overview}\label{S-Over}
In the  classic distance oracle of Thorup and Zwick~\cite{DBLP:journals/jacm/ThorupZ05}, it is possible to bound 
$h_i(u)$ with $d(u,v)+h_{i-1}(v)$, for every $1\leq i\leq k-1$.
% The key to improving the $(2k-1)$-stretch of the distance oracle is in improving the bound on $h_i(u)$.
% In the past (see for example, \cite{DBLP:conf/infocom/AgarwalGH11, PatrascuR14,DBLP:conf/wdag/AbrahamG11,DBLP:conf/soda/AgarwalG13,DBLP:conf/esa/Agarwal14}), they managed to improve the bound $h_1(u)$ from $d(u,v)$ to $d(u,v)/2$. 
The key to improving the $(2k - 1)$-stretch of the distance oracle lies in tightening the bound on $h_i(u)$.
Previous work (see, for example, \cite{DBLP:conf/infocom/AgarwalGH11, PatrascuR14, DBLP:conf/wdag/AbrahamG11, DBLP:conf/soda/AgarwalG13, DBLP:conf/esa/Agarwal14}) improved the bound on $h_1(u)$ from $d(u,v)$ to $d(u,v)/2$ by exploiting the following structural property of  $B_0(\cdot)$: if  $P(u,v)\not\subseteq B_0(u) \cup B_0(v)$,  then $h_1(u) + h_1(v) \le d(u,v) + 1\ODD$.

% In unweighted graphs, the key idea for improving the bound on $h_1(u)$ from $d(u,v)$ to $d(u,v)/2$ is to show that if $P(u,v)\not\subseteq B_0(u) \cap B_0(v)$ then $h_1(u) + h_1(v) \le d(u,v) + 1$. To extend this idea to weighted graphs, \cite{DBLP:conf/infocom/AgarwalGH11} defined $B_0^*(u)=B_0(u)\cup N(B_0(u))$ and showed that in weighted graphs if $P(u,v)\not\subseteq B_0^*(u) \cap B_0^*(v)$ then $h_1(u) + h_1(v) \le d(u,v)$.

In this paper, we develop two new techniques that exploit structural properties of $C(\cdot)$ to construct a distance oracle with improved stretch guarantees. These techniques enable us to bound not only $h_1(u)$ by $d(u,v)/2$, but also $h_2(u)$ by $d(u,v)$, thereby improving overall stretch.
We begin by presenting the borderline vertices technique, which applies to both weighted and unweighted graphs.

% In this paper, we develop two new techniques that 
% exploit   structural properties  of  $C(\cdot)$ and  
% allow us to obtain a new distance oracle with improved stretch by bounding not only $h_1(u)$ with $d(u,v)/2$ but also bounding $h_2(u)$ with $d(u,v)$. 
% First, we overview the borderline vertices technique that works for both weighted and unweighted graphs. 

\subsection{The borderline vertices technique. [\autoref{S-borderline}]}
% \textbf{The borderline vertices technique. [\autoref{S-borderline}]}

Let $C^*(u)=C(u)\cup N(C(u))$ be an augmented cluster.
Roughly speaking, using the borderline vertices technique we show that if $P(u,v) \not\subseteq C^*(u,A_c) \cap C^*(v,A_c)$ then $\max(h_{c+1}(u), h_{c+1}(v))\le d(u,v)$. In particular, by setting $c=1$ we get an improvement over the result of \cite{DBLP:conf/infocom/AgarwalGH11}, showing not only that $h_1(u)\le d(u,v)/2$ but also that $h_2(u) \le d(u,v)$.

Let $P=P(u,v)$ be a shortest path between $u$ and $v$. 
% For the sake of simplicity, we discuss a simple special case of our borderline vertices technique. 
The borderline vertex $\tau_u(P)$ of $C^*(u)$ in $P$ is the farthest vertex from $u$ in $C^*(u)\cap P$. Similarly, the borderline vertex 
$\tau_v(P)$ of $C^*(u)$ in $P$ is the farthest vertex from $v$ in $C^*(v)\cap P$. 
Obviously, a distance oracle cannot store $\tau_u(P)$ for every shortest path $P$, as this would require $\Theta(n^2)$ space.
We overcome this limitation by using $\Ot(\mu \cdot |C(u)|)$ query time to iterate over all vertices in the augmented cluster $C^*(u)$, and in particular $\tau_u=\tau_u(P)$.

To obtain our improved bound on $h_2(v)$, we analyze two different cases regarding $p_1(\tau_u)$. The case that $p_1(\tau_u) \in B_1(v)$ and the case that $p_1(\tau_u) \notin B_1(v)$.
If $p_1(\tau_u) \in B_1(v)$, then the query algorithm can return as an estimation $\hat{d}(u,v)=u\rightarrow \tau_u\rightarrow p_1(\tau_u)\rightarrow v$, where $x\rightarrow y \rightarrow z\rightarrow w$ is $d(x,y)+d(y,z) + d(z,w)$.
In this case, since $\tau_u\notin C(u,A_1)$ (as it is the farthest vertex from $u$ in $C^*(u)$) we get that $h_1(\tau_u) \le d(u,\tau_u) \le d(u,v)$, and therefore $\hat{d}(u,v) \le 3d(u,v)$.

Otherwise, if $p_1(\tau_u) \notin B_1(v)$ then $h_2(v) \le d(v,p_1(\tau_u)) \le d(v,\tau_u) + h_1(\tau_u)$. Since $h_1(\tau_u) \le d(u,\tau_u)$ we get that $h_2(v) \le d(v,\tau_u) + d(u,\tau_u) = d(u,v)$, as wanted. (See~\autoref{fig:borderline}).
Using symmetric arguments for $\tau_v$, we can also get that $h_2(u) \le d(u,v)$. (See \autoref{S-borderline} for more details.)

In unweighted graphs, we can avoid the $\mu$ factor from the query time at the cost of bounding $h_2(u)$ and $h_2(v)$ by $d(u,v)+2$ rather than $d(u,v)$. 
The difference is that we let $\tau'_u(P)$ be defined as the farthest vertex from $u$ in $C(u) \cap P$ (unlike $\tau_u(P)$ which considers $C^*(u)\cap P$).
% Using similar arguments to the ones described above.
% note that $h_1(\tau'_u) \le d(u,\tau'_u)+2$ and hence $h_2(v) \le d(u,v)+2$. 
(See \autoref{L-bound-h_{c+1}(v)-tau_u-+2}.) 
Next, we provide an overview of our main three-way tradeoff that uses the borderline vertices technique.

\subsection{$(2k-1-4c)$-stretch distance oracle. [\autoref{S-2k-1-4c-Weighted} -~\autoref{T-DO-2k-1-4c-Weighted}]}

% \textbf{$(2k-1-4c)$-stretch distance oracle. [\autoref{S-2k-1-4c-Weighted} -~\autoref{T-DO-2k-1-4c-Weighted}]}
Let $0<c<k/2-1$ be an integer, let $V=A_0,A_c,A_{c+1},\dots,A_k=\emptyset$, such that $|A_i|=n^{1-i/k}$.
The storage of the distance oracle saves: the graph $G$, $d(s_1,s_2)$, for every $s_1\in A_{c+1}, s_2\in A_{k-c-2}$ and $B(u)=\cup_{i\in \{0\cup[c,k]\}} B_i(u)$ for every $u\in V$. Where $B_0(u)=B(u,V,A_c)$.

The query works as follows.
First, we set $\hat{d}(u,v)$ to $\Intersection(u,v,C^*(u, A_{c}), C^*(v, A_{c}))$. Then, we iterate over every $u'\in C^*(u, A_{c})$ (to iterate over $\tau_u$), and update $\hat{d}(u,v)$ be $\min(\hat{d}(u,v), d(u,u')+\MTZQuery(u',v))$. Similarly, for every $v'\in C^*(v, A_{c})$ we update $\hat{d}(u,v)$ to $\min(\hat{d}(u,v), d(v,v')+\MTZQuery(v',u))$.
Finally, the algorithm returns $\min(\hat{d}(u,v), u\rightarrow p_{c+1}(u)\rightarrow p_{k-c-2}(v)\rightarrow v, u\rightarrow p_{k-c-2}(u)\rightarrow p_{c+1}(v)\rightarrow v)$.

\begin{figure}
    \centering
    \tikzset{every picture/.style={line width=0.75pt}} %set default line width to 0.75pt        
    \input{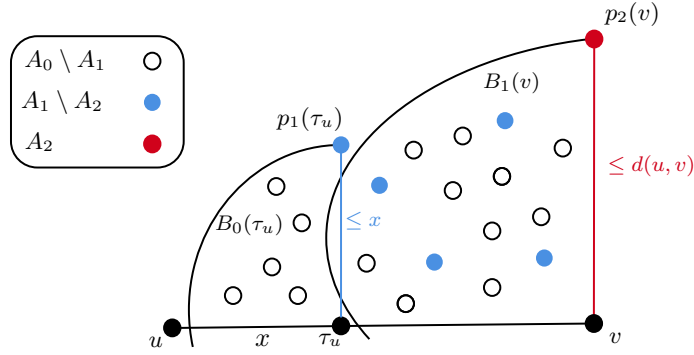}

    \caption{$p_1(\tau_u)\notin B_1(v)$.}
    \label{fig:borderline}
\end{figure}

It is straightforward to see that the space is $O(m+n^{1+1/k})$ and that the query takes $O(\mu |C(u,A_c)|)=O(\mu n^{c/k})$ time (see~\autoref{L-Space-2k-5-Weighted} and~\autoref{L-2k-4c-1-weighted-Q-Time}).
The main technical contribution is in proving that $\hat{d}(u,v) \le (2k - 4c - 1)\cdot d(u,v)$ (see~\autoref{L-Approximation-2k-4c-1-Weighted}).
To prove it, we use the borderline vertices technique.
Since that $\tau_u \in C^*(u, A_c)$, when the algorithm iterates over $C^*(u, A_c)$, there is an iteration in which $u' = \tau_u$. In this iteration we guarantee that $\hat{d}(u,v) \le d(u, \tau_u) + \MTZQuery(\tau_u, v)$.
Since $\tau_u$ is a borderline vertex, we either have that $\hat{d}(u,v) \le u \rightarrow \tau_u \rightarrow p_c(\tau_u) \rightarrow v \le 3d(u,v)$, and the estimation guarantee holds, or that $h_{c+1}(v) \le d(u,v)$.
Similarly, since $\tau_v \in C^*(v, A_c)$ and we iterate over $C^*(u, A_c)$, there is an iteration in which $v' = \tau_v$. In this iteration, either a short path is found ($v \rightarrow \tau_v \rightarrow p_c(\tau_u) \rightarrow u$), or we have that $h_{c+1}(u) \le d(u,v)$.

From the properties of $\MTZQuery$, for every $c+1 < i < k$, either a short path is found or $h_i(u) \le d(u,v) + h_{i-1}(v)$. Since $h_{c+1}(u) \le d(u,v)$ and $h_{c+1}(v) \le d(u,v)$, it follows that
\[
h_{k-c-2}(u) \le \max(h_{c+1}(u), h_{c+1}(v)) + (k - 2c - 3)d(u,v) \le (k - 2c - 2)d(u,v).
\]
In the query procedure, we have that $\hat{d}(u,v) \le u \rightarrow p_{k - c - 2}(u) \rightarrow p_{c + 1}(v) \rightarrow v \le 2h_{k-c-2}(u) + d(u,v) + 2h_{c+1}(v) $. Thus:
\[
\hat{d}(u,v) \le 2(k - 2c - 2)d(u,v) + d(u,v) + 2d(u,v) = (2k - 4c - 1)d(u,v),
\]
as required. (See~\autoref{L-Approximation-2k-4c-1-Weighted} for more details.)

% From the properties of $\MTZQuery$, we have that $h_i(u) \le d(u,v) + h_{i-1}(v)$ for every $c+1<i<k$ or a short path is found.
% Using the fact that $h_{c+1}(u) \le d(u,v)$ and $h_{c+1}(v) \le d(u,v)$, we get that $h_{k-c-2}(u) \le \max(h_{c+1}(u), h_{c+1}(v))+(k-2c-3)d(u,v) \le (k-2c-2)d(u,v)$.
% In our query we have that $\hat{d}(u,v) \le u\rightarrow p_{k-c-2}(u) \rightarrow p_{c+1}(v) \rightarrow v$. From the triangle inequality we have that $u\rightarrow p_{k-c-2}(u) \rightarrow p_{c+1}(v) \rightarrow v \le 2h_{k-c-2}(u) + 2h_{c+1}(v)+d(u,v)$. Since 
% $h_{k-c-2}(u) \le (k-2c-2)d(u,v)$ and $h_{c+1}(v) \le d(u,v)$, we get that 
% $\hat{d}(u,v) \le 2(k-2c-2)d(u,v)+2d(u,v)+d(u,v)=(2k-4c-1)d(u,v)$, as required.

\subsection{The middle vertex technique. [\autoref{S-middle}]}
% \textbf{The middle vertex technique. [\autoref{S-middle}]}

For unweighted graphs, we develop the middle vertex technique that allows us to bound either $h_2(u)$ or $h_2(v)$ by $d(u,v)/2 + 1_{\text{ODD}}$, improving the previous $3d(u,v)/2 + 1_{\text{ODD}}$ bound from~\cite{DBLP:conf/focs/DalirrooyfardJW22}. Roughly speaking, using the middle vertex technique technique we show that if $P(u,v) \not\subseteq C(u,A_1) \cap C(v,A_1)$ then $\min(h_{2}(u), h_{2}(v))\le d(u,v)+1\ODD$. 
The middle vertex technique requires $O(n^{1/k})$ query time, matching the query time of the unweighted borderline technique. However, by using the middle vertex technique we reduce the additive error from  $2$ to $1\ODD$ in the bound of $\min(h_{2}(u), h_{2}(v))$.

For the sake of simplicity, let $A_0,A_1,A_2,A_3$ be a vertex hierarchy, where $A_0=V$, $A_3=\emptyset$ and let $d(u,v)=\delta$ be even.
Let $\tau$ be the middle vertex on $P(u,v)$. That is, $d(u,\tau)=d(\tau,v)=\delta/2$ ($\delta/2$ is an integer since $\delta$ is even). 
The middle vertex $\tau$ has the following useful property. If $h_1(\tau)>\delta/2$ then $u,v\in B_0(\tau)$. 
The problem is that $\tau$ is unknown. However, the case that $u,v\in B_0(\tau)$ is equivalent to the case that $\tau\in C(u)\cap C(v)$. 
We use the $O(n^{1/k})$ query time to compute for every $w\in C(u) \cap C(v)$ the value of $d(u,w)+d(w,v)$. Therefore, even without knowing $\tau$, in this case we have that $\hat{d}(u,v)=\delta$. (See~\autoref{F-mid-vertex}(a).)

If $h_1(\tau)\leq \delta/2$ then $d(u,p_1(\tau))\leq d(u,\tau) + h_1(\tau) \le \delta$ and $d(v,p_1(\tau))\leq \delta$.  
If $p_1(\tau)\in B_1(u)\cap B_1(v)$, we exploit, again, the $O(n^{1/k})$ query time to compute for every $w\in B_1(u)\cap B_1(v)$ the value $d(u,w)+d(w,v)$. Therefore, even without knowing $p_1(\tau)$ the query algorithm returns $\hat{d}(u,v)\le d(u,p_1(\tau))+d(v,p_1(\tau)) \le 2\delta$. (See~\autoref{F-mid-vertex}(b).) 
If $p_1(\tau)\notin B_1(u)\cap B_1(v)$ then without loss of generality $p_1(\tau)\notin B_1(u)$ and therefore $h_2(u)\leq d(u,p_1(\tau))\leq \delta$. (See~\autoref{F-mid-vertex}(c).)

Our distance oracles for unweighted graphs use the middle vertex technique (see~\autoref{S-4-2} and~\autoref{S-DO-4-Unweighted}). When combined with the unweighted borderline vertices technique, this enables the construction of two additional oracles (see~\autoref{S-3,2+2ODD-4} and~\autoref{S-2k-5+6-unweighted}). 
% Next, we provide an overview of our results in the regime of stretch $<2$.
\begin{figure}[t]
\begin{center}
\scalebox{.75}
{
\input{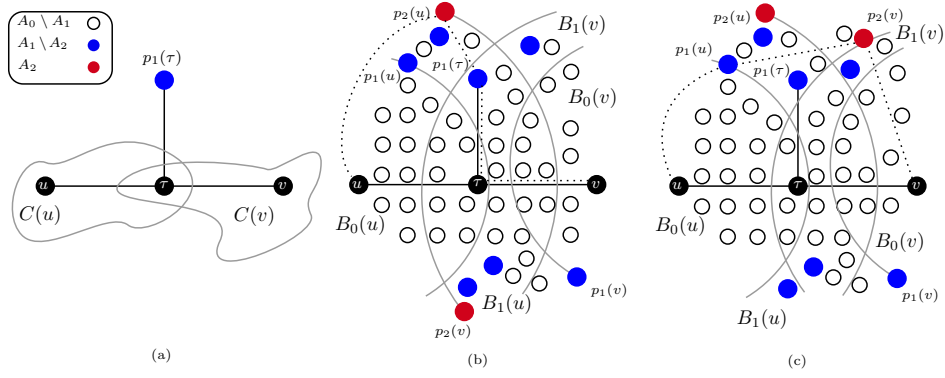}
}
\end{center}
\caption{(a) $\tau\in C(u)\cap C(v)$. (b) $p_1(\tau) \in B_1(u)\cap B_1(v)$. (c) $p_1(\tau) \notin B_1(u)\cap B_1(v)$. In the figure $p_1(\tau) \notin  B_1(v)$ }\label{F-mid-vertex}
\end{figure}

\subsection{Distance oracles with stretch $<2$  [\autoref{S-at-most-2} and \autoref{S-at-most-2weighted}]}
% \textbf{Better than $2$-stretch distance oracles [\autoref{S-at-most-2} and \autoref{S-at-most-2weighted}]}

Let $A\subseteq V$ be a set of size $O(n^{1-c})$, let $u,v\in V$, let $\delta=d(u,v)$ and let $P=P(u,v)$. Let $t \geq 0$ be an integer. Let  $S_t(u)=\bigcup_{w\in S_{t-1}(u)}{B(w)}$, where $S_0(u)=\{u\}$. Let $d_t$ be a distance function induced by $S_t$ (for the exact definition of $d_t$ see~\autoref{S-5-1}). 
Let $u=u_0$ (resp. $v=v_0$). For every $1\leq i \leq t$, let $u_i \in B(u_{i-1})$ (resp. $v_i\in B(v_{i-1})$) be the farthest vertex in $P$ from $u$ (resp. $v$). 
(See Figure~\ref{F-S} for an illustration.)
Let $P_t(u)=P(u,u_t)$. By the definition of $S_{t}(u)$ we have that $P_t(u)\subseteq S_{t}(u)$. 
We prove the following two properties: 
\begin{itemize}
    \item If $P_t(u) \cap P_t(v) \neq  \emptyset$ then $d_t(u,w)+d_t(v,w)=\delta$ for some $w\in P_t(u) \cap P_t(v)$. (See Figure~\ref{F-S}(b).)
    \item If $P_t(u) \cap P_t(v) = \emptyset$ then $\sum_{i=0}^{t-1}h(u_i) + h(v_i) \le \delta+2t-1$. (See Figure~\ref{F-S}(a).)
\end{itemize}

Next, we overview the distance oracle presented in ~\autoref{S-5-t} that uses these properties. 
The distance oracles stores the graph $G$ and $d(x,y)$, for every $\langle x,y \rangle \in A \times V$, at the cost of $O(|V|\cdot |A|)=O(m+n^{2-c})$ space. 
In the query of the distance oracle, we compute $S_t(u)$ and $S_t(v)$. 
To address the case that $P_t(u) \cap P_t(v) \neq \emptyset$
we compute for every $w\in S_t(u)\cap S_t(v)$ in $O(\min (|S_t(u)|,|S_t(v)|))$ time
the value $d_t(u,w)+d_t(v,w)$. In such a case, from the first property we have $d_t(u,w)+d_t(v,w)=\delta$, for some $w$. 

To address the case that $P_t(u) \cap P_t(v) = \emptyset$ we compute $d(u,p(w))+d(p(w),v)$ for every $w\in S_t(u)\cup S_t(v)$. 
Let $w\in \bigcup_{i=0}^{t-1}\{u_i,v_i\}$ be the vertex with minimal $h(w)$ value.
Using the bound $\sum_{i=0}^{t-1}h(u_i) + h(v_i) \le d(u,v)+2t-1$, we can show that $h(w) \le \lfloor (d(u,v) + 2t - 1)/2t \rfloor = \lceil d(u,v)/2t \rceil$ and get that $d(u,p(w))+d(p(w),v)\leq \delta+ 2\lceil \delta/2t \rceil$. 
% To enable the computation of $d(u,p(w))+d(p(w),v)$ in the query we store $d(x,y)$, for every $\langle x,y \rangle \in A \times V$, at the cost of $O(|V|\cdot |A|)=O(n^{2-c})$ space. 

In~\autoref{S-5-3} we reduce the space from $O(n^{2-c})$ to $O(n^{2-2c})$ by  saving $d(x,y)$, for every  $\langle x,y \rangle \in A \times A$ instead of saving $d(x,y)$, for every $\langle x,y \rangle \in A \times V$.

This information prevents us from using paths of the form $u\rightsquigarrow p(w)\rightsquigarrow v$.
We can still use paths of the form $u\rightsquigarrow u_i \rightsquigarrow p(u_i) \rightsquigarrow p(v_j) \rightsquigarrow v_j \rightsquigarrow v$.  However, since we do not know the vertices $u_i$ and $v_j$ we need to consider all vertex pairs in $S_{t-1}(u)\times S_{t-1}(v)$. This increases the query time from $O(n^{tc})$ to $O(n^{2(t-1)c})$. 
To avoid $O(n^{2(t-1)c})$ query time, we show that it suffices to consider only 
vertex pairs in $S_{i}(u)\times S_{t-1-i}(v)$, where $0\leq i\leq t-1$ in $O(tn^{(t-1)c})$-time, and still get the same approximation without iterating over all vertex pairs in $S_{t-1}(u) \times S_{t-1}(v)$.

We define the set $Q=\{\langle u_i,v_{t-i-1}\rangle \mid 0\leq i\leq t-1\}$. 
Let $q=\langle q_u,q_v \rangle\in Q$ be the vertex pair with minimal $h(q_u)+h(q_v)$ value.
We prove that $h(q_u)+h(q_v) \le \lfloor (d(u,v) + 2t - 1)/t \rfloor = \lceil d(u,v)/t \rceil + 1$.
While iterating over the pairs of $S_{i}(u)\times S_{t-1-i}(v)$, we encounter $q$, and get that $\hat{d}(u,v) \le \delta + 2\lceil \delta/t \rceil + 2$. 
Thus, with $O(n^{2-2c})$ space and $O(n^{tc})$ query time we get an estimation $\hat{d}(u,v) \le \delta + 2\lceil \delta/t \rceil + 2$. 
By setting $c=1/3-\eps$ and $t=2,3$ we get two new distance oracles, one with $O(n^{4/3+2\eps})$ space and $O(n^{2/3-2\eps})$ query, and estimation $\hat{d}(u,v) \le \delta + 2\lceil \delta/2 \rceil + 2\leq 2\delta+4$, and another  with $O(n^{4/3+2\eps})$ space and $O(n^{1-3\eps})$ query, and estimation $\hat{d}(u,v) \le \delta + 2\lceil \delta/3 \rceil + 2\leq \frac{5}{3}\delta+4$.

We also consider weighted graphs with non-negative real edge weights. 
Agarwal, Godfrey and Har{-}Peled~\cite{DBLP:journals/corr/abs-1201-2703} defined $B^*(u)$ to be $B(u)\cup N(B(u))$. We revise the definition of $S_t$ to be $\bigcup_{w\in S_{t-1}(u)}{B^*(u)}$. This new definition allows us to extend the results from~\autoref{S-at-most-2} to weighted graphs. The distance oracles obtained using this approach are presented in~\autoref{S-at-most-2weighted}.

We remark that the usage of the sets $S_t(u)$ and $S_t(v)$ in the query algorithm is similar to the usage of the graph $H$ in the work of~\cite{bilo2023improved}. However, in our analysis of the unweighted case, we achieve a slightly better distance approximation. In addition, we introduce a new approach to obtain a distance estimation using two vertices from $A$ that are close to $P(u,v)$ rather than a single vertex from $A$, as in~\cite{bilo2023improved}.
\begin{figure}[t]
\centering
\makebox[\textwidth][c]{\scalebox{.75}{\input{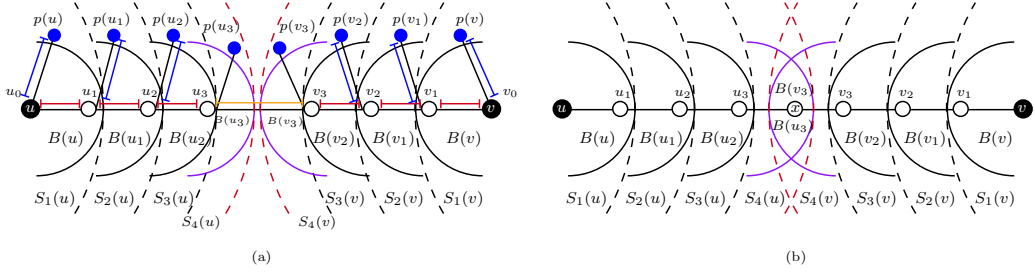}}}
\caption{(a)  $P_4(u)\cap P_4(v) = \emptyset$. The blue interval between $u_i$ ($v_i$) and $p(u_i)$ ($p(v_i)$), where $i\in \{0,1,2\}$, is bounded by the red interval between $u_i$ ($v_i$) and $u_{i+1}$ ($v_{i+1}$) plus 1. The orange interval plus 1 bounds $h(u_3)+h(v_3)$. (b) $P_4(u)\cap P_4(v)= \{ x\}$} \label{F-S}
\end{figure}

% \cite{DBLP:conf/stoc/BiloCCC0KS23} studied weighted graphs with positive edge weights in the range $[1,W]$. We generalize  
% our results to weighted graphs with non-negative real edge weights and remove the additive term $2W$ from the distance approximation at the cost of saving the graph $G$.

\section{A new distance oracles tradeoff in sparse weighted graphs}\label{S-New-Weighted}

In this section we present 
a $(2k-1-4c)$-stretch distance oracle with $\mu n^{\frac{c}{k}}$ query time that uses $O(n^{1+1/k})$ space. 
We start by describing a new technique, called the \textit{borderline vertices technique}, that might be of independent interest. 
We then use the borderline vertices technique to obtain our main stretch/space/query time tradeoff. 
% Finally, we present several interesting implications of our new tradeoffs that improve upon previous works.

\subsection{New technique:  borderline vertices}\label{S-borderline}
Let $G=(V, E)$ be a weighted undirected graph with real non-negative edges.
Let $\mu = m/n$, be the average degree of the graph.
Agarwal, Godfrey, and Har-Peled~\cite{DBLP:journals/corr/abs-1201-2703} showed that in the case of distance oracles in weighted graphs that use $\Omega(m)$ space, it suffices to consider graphs with maximum degree $\mu$. 
Thus, for the rest of this section, we assume that the maximum degree of $G$ is at most $\mu$. 
Agarwal, Godfrey, and Har-Peled~\cite{DBLP:journals/corr/abs-1201-2703} defined augmented clusters as $B^*(u) = B(u) \cup N(B(u))$, following their work we define augmented clusters.
Let $C^*(u,S)=C(u,S)\cup N(C(u,S))$. We bound the size of augmented clusters as follows.
\begin{lemma}\label{L-C^*-size}
    $|C^*(u)| \leq (1+\mu)|C(u)|$
\end{lemma}
\begin{proof}
        Since the maximum degree of $G$ is bounded by $\mu$, we know that $|N(C(u))| \le \mu \cdot |C(u)|$. From the definition of $C^*(u)$, it follows that $|C^*(u)| = | C(u) \cup N(C(u))| \le |C(u)| + |N(C(u))|$. Since $|N(C(u))|\le \mu |C(u)|$, we get that $|C^*(u)| \le |C(u)| + |N(C(u))| \le |C(u)| + \mu |C(u)| = (1+\mu)|C(u)|$, as required.
\end{proof}
Next, we use the definition of $C^*(\cdot)$ to define the borderline vertices. 
Let $P=P(u,v)$ be a shortest path between $u$ and $v$. Let $c\ge1$ be an integer. 
The borderline vertex $\tau_u(P)$ of $C^*(u,A_c)$ in $P$ is the farthest vertex from $u$ in $C^*(u,A_c)\cap P$, similarly the borderline vertex 
$\tau_v(P)$ of $C^*(u,A_c)$ in $P$ is the farthest vertex from $v$ in $C^*(v,A_c)\cap P$.\footnote{Throughout the paper when  $P$ is clear from the context we use $\tau_u$ and $\tau_v$ instead of $\tau_u(P)$ and $\tau_v(P)$.}
% Instead of a single middle vertex, we now have two middle vertices, $\tau_u(P)$ and $\tau_v(P)$, where $\tau_u(P)$ ($\tau_v(P)$) is the farthest vertex from $u$ ($v$) in $C^*(u, A_c)\cap P$ ($C^*(v, A_c)\cap P$).
In the following lemma, we prove a bound on $h_c(\tau_u)$.
\begin{lemma}\label{L-bound-C^*-h}
    $h_c(\tau_u) \le d(u, \tau_u)$.
\end{lemma}
\begin{proof}
    Assume towards a contradiction, that $d(u, \tau_u) < h_c(\tau_u)$ thus  $\tau_u\in C(u)$. The  vertex after $\tau_u$ in $P(u,v)$ is in $N(\tau_u)$ and thus in $C^*(u)$, contradicting the fact that $\tau_u$ is the farthest from $u$ in $P$ which is also in $C^*(u)$.
\end{proof}

In the case that $p_c(\tau_u) \notin B_c(v)$ we can also a prove a bound on $h_{c+1}(v)$.
\begin{lemma}\label{L-bound-h_{c+1}(v)-tau_u}
    If $p_c(\tau_u) \notin B_c(v)$ then $h_{c+1}(v) \le d(u,v)$.
\end{lemma}
\begin{proof}
    From the definition of $B_c(v)$ it follows that since $p_c(\tau_u) \notin B_c(v)$ then  $h_{c+1}(v) \le d(v,p_c(\tau_u))$.
    From the triangle inequality, it follows that 
    $d(v,p_c(\tau_u)) \le d(v,\tau_u) + h_c(\tau_u)$.
    From~\autoref{L-bound-C^*-h} we know that $h_c(\tau_u) \le d(u,\tau_u)$.
    Thus, $h_{c+1}(v) \le d(v,\tau_u) + h_c(\tau_u) \le d(v,\tau_u) + d(\tau_u,u) = d(u,v)$, where the last inequality follows from the fact that $\tau_u\in P(u,v)$, as required. (See~\autoref{fig:borderline}).
\end{proof}

\subsection{$(2k-1-4c)$-stretch distance oracle with $\mu n^{\frac{c}{k}}$ query time.} \label{S-2k-1-4c-Weighted}
Using the borderline vertices technique, we prove the following theorem:
\Reminder{T-DO-2k-1-4c-Weighted}
\subsubsection{Storage.}
We construct a hierarchy $A_0,A_c,\dots,A_k$, where $A_0=V$, $A_{k}=\emptyset$.
The set $A_c$ is constructed using~\autoref{L-A-center}, with parameter $n^{-c/k}$. 
For every $c < i < k$, the set $A_i$ is constructed using~\autoref{L-TZ-Size}. 
Our distance oracle stores $B_i(u)$ and $p_i(u)$, for every $u\in V$ and $c \le i \leq k-1$. We also save $d(u,v)$ for every $\pair{u,v} \in A_{k-c-2} \times A_{c+1}$. 
In addition, we save the original graph $G$.

\begin{lemma}\label{L-Space-2k-5-Weighted}
    The distance oracle is constructed in $\Ot(mn^{\frac{c+2}{k}})$-time and uses $O(m+n^{1+1/k})$-space.
\end{lemma}
\begin{proof}
    From~\autoref{L-TZ-Size} it follows that  $|B_i(u)|=O(n^{1/k})$, for every $c < i< k-1$. From~\autoref{L-A-center} it follows that $|A_c|O(n^{1-c/k})$, thus $|A_{k-1}|=|A_c|\cdot \left(n^{-1/k}\right)^{k-c-1}=O(n^{1-c/k}\cdot n^{-(k-c-1)/k})=O(n^{1/k})$.
    Since $A_{k}=\emptyset$, for every $u\in V$, we have that $B_{k-1}(u)=A_{k-1}$ and $|B_{k-1}(u)|=\Ot(n^{1/k})$.
    The cost of saving $p_i(u)$, for every $u\in V$, and $c\leq i \leq k-1$, is $O(n)$. 
    The cost of saving $d(u,v)$ for every $\pair{u,v}\in A_{k-c-2}\times A_{c+1}$ is $|A_{k-c-2}| \cdot |A_{c+1}|=\Ot(n^{(c+2)/k})\cdot \Ot(n^{1-(c+1)/k})=O(n^{1+1/k})$. Saving the graph $G$ takes $O(m)$ space.
    Therefore, we conclude that the total space is $O(m+n^{1+1/k})$. 
    
    From~\autoref{L-A-center} it follows that computing $A_1$ takes $\Ot(mn^{c/k})$ time. 
    From~\autoref{L-TZ-Size} it follows that computing $A_i$ and $B_i(u)$ for every $c < i\leq k-1$ takes $\Ot(mn^{1/k})$ time. 
    Computing $p_i(u)$, for every $u\in V$, and  $0\leq i \leq k-1$, takes $O(km)$ time.
    Computing $d(u,v)$ for every $\pair{u,v}\in A_{k-c-2}\times A_{c+1}$ takes $O(m\cdot |A_{k-c-2}|)=O(n^{\frac{c+2}{k}}(m+n\log{n}))$ time. 
    Therefore, we conclude that the total construction time is $O(n^{\frac{c+2}{k}}(m+n\log{n}))$. 
\end{proof}

\subsubsection{Query algorithm.}
The input to the query algorithm is two vertices $u,v\in V$.
The query output is an estimation $\hat{d}(u,v)$. The query works as follows. 
First,  $\hat{d}(u,v)$ is set to $\Intersection(u,v,C^*(u, A_{c}),C^*(v, A_{c}))$. 
For every $u'\in C^*(u,A_{c})$ 
we set $\hat{d}(u,v)$ to $\min(\hat{d}(u,v), d(u,u')+\MTZQuery(u',v))$.
Similarly, for every $v'\in C^*(v,A_{c})$ 
we set  $\hat{d}(u,v)$ to  $\min(\hat{d}(u,v), d(v,v')+\MTZQuery(v',u))$.
Next, we set $\hat{d}(u,v)$ to  $\min(\hat{d}(u,v), h_{c+1}(u)+d(p_{c+1}(u),p_{k-c-2}(v))+h_{k-c-2}(v), h_{k-c-2}(u)+d(p_{k-c-2}(u),p_{c+1}(v))+h_{c+1}(v))$.
Finally, the algorithm returns $\hat{d}(u,v)$ as its estimation.
A pseudo-code for the query algorithm is given in Algorithm~\ref{Algorithm-Distance-Oracle-2k-4c}.
Next, we bound $\hat{d}(u,v)$. 
\begin{lemma}\label{L-Approximation-2k-4c-1-Weighted}
    $\hat{d}(u,v) \le (2k-1-4c)\cdot d(u,v)$
\end{lemma}
\begin{proof}
    Let $P=P(u,v)$, let $\tau_u=\tau_u(P)$ and let $\tau_v=\tau_v(P)$.
    We start by bounding $h_{c+1}(u)$ and $h_{c+1}(v)$ using the following two claims:
    \begin{claim}\label{C-2k-1-4c-H1}
        If $h_c(\tau_u) > d(\tau_u,v)$ then $\hat{d}(u,v) = d(u,v)$
    \end{claim}
    \begin{proof}
        If $h_c(\tau_u) > d(\tau_u,v)$ then we know that $v\in B(\tau_u,A_c)$ and therefore $\tau_u\in C^*(v,A_c)$. By definition, $\tau_u\in C^*(u,A_c)$. Therefore, we get that $\tau_u \in C^*(v,A_c) \cap C^*(u,A_c)$ and from~\autoref{P-Intersection-With-P-Returns-d(u,v)} we get that $\hat{d}(u,v) \le \Intersection(u,v,C^*(u,A_c), C^*(v,A_c)) = d(u,v)$, as required.
    \end{proof}
    
    \begin{claim}\label{C-2k-1-4c-H}
        Either $h_{c+1}(v) \le d(u,v)$ or $\hat{d}(u,v)\leq 3d(u,v)$.
    \end{claim}
    \begin{proof}
    We divide the proof into two cases. The case that $p_c(\tau_u)\in B_{c}(v)$ and the case that $p_c(\tau_u)\not\in B_{c}(v)$.
    Consider the case that $p_c(\tau_u)\in B_{c}(v)$. 
    In the query algorithm, we encounter $\tau_u$ because
    $\tau_u\in C^*(u,A_c)$. 
    Since $p_c(\tau_u) \in B_c(v)$, we have that $\MTZQuery(\tau_u,v) \le 2h_c(\tau_u) + d(\tau_u,v)$. From~\autoref{C-2k-1-4c-H1} it follows that either $h_c(\tau_u) \le d(\tau_u,v)$ or $\hat{d}(u,v)= d(u,v)$, and in the latter case $\hat{d}(u,v) \le 3d(u,v)$ and the claim holds. Thus, we can assume that $h_c(\tau_u) \le d(\tau_u,v)$.
    Since $h_c(\tau_u) \le d(\tau_u,v)$
    we get that $\MTZQuery(\tau_u,v) \le 2h_c(\tau_u) + d(\tau_u,v) \le 3d(\tau_u,v) \le 3d(u,v)$, where the last inequality follows from the fact that $\tau_u\in P$.
    
    Consider now the case that $p_c(\tau_u)\not\in B_{c}(v)$. From~\autoref{L-bound-h_{c+1}(v)-tau_u} it follows that $h_{c+1}(v) \le d(u,v)$.
    
    \end{proof}
    From~\autoref{C-2k-1-4c-H} it follows that either $h_{c+1}(v) \le d(u,v)$ or $\hat{d}(u,v)\le 3d(u,v)$, and also that either $h_{c+1}(u) \le d(u,v)$ or $\hat{d}(u,v)\le 3d(u,v)$. Therefore, it follows that $\max(h_{c+1}(u),h_{c+1}(v))\le d(u,v)$ or the lemma holds. Thus, we can assume that $\max(h_{c+1}(u),h_{c+1}(v))\le d(u,v)$.
    Next, we show that either the lemma holds or that $\min(h_{k-c-2}(u), h_{k-c-2}(v)) \le (2k-2-2c)d(u,v)$.
    
    By applying Property~\ref{P-ADO.Query-Correctness} with $i=k-c-2$ and $j=c+1$, and using the fact that $\min(h_{c+1}(u),h_{c+1}(v))\le d(u,v)$ we get that either 
    \[\min(h_{k-c-2}(u), h_{k-c-2}(v)) \le \min(h_{c+1}(u), h_{c+1}(v))+(k-3-2c)d(u,v) \le (k-2-2c)d(u,v) \]
    or
    \[
        \hat{d}(u,v) \le 2(k-4-2c)d(u,v)+2\min(h_{c+1}(u), h_{c+1}(v)) \le (2k-4c-6)d(u,v),
    \]
    Therefore, we get that either $\min(h_{k-c-2}(u), h_{k-c-2}(v)) \leq (k-2-2c)d(u,v)$ or the lemma holds. Thus, we can assume that $\min(h_{k-c-2}(u), h_{k-c-2}(v)) \leq (k-2-2c)d(u,v)$.
    
    The query algorithm guarantees that 
    $\hat{d}(u,v) \le \min(h_{k-c-2}(u) + d(p_{k-c-2}(u), p_{c+1}(v)) + h_{c+1}(v), h_{c+1}(u) + d(p_{c+1}(u), p_{k-c-2}(v)) + h_{k-c-2}(v))$.
    From the triangle inequality, it follows that 
    $\hat{d}(u,v) \le d(u,v) + 2\min(h_{k-c-2}(u)+h_{c+1}(v), h_{c+1}(u) + h_{k-c-2}(v))$.
    Thus, we get:
    \begin{align*}
        \hat{d}(u,v) &\le \min(h_{k-c-2}(u) + d(p_{k-c-2}(u), p_{c+1}(v)) + h_{c+1}(v), h_{c+1}(u) + d(p_{c+1}(u), p_{k-c-2}(v)) + h_{k-c-2}(v)) \\
        &\le \min(2h_{k-c-2}(u) + d(u,v) +2h_{c+1}(v), 2h_{c+1}(u) + d(u,v) + 2h_{k-c-2}(v))\\
        &\le d(u,v)+2\min(h_{k-c-2}(u)+h_{c+1}(v), h_{c+1}(u)+h_{k-c-2}(v))\\
        &\le d(u,v) + 2\min(h_{k-c-2}(u),h_{k-c-2}(v))+2\max(h_{c+1}(u),h_{c+1}(v))\\
        &\le d(u,v) + 2(k-2-2c)d(u,v) + 2d(u,v) = (2k-1-4c)d(u,v),
    \end{align*}
    where the last inequality follows from the fact that we assume that $\min(h_{k-c-2}(u), h_{k-c-2}(v)) \leq (k-2-2c)d(u,v)$ and $\max(h_{c+1}(u),h_{c+1}(v))\le d(u,v)$.
\end{proof}

\begin{lemma} \label{L-2k-4c-1-weighted-Q-Time}
 The query algorithm takes $\Ot(k\mu n^{c/k})$ time.
\end{lemma}
\begin{proof}
    Computing $C^*(u,A_c)$ and $C^*(v, A_{c})$ takes $\Ot(\mu \cdot |C^*(v, A_{c}|))=\Ot(\mu n^{c/k})$.
    From~\autoref{Intersection-Runtime} computing $\Intersection(u,v,C^*(u, A_{c}),C^*(v, A_{c}))$ takes $O(\min(|C^*(u, A_{c})|, |C^*(u, A_{c})|))$ time. Since $|C^*(u)|=\mu n^{c/k}$, this step costs $O(\mu n^{c/k})$ time.
    For every $w\in C^*(u)$, computing $\MTZQuery(w,v)$ takes $O(k)$ time.
    Since $|C^*(u)|=\mu n^{c/k}$, this step costs $O(k\mu n^{c/k})$ time.
    Computing $\min(\hat{d}(u,v), h_{c+1}(u)+d(p_{c+1}(u),p_{k-c-2}(v))+h_{k-c-2}(v), h_{k-c-2}(u)+d(p_{k-c-2}(u),p_{c+1}(v))+h_{c+1}(v))$ takes $O(1)$ time. 
    Therefore, we conclude that the running time is $O(k \mu n^{1/k})$. 
\end{proof}
\begin{algorithm2e}[t] 
\caption{$\Query(u, v)$}\label{Algorithm-Distance-Oracle-2k-4c}
$\hat{d}(u,v) \gets \Intersection(u,v,C^*(u, A_{c}),C^*(v, A_{c}))$ \\
\lForEach{$u'\in C^*(u, A_{c})$} {
    $\hat{d}(u,v) = \min(\hat{d}(u,v), d(u,u') + \MTZQuery(u',v))$
}
\lForEach{$v'\in C^*(v, A_{c})$} {
    $\hat{d}(u,v) = \min(\hat{d}(u,v), \MTZQuery(u,v') + d(v',v))$
}
$\hat{d}(u,v) = \min(\hat{d}(u,v), h_{k-c-2}(u) + d(p_{k-c-2}(u), p_c(v)) + h_{c+1}(v))$ \\
$\hat{d}(u,v) = \min(\hat{d}(u,v), h_{c+1}(u) + d(p_{c+1}(u), p_{k-c-2}(v)) + h_{k-c-2}(v))$ \\
\Return $\hat{d}(u,v)$
\end{algorithm2e}
The proof of Theorem~\ref{T-DO-2k-1-4c-Weighted} follows from~\autoref{L-Space-2k-5-Weighted} and~\autoref{L-Approximation-2k-4c-1-Weighted} and~\autoref{L-2k-4c-1-weighted-Q-Time}.

\subsection{Immediate corollaries.}\label{S-2k-4c-weighted}
By slightly changing the oracle described above, and instead of saving $d(u,v)$ for every $\pair{u,v}\in A_{k-c-2} \times A_{c+1}$, saving $d(u,v)$ for every $\pair{u,v}\in A_{k-c-1} \times A_c$. 
We get that the storage remains the same, and the construction time is $O(m|A_{k-1-c}|)=O(mn^\frac{c+1}{k})$ ($O(m|A_{k-c}|)=O(mn^\frac{c}{k})$). 
Respectively, we change the last line of the query to be set $\hat{d}(u,v)$ to $\min(\hat{d}(u,v), h_{c}(u)+d(p_{c}(u),p_{k-c-1}(v))+h_{k-c-1}(v), h_{k-c-1}(u)+d(p_{k-c-1}(u),p_{c}(v))+h_{c}(v))$, and to bound $h_c(u) \le d(u,v)/2$ with $\Intersection(B^*(u),B^*(v))$.
We get the following $(2k-4c)$-stretch distance oracle.

\Reminder{T-2k-4c-Weighted}

Next, by setting $c=1$, and saving $d(u,v)$ for every $\pair{u,v}\in A_{k-1} \times A_0$, as in the oracle of~\cite{DBLP:journals/jacm/ThorupZ05}, we get the following $(2k-3)$-stretch distance oracle.

\begin{theorem}[\autoref{S-2k-4c-weighted}]\label{T-2k-3-Weighted} \Copy{T-2k-3-Weighted} 
{
    Let $k \ge 3$ be an integer.
    There is an $\Ot(m+n^{1+1/k})$ space distance oracle that given  two query vertices $u,v\in V$ computes in $O(\mu n^{1/k})$-time a distance estimation $\hat{d}(u,v)$ that satisfies
    $d(u,v)\le \hat{d}(u,v) \le (2k-3)\cdot d(u,v)$. The distance oracle is constructed in $\Ot(mn^{1/k})$ expected time. 
}
\end{theorem}

\section{Distance oracles in dense unweighted graphs}\label{S-more-3}
In this section, we first show that when applying the borderline vertices technique in unweighted graphs, we can avoid paying the $\mu$ factor in the query time at a small price of bounding $h_2(u)$ with $d(u,v)+2$ instead of $d(u,v)$.
We then develop a new technique specifically for the case of unweighted graphs, called \textit{the middle vertex technique}, that allows us to bound $h_2(u)$ with $d(u,v) + 1\ODD$ instead of $2d(u,v)$.
This technique is based on certain properties of a special vertex $\tau$ that lies in the middle of the shortest path between $u$ and $v$. More specifically, we show that either $\tau\in C(u)\cap C(v)$ or
$p_1(\tau)\in B_1(u)\cap B_1(v)$ or  $h_2(u)\leq d(u,v) + 1\ODD$. 
We then present several new small stretch distance oracles that are based on the middle vertex technique.

\subsection{Toolbox}
We now prove several properties of bunches. 
These lemmas were proven in various previous works (see, for example,~\cite {DBLP:conf/stoc/BiloCCC0KS23,DBLP:conf/soda/AgarwalG13,PatrascuR14}).
We prove these properties with respect to a set $A\subseteq V$ but we remark that these properties hold for 
any hierarchy $A_0,A_1,\ldots, A_k$, with $k\geq 2$,  by replacing $A$ with $A_1$.
Let $u,v\in V$, and let $P=P(u,v)$ be the shortest path between $u$ and $v$. 

In the next two lemmas, we prove useful properties of $B(u)=B(u, V, A)$ and $B(v)=B(v, V, A)$ in unweighted graphs. When $B(u) \cap B(v) \neq  \emptyset$ we prove the following lemma:
\begin{lemma}\label{L-Intersect}
     If $B(u) \cap B(v) \neq \emptyset$ then $P \cap (B(u) \cap B(v)) \neq \emptyset$.
\end{lemma}
\begin{proof}
    Let $x\in B(u)\cap B(v)$. 
    Let $u'$ be the farthest vertex from $u$ in $P\cap B(u)$. Since the graph is unweighted it holds for every $w\in B(u)$ that $d(u,w)\leq h(u)-1$. This implies that $d(u, u')=h(u)-1$.
    Assume towards a contradiction that $u'\notin B(v)$. By definition of $B(v)$ we have $d(u',v)\geq h(v)$.
    Since $u'\in P$ we have $d(u,v) = d(u,u') + d(u',v)$. Since $d(u, u')=h(u)-1$ and $d(u',v)\geq h(v)$ we get that 
    $d(u,v) \geq h(u)-1+h(v)$. 
    On the other hand, from the triangle inequality, it follows that $d(u,v) \le d(u,x) + d(x,v)$.
    Since $x\in B(u)$ and the graph is unweighted we have that $d(u,x)\leq h(u)-1$. Similarly, since $x\in B(v)$ we have $d(v,x)\leq h(v)-1$.
    Thus, we get that $d(u,v) \le d(u,x) + d(x,v)\leq h(u)-1+h(v)-1$.
    In contradiction to the fact that $d(u,v) \geq h(u)-1+h(v)$. 
    Therefore, $u'\in B(v)$, and $P \cap (B(u) \cap B(v)) \neq \emptyset$, as required.
\end{proof}

When $B(u) \cap B(v) = \emptyset$ we prove the following lemma:
\begin{lemma}\label{L-NO-Intersect}
    If $B(u) \cap B(v) = \emptyset$ then $h(u) + h(v) \le d(u,v)+1$, and specifically $\min(h(u),h(v))\le d(u,v)/2 + 0.5\ODD$.
\end{lemma}
\begin{proof}
To prove that $\min(h(u),h(v))\le d(u,v)/2 +0.5\ODD$ it is enough to show that  $h(u)+h(v) \le d(u,v)+1$, because this implies that $\min(h(u),h(v))\leq (d(u,v)+1)/2$, and since both $h(u)$ and $h(v)$ are integers we have $\min(h(u),h(v))\leq \lfloor d(u,v)/2 + 1/2 \rfloor=d(u,v)/2 + 0.5\ODD$, as required.

Let $u'$ be the farthest vertex from $u$ in $P\cap B_0(u)$. Thus, $d(u,u')=h(u)-1$. 
Assume for the sake of contradiction that $h(u)+h(v) > d(u,v)+1$. 
This implies that $h(v) > d(u,v)-h(u)+1$. 
    
Since $u'$ is in $P$, we know that $d(u,v) = d(u,u') + d(u',v)$ and since $d(u,u')= h(u)-1$ we get that $d(u,v)=h(u)-1 + d(u',v)$, thus $d(u',v)=d(u,v)-h(u)+1$. 
Now since $h(v) > d(u,v)-h(u)+1$ we get that $h(v)>d(u',v)$. This implies that $u'\in  B_0(v)$, contradiction to the fact that $B_0(u) \cap B_0(v) \neq \emptyset$.
\end{proof}

Next, we define $\tau'_u(P)$ and $\tau'_v(P)$, where $\tau'_u(P)$ ($\tau'_v(P)$) is the farthest vertex from $u$ ($v$) in $C(u, A_c)\cap P$ ($C(v, A_c)\cap P$).

In the following lemma, we prove a bound on $h_c(\tau'_u)$.
\begin{lemma}\label{L-bound-C-h}
    $h_c(\tau'_u) \le d(u,\tau'_u)+2$.
\end{lemma}
\begin{proof}
    From the definitions of $\tau_u$ and $\tau'_u$ we know that $\tau_u\in N(\tau'_u)$. Therefore, $d(\tau'_u, \tau_u) = 1$. From the triangle inequality we know that $h_c(\tau'_u) \le d(\tau'_u,\tau_u)+h_c(\tau_u) \le 1+h_c(\tau_u)$.
    Since $\tau'_u \in P(u,\tau_u)$ we know that $d(u,\tau_u) = d(u,\tau'_u) + 1$.
    From~\autoref{L-bound-C^*-h} we know that $h_c(\tau_u) \le d(u,\tau_u)$. 
    Overall, we get that 
    $h_c(\tau'_u) \le 1+h_c(\tau_u) \le 1+d(u,\tau_u)+1=d(u,\tau_u)+2$, as required.
\end{proof}

In the case that $p_c(\tau_u) \notin B_c(v)$ we can also a prove a bound on $h_{c+1}(v)$.
\begin{lemma}\label{L-bound-h_{c+1}(v)-tau_u-+2}
    If $p_c(\tau'_u) \notin B_c(v)$ then $h_{c+1}(v) \le d(u,v)+2$.
\end{lemma}
\begin{proof}
    From the definition of $B_c(v)$ it follows that since $p_c(\tau'_u) \notin B_c(v)$ then $h_{c+1}(v) \le d(v,p_c(\tau'_u))$.
    From the triangle inequality, it follows that 
    $d(v,p_c(\tau'_u)) \le d(v,\tau'_u) + h_c(\tau'_u)$.
    From~\autoref{L-bound-C-h} we know that $h_c(\tau'_u) \le d(u,\tau'_u)+2$.
    Thus, $h_{c+1}(v) \le d(v,\tau'_u) + h_c(\tau'_u) \le d(v,\tau'_u) + d(\tau'_u,u) +2= d(u,v)+2$, where the last inequality follows from the fact that $\tau'_u\in P(u,v)$, as required.
\end{proof}

\subsection{New technique: the middle vertex \texorpdfstring{$\tau$}{}}\label{S-middle}

Let $P=P(u,v)$ be a shortest path between $u$ and $v$. We define the \textit{middle} vertex  $\tau(P,u,v)$ to be a vertex of $P$ that satisfies $d(u,\tau(P,u,v))=d(u,v)/2+0.5\ODD$. Throughout the paper when  $P$, $u$, and $v$ are clear from the context we use $\tau$ instead of $\tau(P,u,v)$.  In the next two lemmas, we prove useful relations between $C(u)=C(u, A)$, $C(v)=C(v, A)$, and the middle vertex $\tau$.
We first show:

\begin{lemma}
    \label{L-h-rho-big}
    If $h(\tau) > d(u,v)/2+0.5\ODD$ then $\tau\in C(u)\cap C(v)$.
\end{lemma}
\begin{proof}
By definition we have  $d(u,\tau)=d(u,v)/2+0.5\ODD$ and $d(v,\tau)\leq d(u,v)/2-0.5\ODD$. Since $h(\tau) > d(u,v)/2+0.5\ODD$ we get that $\tau\in C(u)$ and  $\tau\in  C(v)$, as required.
\end{proof}

Next, we show: 
\begin{lemma}
     \label{L-tau-In-B1uB1v}
     $d(u,p_1(\tau)) + d(p_1(\tau),v)\leq  d(u,v)+2h_1(\tau)$.
\end{lemma} 
\begin{proof}
From the triangle inequality it follows that $d(u,p_1(\tau))\leq d(u,\tau)+h_1(\tau)$ and $d(v,p_1(\tau))\leq d(v,\tau)+h_1(\tau)$. Thus, 
$d(u,p_1(\tau))+d(p_1(\tau),v)\leq d(u,\tau)+d(v,\tau)+2h_1(\tau)$. 
Since $\tau\in P$  we have $d(u,\tau)+d(v,\tau)=d(u,v)$ and we get $d(u,\tau)+d(v,\tau)+2h_1(\tau)\leq d(u,v)+2h_1(\tau)$, as required.
\end{proof}

Consider now  $B_1(u)$ and $ B_1(v)$ in the case that $k>2$ and assume  that $p_1(\tau)\notin B_1(u) \cap B_1(v)$. We show:
\begin{lemma}
    \label{L-tau-Not-In-B1uB1v}
    If $p_1(\tau) \notin B_1(u) \cap B_1(v)$ then either $h_2(u) \leq d(u,v)/2+0.5\ODD+h_1(\tau)$ or $h_2(v)\leq d(u,v)/2+0.5\ODD+h_1(\tau)$.
\end{lemma}
\begin{proof}
Since $p_1(\tau)\notin B_1(u)\cap B_1(v)$, we have either $p_1(\tau)\notin B_1(u)$ or $p_1(\tau)\notin B_1(v)$. 
If $p_1(\tau)\notin B_1(u)$ then from the definition of $B_1(u)$ we have $h_2(u) \leq d(u,p_1(\tau))$.
Similarly, if $p_1(\tau)\notin B_1(v)$  then $h_2(v) \leq d(v,p_1(\tau))$.
Thus, either $h_2(u) \leq d(u,p_1(\tau))$ or $h_2(v)\leq d(v,p_1(\tau))$.
From $\tau$'s definition  
  $d(u,\tau)=d(u,v)/2+0.5\ODD$ and $d(v,\tau)\leq d(u,v)/2+0.5\ODD$. Thus, 
either  
$h_2(u) \leq d(u,v)/2+0.5\ODD+h_1(\tau)$ or $h_2(v) \leq d(u,v)/2+0.5\ODD+h_1(\tau)$, as required.
\end{proof}

\subsection{$(3,2\ODD)$-approximation with \texorpdfstring{$n^{4/3}$}{n\^(4/3)}-space and \texorpdfstring{$n^{1/3}$}{n\^(1/3)}-query} \label{S-4-2}
In this section we use the the middle vertex technique  and  prove the following theorem:
\Reminder{T-DO-3-Unweighted}
\subsubsection{Storage.}
We construct a hierarchy $A_0,A_1,A_2,A_3$, where $A_0=V$, $A_3=\emptyset$. 
The set $A_1$ is constructed using~\autoref{L-A-center}, with parameter $n^{-1/3}$. 
The set $A_2$ is constructed using~\autoref{L-TZ-Size} with $k=3$. 
We save $B_i(u)$ and $p_i(u)$, for every $u\in V$ and $0\leq i \leq 2$. We also save $C(w)$, for every $w\in V\setminus A_1$.
\begin{lemma}\label{L-Space-3-Unweighted}
    The distance oracle is constructed in  $\Ot(mn^{1/3})$-time  and  uses $\Ot(n^{4/3})$-space.
\end{lemma}
\begin{proof}
    From~\autoref{L-A-center} it follows that $|C(w)|=O(n^{1/3})$, for every $w\in A_1\setminus V$, and also that  $|B_0(u)|=O(n^{1/3})$,  for every $u\in V$.
    From~\autoref{L-TZ-Size} it follows that  $|B_1(u)|=O(n^{1/3})$. 
    Since $B_2(u)=A_2$, $|B_2(u)|=\Ot(n^{1/3})$.
    The cost of saving $p_i(u)$, for every $u\in V$, and  $0\leq i \leq 2$, is $O(n)$. 
    We conclude that the total space  is $\Ot(n^{4/3})$. 
    From~\autoref{L-A-center} it follows that computing $A_1$ takes $\Ot(mn^{1/3})$ time. 
    From~\autoref{L-TZ-Size} it follows that computing $A_2$, $B_1(u)$ and $B_2(u)$ takes $\Ot(mn^{1/3})$ time. 
    Computing $p_i(u)$, for every $u\in V$, and  $0\leq i \leq 2$, takes $O(m)$ time.  We conclude that the running time  is $\Ot(mn^{1/3})$. 
\end{proof}

\subsubsection{Query algorithm.}
The input to the query algorithm is two vertices $u,v\in V$.
The output of the query is an estimation $\hat{d}(u,v)$, which is the minimum between $\Intersection(u,v,C(u), C(v))$, \\$\Intersection(u, v, B_1(u), B_1(v))$ and 
$\min(d(u,p_2(u))+d(p_2(u),v), d(u,p_2(v))+d(p_2(v),v))$. 
A pseudo-code for the query algorithm is given in Algorithm~\ref{Algorithm-Q3}.
\begin{algorithm2e}[t] 
\caption{$\Query(u, v)$}\label{Algorithm-Q3}
$\hat{d}(u,v) \gets \Intersection(u,v,C(u), C(v))$ \\
$\hat{d}(u,v) \gets \min(\hat{d}(u,v), \Intersection(u, v, B_1(u), B_1(v))$ \\
$\hat{d}(u,v) \gets \min(\hat{d}(u,v), d(u,p_2(u))+d(p_2(u),v), d(u,p_2(v))+d(p_2(v),v))$ \\
\Return $\hat{d}(u,v)$
\end{algorithm2e}
Next, we bound $\hat{d}(u,v)$. 

\begin{lemma}\label{L-Q3-Bound}
    $\hat{d}(u,v) \leq 3d(u,v) + 2\ODD$.
\end{lemma}
\begin{proof}
Let $P$ be a shortest path between $u$ and $v$. Recall that $\tau=\tau(u,v,P)$, is the middle vertex between $u$ and $v$ in $P$.
If $h_1(\tau) > d(u,v)/2+0.5\ODD$  then it follows from~\autoref{L-h-rho-big} that $\tau\in C(u)\cap C(v)$.
In such a case the value returned by the call to  $\Intersection(u,v,C(u),C(v))$ is $d(u,v)$ and the claim holds. (See~\autoref{F-mid-vertex}(a)).

Otherwise, $h_1(\tau) \leq  d(u,v)/2+0.5\ODD$. 
We divide the rest of the proof into two cases. 
The case that $p_1(\tau) \in B_1(u)\cap B_1(v)$ and the case that $p_1(\tau) \notin B_1(u)\cap B_1(v)$.

If $p_1(\tau) \in B_1(u)\cap B_1(v)$ then $\Intersection(u,v,B_1(u),B_1(v))$ returns a value bounded by $d(u,p_1(\tau))+d(p_1(\tau),v)$. 
From~\autoref{L-tau-In-B1uB1v} it follows that $d(u,p_1(\tau))+d(p_1(\tau),v)\leq d(u,v)+2h_1(\tau)$. 
Since $h_1(\tau) \leq d(u,v)/2+0.5\ODD$ we get that $\hat{d}(u,v) \le d(u,v)+2h_1(\tau)\le d(u,v) + 2(d(u,v)/2+0.5\ODD) = 2d(u,v)+1\ODD$, as required. (See~\autoref{F-mid-vertex}(b)).

Consider now the case that $p_1(\tau) \notin B_1(u)\cap B_1(v)$. It follows from~\autoref{L-tau-Not-In-B1uB1v}
that either $h_2(u) \leq d(u,v)/2+0.5\ODD + h_1(\tau)$ or $h_2(v) \leq d(u,v)/2+0.5\ODD+h_1(\tau)$.
Assume that $h_2(u) \leq d(u,v)/2+0.5\ODD+h_1(\tau)$. (A symmetrical argument holds for the case that $h_2(v) \leq d(u,v)/2+0.5\ODD +h_1(\tau)$.) 
Since $h_1(\tau) \leq  d(u,v)/2+0.5\ODD$, we get that 
$h_2(u) \leq d(u,v)/2+0.5\ODD + h_1(\tau) \leq 2(d(u,v)/2+0.5\ODD)=d(u,v)+1\ODD$.
Since $\hat{d}(u,v)\leq \min(h_2(u)+d(p_2(u),v), d(u,p_2(v))+h_2(v))$, we know that $\hat{d}(u,v) \leq h_2(u)+d(p_2(u),v)$. From the triangle inequality, we know that $d(p_2(u),v) \leq h_2(u) + d(u,v)$, and we get that $\hat{d}(u,v)\leq h_2(u)+d(p_2(u),v)\leq 2h_2(u) + d(u,v)$.
Since $h_2(u) \leq 2(d(u,v)/2+0.5\ODD)=d(u,v)+1\ODD$, we get that $\hat{d}(u,v) \leq 2h_2(u) + d(u,v) \leq d(u,v) + 2(d(u,v)+1\ODD)=3d(U,v)+2\ODD$, as required. (See~\autoref{F-mid-vertex}(c)).
\end{proof}

\begin{lemma}\label{L-Q3-Time}
 The query algorithm takes $O(n^{1/3})$ time.
\end{lemma}
\begin{proof}
    By~\autoref{Intersection-Runtime} the runtime of $\Intersection(u,v,C(u),C(v))$  is $O(\min(|C(u)|,|C(v)|))$ 
    and the runtime of $\Intersection(u,v,B_1(u),B_1(v))$ is $O(\min(|B_1(u)|,|B_1(v)|))$. 
    Since $\min(|C(u)|,|C(v)|)=O(n^{1/3})$ and $\min(|B_1(u)|,|B_1(v)|)=O(n^{1/3})$, we get that the cost of these two steps is $O(n^{1/3})$. 
    Computing $d(u,p_2(u))+d(p_2(u),v), d(u,p_2(v))+d(p_2(v),v))$ takes $O(1)$ time. We conclude that the running time  is $O(n^{1/3})$. 
\end{proof}

Theorem~\ref{T-DO-3-Unweighted} follows from~\autoref{L-Space-3-Unweighted},~\autoref{L-Q3-Bound} and~\autoref{L-Q3-Time}.

\subsection{$(4,3\ODD)$-approximation with \texorpdfstring{$n^{5/4}$}{n\^(5/4)}-space and \texorpdfstring{$n^{1/4}$}{n\^(1/4)}-query} \label{S-DO-4-Unweighted}

In this section we show that if we save in addition $d(x,y)$, for  $\langle x,y\rangle \in A_2 \times A_1$,
it is possible using the middle vertex technique to  prove also the following:

\Reminder{T-DO-4-Unweighted}
\subsubsection{Storage.}
We construct a hierarchy $A_0,A_1,A_2,A_3$, where $A_0=V$, $A_3=\emptyset$. 
The set $A_1$ is constructed using~\autoref{L-A-center}, with parameter $n^{-1/4}$. 
The set $A_2$ is constructed using~\autoref{L-TZ-Size} with $k=4$. 
We save $B_i(u)$ and $p_i(u)$, for every $u\in V$ and $0\leq i \leq 2$. We also save $C(w)$, for every $w\in V\setminus A_1$.
In addition, for every $x\in A_2$ and $y\in A_1$ we store $d(x,y)$.

\begin{lemma}\label{L-Distance-Oracle-6-Space-Size}
    The distance oracle is constructed in $\Ot(mn^{1/2})$-time  and  uses $\Ot(n^{5/4})$-space.
\end{lemma}
\begin{proof}
    From~\autoref{L-A-center} it follows that $|C(w)|=O(n^{1/4})$, for every $w\in A_1\setminus V$, and also that  $|B_0(u)|=O(n^{1/4})$,  for every $u\in V$.
    From~\autoref{L-TZ-Size} it follows that  $|B_1(u)|=O(n^{1/4})$. 
    Since $B_2(u)=A_2$, $|B_2(u)|=\Ot(n^{2/4})$.
    The cost of saving $p_i(u)$, for every $u\in V$, and  $0\leq i \leq 2$, is $O(n)$. 
    The cost of saving $d(x,y)$ for every $x\in A_2$ and $y\in A_1$, is $O(|A_1||A_2|)=O(n^{3/4} \cdot n^{2/4})=O(n^{5/4})$.
    We conclude that the total space  is $\Ot(n^{5/4})$. 
    
    By~\autoref{L-A-center}  computing $A_1$ takes $\Ot(mn^{1/4})$ time. 
    By~\autoref{L-TZ-Size}   computing $A_2$, $B_1(u)$ and $B_2(u)$ takes $\Ot(mn^{1/4})$ time. 
    Computing $p_i(u)$, for every $u\in V$, and  $0\leq i \leq 2$, takes $O(m)$ time.
    Computing $d(x,y)$ for every $x\in A_2$ and $y\in A_1$ takes $O(m|A_2|)=\Ot(mn^{1/2})$ time.
    Thus, the running time  is $\Ot(mn^{1/2})$. 
\end{proof}

\subsubsection{Query algorithm.}
The input to the query algorithm is two vertices $u,v\in V$.
The output of the query is an estimation $\hat{d}(u,v)$.
The algorithm first sets $\hat{d}(u,v)$ to be the minimum between \\$\Intersection(u,v,C(u), C(v))$, $\Intersection(u,v,B_0(u), B_0(v))$ and $\Intersection(u,v,B_1(u), B_1(v))$.\;\;
 Next, for each $w_1\in B_1(v)$ the algorithm sets  $\hat{d}(u,v)$ to be $\min (\hat{d}(u,v) ,d(u,p_2(u))+d(p_2(u),w_1)+d(w_1,v))$.
Similarly, for each $w_1\in B(u)$ the algorithm sets   $\hat{d}(u,v)$ to be $\min (\hat{d}(u,v),d(u,w_1)+d(w_1, p_2(v))+d(p_2(v),v))$.
A pseudo-code for the query algorithm is given in Algorithm~\ref{A-DO-4}.
\begin{algorithm2e}[t] 
\caption{$\Query(u, v)$}\label{A-DO-4}
$\hat{d}(u,v) \gets \Intersection(u,v,C(u), C(v))$ \\ 
$\hat{d}(u,v) \gets \min(\hat{d}(u,v), \Intersection(u,v,B_0(u), B_0(v))$ \\ 
$\hat{d}(u,v) \gets \min(\hat{d}(u,v), \Intersection(u,v,B_1(u), B_1(v))$ \\
\lForEach{$w_1\in B_1(v)$} {
    $\hat{d}(u,v)\gets \min(\hat{d}(u,v), d(u,p_2(u))+d(p_2(u),w_1)+d(w_1,v))$
}
\lForEach{$w_1\in B_1(u)$} {
    $\hat{d}(u,v)\gets \min(\hat{d}(u,v), d(u,w_1)+d(w_1, p_2(v))+d(p_2(v),v))$
}
\Return $\hat{d}(u,v)$
\end{algorithm2e}
Next, we bound  $\hat{d}(u,v)$. 
\begin{lemma}\label{L-Q-k-is-4-Bound}
    $\hat{d}(u,v) \leq 4d(u,v) + 3\ODD$.
\end{lemma}
\begin{proof}
Let $P$ be a shortest path between $u$ and $v$ and let  $\tau=\tau(u,v,P)$ be the middle vertex.
If $h_1(\tau) > d(u,v)/2+0.5\ODD$ the claim holds since $\tau\in C(u)\cap C(v)$ by~\autoref{L-h-rho-big} and
 the value returned by the call to  $\Intersection(u,v,C(u),C(v))$ is $d(u,v)$. 
Therefore, for the rest of the proof we have $h_1(\tau) \leq  d(u,v)/2+0.5\ODD$. 

We divide the rest of the proof into the case that $p_1(\tau) \in B_1(u)\cap B_1(v)$ and the case that $p_1(\tau) \notin B_1(u)\cap B_1(v)$. 

If $p_1(\tau) \in B_1(u)\cap B_1(v)$ then the call to $\Intersection(u,v,B_1(u),B_1(v))$ ensures that $\hat{d}(u,v)\leq d(u,p_1(\tau))+d(p_1(\tau),v)$. 
From~\autoref{L-tau-In-B1uB1v} it follows that $d(u,p_1(\tau))+d(p_1(\tau),v)\leq d(u,v)+2h_1(\tau)$. 
Since $h_1(\tau) \leq d(u,v)/2+0.5\ODD$ we get that $\hat{d}(u,v)\leq d(u,v)+2h_1(\tau)\le d(u,v) + 2(d(u,v)/2+0.5\ODD)=2d(u,v)+1\ODD$, as required.

Consider now the case that $p_1(\tau) \notin B_1(u)\cap B_1(v)$. 
Combining~\autoref{L-tau-Not-In-B1uB1v} with the fact that $h_1(\tau) \leq  d(u,v)/2+0.5\ODD$ we get that either $h_2(u) \leq d(u,v)/2+0.5\ODD+h_1(\tau)\leq 2(d(u,v)/2+0.5\ODD)=d(u,v)+1\ODD$ or $h_2(v) \leq d(u,v)/2+0.5\ODD+h_1(\tau) \leq 2(d(u,v)/2+0.5\ODD)=d(u,v)+1\ODD$. 
Assume, wlog, that 
$h_2(u) \leq d(u,v)+1\ODD$.

From~\autoref{L-Intersect} it follows that if $B_0(u)\cap B_0(v) \neq \emptyset$ then $P\cap B_0(u)\cap B_0(v) \neq \emptyset$. In such a case the call to $\Intersection(u,v,B_0(u),B_0(v))$ returns  $d(u,v)$ and the claim holds.
Otherwise, $B_0(u)\cap B_0(v) =\emptyset$ and by~\autoref{L-NO-Intersect} we have $\min(h_1(u), h_1(v)) \le d(u,v)/2+0.5\ODD$. Therefore, for the rest of the proof we have $\min(h_1(u), h_1(v)) \le d(u,v)/2+0.5\ODD$.

We divide the rest of the proof into two cases. The case that  $p_1(\tau)\in B_1(v)$ and the case that $p_1(\tau)\notin B_1(v)$.\footnote{Notice that we assumed that $h_2(u) \leq 2(d(u,v)/2+0.5\ODD)$. In case that $h_2(v) \leq d(u,v)+1\ODD$
the two cases are $p_1(\tau)\in B_1(u)$ and  $p_1(\tau)\notin B_1(u)$.}
Consider first the case that $p_1(\tau)\in B_1(v)$. The query algorithm encounters  $p_1(\tau)$ while iterating over $B_1(v)$.  Thus, $\hat{d}(u,v) \leq h_2(u) + d(p_2(u), p_1(\tau)) + d(p_1(\tau), v)$.
(See the dotted path in~\autoref{F-mid-vertex}(b).)

From the triangle inequality we have that $d(p_2(u), p_1(\tau)) \leq h_2(u) + d(u,\tau) + h_1(\tau)$ and 
$d(p_1(\tau), v) \leq h_1(\tau) + d(\tau,v)$. 

Therefore, we get: 
\begin{align*} 
\hat{d}(u,v)    &\leq  h_2(u) + d(p_2(u), p_1(\tau)) + d(p_1(\tau), v)\leq  h_2(u) +(h_2(u) + d(u,\tau) + h_1(\tau)) + (h_1(\tau) + d(\tau,v)) \\ 
                &= 2h_2(u) + 2h_1(\tau) +d(u,\tau) + d(\tau,v)= d(u,v) + 2h_2(u) + 2h_1(\tau).
\end{align*}

The last equality follows from the fact that  $\tau\in P$ and $d(u,\tau)+d(\tau,v)=d(u,v)$.
Since we are in the case that $h_1(\tau) \leq d(u,v)/2+0.5\ODD$ and the case that $h_2(u) \leq d(u,v)+1\ODD$, it follows that $\hat{d}(u,v) \leq d(u,v) + 2h_2(u) + 2h_1(\tau) \leq 4d(u,v) + 3\ODD$, as required.

Consider now the case that $p_1(\tau)\notin B_1(v)$. 
This implies that $h_2(v) \leq d(v,p_1(\tau))$. From the triangle inequality we have that $d(v,p_1(\tau)) \leq d(v,\tau) + h_1(\tau)$. By definition of $\tau$  we have $d(v,\tau)\leq d(u,v)/2+0.5\ODD$. Recall also that we are in the case that $h_1(\tau) \leq d(u,v)/2+0.5\ODD$. 

Therefore, we get: 
\[
h_2(v)  \leq  d(v,p_1(\tau))
        \leq  d(v,\tau) + h_1(\tau)
        \leq  d(u,v)/2+0.5\ODD+d(u,v)/2+0.5\ODD=d(u,v) + 1\ODD
\]

As $p_1(u)\in B_1(u)$, the query algorithm encounters $p_1(u)$ while iterating over $B_1(u)$. Thus, $\hat{d}(u,v) \leq h_1(u) + d(p_1(u), p_2(v)) + h_2(v)$. (See the dotted path in~\autoref{F-mid-vertex}(c).)
From the triangle inequality we have that $d(p_1(u),p_2(v)) \leq h_1(u) + d(u,v) + h_2(v)$. Thus, we get: \[\hat{d}(u,v) \leq h_1(u) + d(p_1(u), p_2(v)) + h_2(v) \leq h_1(u) +( h_1(u) + d(u,v) + h_2(v)) + h_2(v) = d(u,v)+2h_2(v)+2h_1(u)\]

As $p_1(v)\in B_1(v)$ using a symmetrical argument for $p_1(v)$ we get  $\hat{d}(u,v) \leq  d(u,v)+2h_2(u)+2h_1(v)$. 

Since we are in the case that $\min(h_1(u), h_1(v)) \le d(u,v)/2+0.5\ODD$ then either 
$h_1(u) \leq d(u,v)/2+0.5\ODD$ or $h_1(v) \leq d(u,v)/2+0.5\ODD$.  
In the case that $h_1(u) \leq d(u,v)/2+0.5\ODD$  
we use the inequality $\hat{d}(u,v)\leq d(u,v)+2h_2(v)+2h_1(u)$. 
We also have shown above that $h_2(v) \leq d(u,v)+1\ODD$. Thus, we   
get $\hat{d}(u,v) \leq d(u,v) + 2(d(u,v)+1\ODD)+2d(u,v)/2+0.5\ODD=4d(u,v)+3\ODD$, as required. 

In the case that $h_1(v) \leq d(u,v)/2+0.5\ODD$  
 we use the inequality $\hat{d}(u,v)\leq d(u,v)+2h_2(u)+2h_1(v)$. 
Since we are under the assumption that $h_2(u) \leq d(u,v)+1\ODD$, we   
get that $\hat{d}(u,v) \leq 4d(u,v)+3\ODD$, as required. 

\end{proof}

\begin{lemma}\label{L-Q4-Time}
 The query algorithm takes $O(n^{1/4})$ time.
\end{lemma}
\begin{proof}
    By~\autoref{Intersection-Runtime} the runtime of $\Intersection(u,v,C(u),C(v))$  is $O(\min(|C(u)|,|C(v)|))$ 
    and the runtime of $\Intersection(u,v,B_1(u),B_1(v))$ is $O(\min(|B_1(u)|,|B_1(v)|))$. 
    Since $\min(|C(u)|,|C(v)|)=O(n^{1/4})$ and $\min(|B_1(u)|,|B_1(v)|)=O(n^{1/4})$, we get that the cost of these two steps is $O(n^{1/4})$. 
    Computing $d(u,p_2(u))+d(p_2(u),w_1)+d(w_1,v)$ for every $w_1\in B_1(u)$ takes $O(|B_1(u)|)=O(n^{1/4})$ time.
    Computing $d(u,w_1)+d(w_1, p_2(v))+d(p_2(v),v)$ for every $w_1\in B_1(u)$ takes $O(|B_1(u)|)=O(n^{1/4})$ time.
    We conclude that the running time is $O(n^{1/4})$.
\end{proof}
Theorem~\ref{T-DO-4-Unweighted} follows from~\autoref{L-Distance-Oracle-6-Space-Size},~\autoref{L-Q-k-is-4-Bound} and~\autoref{L-Q4-Time}.

\subsection{$(3,2+2\ODD)$-approximation with \texorpdfstring{$n^{5/4}$}{n\^(5/4)}-space and \texorpdfstring{$n^{1/2}$}{n\^(1/2)}-query}\label{S-3,2+2ODD-4}
In this section, we use the distance oracle presented in Section~\ref{S-DO-4-Unweighted} with additional query time. We obtain:

\Reminder{T-DO-3-2-Unweighted}

\subsubsection{Storage.}
The storage is the same as in Section~\ref{S-DO-4-Unweighted}.
We construct a hierarchy $A_0,A_1,A_2,A_3$, where $A_0=V$, $A_3=\emptyset$. 
The set $A_1$ is constructed using Lemma~\ref{L-A-center}, with parameter $n^{-1/4}$. 
The set $A_2$ is constructed using Lemma~\ref{L-TZ-Size} with $k=4$. 
We save $B_i(u)$ and $p_i(u)$, for every $u\in V$ and $0\leq i \leq 2$. We also save $C(w)$, for every $w\in V\setminus A_1$.
In addition, for every $x\in A_2$ and $y\in A_1$ we store $d(x,y)$.
The distance oracle is constructed in  $\Ot(mn^{1/2})$-time  and  uses $\Ot(n^{5/4})$-space, as proven in Lemma~\ref{L-Distance-Oracle-6-Space-Size}.

\subsubsection{Query algorithm.}
The input to the query algorithm is two vertices $u,v\in V$.
The output of the query is an estimation $\hat{d}(u,v)$.
The algorithm first sets $\hat{d}(u,v)$ to be the minimum between 
% $\Intersection(u,v,\bigcup_{w\in B_0(u)}{B_0(w)}, \bigcup_{w\in B_0(v)}{B_0(w)})$, 
$\Intersection(u,v,\bigcup_{w\in B_0(u)}$ \\${C(w)}, \bigcup_{w\in B_0(v)}{C(w)})$ and 
$\Intersection(u,v,\bigcup_{w\in B_0(u)}{B_1(w)}, \bigcup_{w\in B_0(v)}{B_1(w)})$.

Next, for every vertex $w\in B_0(u)$ the algorithm sets $\hat{d}(u,v)$ to be $\min(\hat{d}(u,v), d(u,w) + d(w,p_2(w)) + d(p_2(w),p_1(v)) + d(p_1(v),v))$ and for every vertex $w\in B_0(v)$ the algorithm sets $\hat{d}(u,v)$ to be $\min(\hat{d}(u,p_1(u)) + d(p_1(u),p_2(w)) + d(p_2(w),w) + d(w,v))$.
Finally, the algorithm returns $\hat{d}(u,v)$ as the estimated distance.
A pseudo-code for the query algorithm is given in Algorithm~\ref{A-Query-5/4-1/2-3}.
Next, we bound  $\hat{d}(u,v)$. 
\begin{lemma} \label{L-Bound-4-3}
    $\hat{d}(u,v) \le 3d(u,v) + 2 + 2\ODD$. 
\end{lemma}
\begin{proof}
    Let $P$ be a shortest path from $u$ to $v$. Let $u'\in B_0(u)$ be the farthest vertex in $P$ from $u$, and let $v'\in B_0(v)$ be the farthest vertex in $P$ from $v$. Let $\tau'=\tau(u',v',P(u',v'))$ be the middle vertex between $u'$ and $v'$ in $P(u',v')$.

    We divide the proof into two cases. The case that $h_1(\tau') > \lceil d(u', v')/2 \rceil$ and the case that $h_1(\tau') \le \lceil d(u', v')/2 \rceil$.
    Consider the case that $h_1(\tau') > \lceil d(u', v')/2 \rceil$. From Lemma~\ref{L-h-rho-big} we know that $\tau'\in C(u')\cap C(v')$.
    Since $u'\in B_0(u)$ and $v'\in B_0(v)$ the value returned by  $\Intersection(u,v,\bigcup_{w\in B_0(u)}{C_0(w)},\bigcup_{w\in B_0(v)}{C_0(w)})$ is at most $d(u,u') + d(u',\tau') + d(v',\tau') + d(v',v)$. 
    Since $B_0(u)\cap B_0(v) =\emptyset$ 
    when traveling on $P$ from $u$ to $v$, we first encounter $u'$ and then $v'$. Since $\tau'$ is the middle vertex between $u'$ and $v'$ in 
    $P(u',v')$ we encounter $\tau'$ when traveling on $P(u',v')$ from $u'$ to $v'$. Thus, $d(u,u') + d(u',\tau') + d(v',\tau') + d(v',v)=d(u,v)$ and the claim holds.
    
    Consider now the case that $h_1(\tau') \le \lceil d(u', v')/2 \rceil$.
    We divide the rest of the proof into two cases. The case that $p_1(\tau')\in B_1(u') \cap B_1(v')$ and the case that $p_1(\tau')\notin B_1(u') \cap B_1(v')$.
    Consider the case that $p_1(\tau')\in B_1(u') \cap B_1(v')$. Since $B_1(u') \subseteq \bigcup_{w\in B_0(u)}{B_1(w)}$ and $B_1(v') \subseteq \bigcup_{w\in B_0(v)}{B_1(w)}$ it follows that $p_1(\tau')\in \big(\bigcup_{w\in B_0(u)}{B_1(w)} \big) \cap \big( \bigcup_{w\in B_0(v)}{B_1(w)}\big)$, and the call $\Intersection(u,v,\bigcup_{w\in B_0(u)}{B_1(w)}, \bigcup_{w\in B_0(v)}{B_1(w)})$ returns a value bounded by $d(u,u') + d(u',p_1(\tau')) + d(v', p_1(\tau')) + d(v',v)$.
    From Lemma~\ref{L-tau-In-B1uB1v} it follows that $d(u',p_1(\tau'))+d(p_1(\tau'),v')\leq d(u',v')+2h_1(\tau')$. 
    Since $h_1(\tau') \leq \lceil d(u',v')/2\rceil$ we get that $\hat{d}(u,v) \le d(u,u') + d(u',v') + d(v',v) + 2h_1(\tau') \le d(u,u') + d(u',v') + d(v',v) + 2\lceil d(u',v')/2 \rceil$.

    As before since $B_0(u)\cap B_0(v) =\emptyset$ we have $d(u,v) = d(u,u') + d(u',v') + d(v',v)$  and $\hat{d}(u,v) \le d(u,v) + 2\lceil d(u',v')/2 \rceil \le d(u,v) + 2\lceil d(u,v)/2 \rceil$, 
    as required.

    Consider now the case that $p_1(\tau')\notin B_1(u') \cap B_1(v')$. It follows from Lemma~\ref{L-tau-Not-In-B1uB1v}
    that either $h_2(u') \leq \lceil d(u',v')/2\rceil+h_1(\tau')$ or $h_2(v') \leq \lceil d(u',v')/2\rceil+h_1(\tau')$.
    Recall that we are in the case that $h_1(\tau') \le \lceil d(u',v')/2 \rceil$ and thus we have that either $h_2(u') \leq 2\lceil d(u',v')/2\rceil)$ or $h_2(v') \leq 2\lceil d(u',v')/2\rceil$.
    Assume, wlog, that $h_2(u') \leq 2\lceil d(u',v')/2\rceil$.

    The algorithm iterates over $B_0(u)$ and encounters $u'$. Thus, $\hat{d}(u, v) \le d(u, u') + d(u', p_2(u')) + \\ d(p_2(u'), p_1(v)) + d(p_1(v), v)$.  
    From the triangle inequality, we know that $d(p_2(u'),p_1(v)) \le h_2(u') + d(u',v) + h_1(v)$. Thus:
    \begin{align*}
        \hat{d}(u,v) &\le d(u,u') + d(u',p_2(u')) + d(p_2(u'),p_1(v)) + d(p_1(v),v) \\
                    &\le d(u,u') + h_2(u') + (h_2(u') + d(u',v) + h_1(v)) + h_1(v) \\
                    &\le d(u,u') + d(u',v) + 2(h_2(u') + h_1(v)) = d(u,v) + 2(h_2(u') + h_1(v))
    \end{align*}
    The last equality follows since $u' \in P$. Next, we bound the term $2(h_2(u') + h_1(v))$. 
    Since $h_2(u') \le 2\lceil d(u',v')/2 \rceil$ it follows that $\hat{d}(u,v) \le d(u,v) + 2(2\lceil d(u',v')/2 \rceil + h_1(v))$.
    Since when going from $u$ to $v$
    on $P$ we first encounter $u'$ and then $v'$ we have that $d(u,v) = d(u,u') + d(u',v') + d(v',v)$. Thus, $d(u',v') = d(u,v) - d(u,u') - d(v',v)$.
    Since the graph is unweighted, we have that $d(v,v')=h_1(v)-1$. Thus, $d(u',v') = d(u,v) - d(u,u') - h_1(v) + 1$. 
    We get:
    \begin{align*}
    \hat{d}(u,v) &\le d(u,v) + 2(h_2(u') + h_1(v)) \le d(u,v) + 2(2\lceil d(u',v')/2 \rceil + h_1(v)) \\
                 &\le d(u,v) + 2(2\lceil (d(u,v) - d(u,u') - h_1(v) + 1)/2 \rceil + h_1(v)) \\
                 &\le d(u,v) + 2(2\lceil d(u,v)/2 - 0/2 - h_1(v)/2 + 1/2\rceil + h_1(v)) \\
                 &\le d(u,v) + 2(2\lceil d(u,v)/2\rceil - h_1(v) + 1 + h_1(v)) = d(u,v) + 4\lceil d(u,v)/2\rceil + 2 = 3d(u,v)+2+2\ODD,
    \end{align*}
    as required.
\end{proof}
\begin{algorithm2e}[t] 
\caption{$\Query(u,v)$} \label{A-Query-5/4-1/2-3}
    $\hat{d}(u,v) \gets \Intersection(u,v,\bigcup_{w\in B_0(u)}{C_0(w)}, \bigcup_{w\in B_0(v)}{C_0(w)})$ \\
    $\hat{d}(u,v) \gets \min(\hat{d}(u,v), \Intersection(u,v,\bigcup_{w\in B_0(u)}{B_1(w)}, \bigcup_{w\in B_0(v)}{B_1(w)}))$ \\
    \ForEach{$w\in B_0(u)$} {
        $\hat{d}(u,v) \gets \min(\hat{d}(u,v), d(u,w) + d(w,p_2(w)) + d(p_2(w),p_1(v)) + d(p_1(v),v))$
    }
    \ForEach{$w\in B_0(v)$} {
        $\hat{d}(u,v) \gets \min(\hat{d}(u,v), d(u,p_1(u)) + d(p_1(u),p_2(w)) + d(p_2(w),w) + d(w,v))$
    }
    \Return $\hat{d}(u,v)$
\end{algorithm2e}

\begin{lemma}\label{L-Q4-3-Time}
 The query algorithm takes $O(n^{1/2})$ time.
\end{lemma}
\begin{proof}
    By Lemma~\ref{Intersection-Runtime} the runtime of $\Intersection(u,v,\bigcup_{w\in B_0(u)}{C_0(w)}, \bigcup_{w\in B_0(v)}{C_0(w)})$ is $\\$ $O(\min(|\bigcup_{w\in B_0(u)}{C_0(w)}|, |\bigcup_{w\in B_0(v)}{C_0(w)}|))$ and the runtime of $\Intersection(u,v,\bigcup_{w\in B_0(u)}{B_1(w)},$\\$ \bigcup_{w\in B_0(v)}{B_1(w)}))$ is $O(\min(|\bigcup_{w\in B_0(u)}{B_1(w)}|,$ $|\bigcup_{w\in B_0(v)}{B_1(w)}|))$. 
    Since $|B_0(u)|=O(n^{1/4})$, $|C_0(w)|=O(n^{1/4})$ and $|B_1(w)|=O(n^{1/4})$ it follows that 
    $O(\min(|\bigcup_{w\in B_0(u)}{C_0(w)}|,|\bigcup_{w\in B_0(v)}{C_0(w)}|))=$\\$O(\min(|\bigcup_{w\in B_0(u)}{B_1(w)}|,|\bigcup_{w\in B_0(v)}{B_1(w)}|)) = O(n^{2\cdot 1/4}=O(n^{1/2})$, and 
    the cost of these two steps is $O(n^{1/2})$.
    
    Computing $d(u,w) + d(w,p_2(w)) + d(p_2(w),p_1(v)) + d(p_1(v),v)$ for every $w\in B_0(u)$ takes $O(|B_0(u)|)=O(n^{1/4})$ time.
    Computing $d(u,p_1(u)) + d(p_1(u),p_2(w)) + d(p_2(w),w) + d(w,v)$ for every $w\in B_0(v)$ takes $O(|B_0(v)|)=O(n^{1/4})$ time.
    We conclude that the running time is $O(n^{2/4})$.
\end{proof}

Theorem~\ref{T-DO-3-2-Unweighted} follows from Lemma~\ref{L-Distance-Oracle-6-Space-Size}, Lemma~\ref{L-Bound-4-3} and Lemma~\ref{L-Q4-3-Time}

\subsection{$(2k-5,4+2\ODD)$-approximation in unweighted graphs}\label{S-2k-5+6-unweighted}
In this section, we use both the \textit{middle vertex technique} and the \textit{borderline vertices technique} to achieve the following results:

\Reminder{T-DO-2k-5-unweighted}
By applying $k=5$, we get the following distance oracle.
\begin{corollary} \label{C-DO-5-unweighted}
    There is an $\Ot(n^{6/5})$-space distance oracle that given two query vertices $u,v\in V$ computes in  $O(n^{1/5})$-time 
    a distance estimation $\hat{d}(u,v)$ that satisfies
    $d(u,v)\le \hat{d}(u,v) \le 5d(u,v) + 4+2\ODD$. The distance oracle is constructed in $\Ot(mn^{\frac{3}{5}})$ expected time. 
\end{corollary}
\subsubsection{Storage.}
We construct a hierarchy $A_0,A_1,\dots,A_k$, where $A_0=V$, $A_{k}=\emptyset$.
The set $A_1$ is constructed using Lemma~\ref{L-A-center}, with parameter $n^{-1/k}$. 
For every $1 < i < k$, the set $A_i$ is constructed using Lemma~\ref{L-TZ-Size}. 
Our distance oracle stores $B_i(u)$ and $p_i(u)$, for every $u\in V$ and $c \le i \leq k-1$. We also save $d(u,v)$ for every $\pair{u,v} \in A_{k-3} \times A_{2}$. We also save $C(w)$, for every $w\in V\setminus A_1$.
\begin{lemma}\label{L-Space-2k-5-UnWeighted}
    The distance oracle is constructed in $O(mn^{\frac{3}{k}})$-time and uses $\Ot(n^{1+1/k})$-space.
\end{lemma}
\begin{proof}
    From Lemma~\ref{L-TZ-Size} it follows that  $|B_i(u)|=O(n^{1/k})$, for every $1 < i< k-1$. From Lemma~\ref{L-A-center} it follows that $|A_1|=\Ot(n^{1-1/k})$, thus $|A_{k-1}|=|A_1|\cdot \left(n^{-1/k}\right)^{k-2}=\Ot(n^{1-1/k}\cdot n^{-(k-2)/k})=\Ot(n^{1/k})$.
    Since $A_{k}=\emptyset$, for every $u\in V$, we have that $B_{k-1}(u)=A_{k-1}$ and $|B_{k-1}(u)|=\Ot(n^{1/k})$.
    The cost of saving $p_i(u)$, for every $u\in V$, and $c\leq i \leq k-1$, is $O(n)$. 
    The cost of saving $d(u,v)$ for every $\pair{u,v}\in A_{k-3}\times A_{2}$ is $|A_{k-3}| \cdot |A_{2}|=\Ot(n^{3/k})\cdot \Ot(n^{1-2/k})=\Ot(n^{1+1/k})$.
    Therefore, we conclude that the total space is $\Ot(n^{1+1/k})$. 
    
    From Lemma~\ref{L-A-center} it follows that computing $A_1$ takes $\Ot(mn^{1/k})$ time. 
    From Lemma~\ref{L-TZ-Size} it follows that computing $A_i$ and $B_i(u)$ for every $1 < i\leq k-1$ takes $\Ot(mn^{1/k})$ time. 
    Computing $p_i(u)$, for every $u\in V$, and  $0\leq i \leq k-1$, takes $O(km)$ time.
    Computing $d(u,v)$ for every $\pair{u,v}\in A_{k-3}\times A_{2}$ takes $O(m\cdot |A_{k-c-2}|)=O(mn^{\frac{3}{k}})$ time. 
    Therefore, we conclude that the total construction time is $O(mn^{\frac{3}{k}})$. 
\end{proof}

\subsubsection{Query algorithm.}
The input to the query algorithm is two vertices $u,v\in V$.
The query output is an estimation $\hat{d}(u,v)$.
The algorithm first sets $\hat{d}(u,v)$ to be the minimum between \\$\Intersection(u,v,C(u), C(v))$ and $\Intersection(u,v,B_1(u), B_1(v))$.\;\;
For every $u'\in C(u)$ 
we set $\hat{d}(u,v)$ to $\min(\hat{d}(u,v), d(u,u')+\MTZQuery(u',v))$.
Similarly, for every $v'\in C(v,A_{c})$ 
we set $\hat{d}(u,v)$ to $\min(\hat{d}(u,v), d(v,v')+\MTZQuery(v',u))$.
Next, we set $\hat{d}(u,v)$ to  $\min(\hat{d}(u,v), h_{2}(u)+d(p_{2}(u),p_{k-3}(v))+h_{k-3}(v), h_{k-3}(u)+d(p_{k-3}(u),p_{2}(v))+h_{2}(v))$.
Finally, the algorithm returns $\hat{d}(u,v)$ as its estimation.
A pseudo-code for the query algorithm is given in Algorithm~\ref{Algorithm-Distance-Oracle-2k-5-Unweighted}.
\begin{algorithm2e}[t] 
\caption{$\Query(u, v)$} \label{Algorithm-Distance-Oracle-2k-5-Unweighted}
$\hat{d}(u,v) \gets \Intersection(u,v,C(u), C(v))$ \\ 
% $\hat{d}(u,v) \gets \min(\hat{d}(u,v), \Intersection(u,v,B_0(u), B_0(v))$ \\ 
$\hat{d}(u,v) \gets \min(\hat{d}(u,v), \Intersection(u,v,B_1(u), B_1(v))$ \\
\lForEach{$u'\in C(u)$} {
    $\hat{d}(u,v)\gets \min(\hat{d}(u,v), d(u,u')+\MTZQuery(u',v))$
}
\lForEach{$v'\in C(v)$} {
    $\hat{d}(u,v)\gets \min(\hat{d}(u,v), d(v,v')+\MTZQuery(v',u))$
}
$\hat{d}(u,v) \gets \min(\hat{d}(u,v), h_{2}(u)+d(p_{2}(u),p_{k-3}(v))+h_{k-3}(v), h_{k-3}(u)+d(p_{k-3}(u),p_{2}(v))+h_{2}(v))$\\
\Return $\hat{d}(u,v)$
\end{algorithm2e}
Next, we bound $\hat{d}(u,v)$. 
\begin{lemma}\label{L-Approximation-2k-5-UnWeighted}
    $\hat{d}(u,v) \le (2k-5)d(u,v)+4+2\ODD$
\end{lemma}
\begin{proof}
    Let $P$ be a shortest path between $u$ and $v$. Recall that $\tau=\tau(u,v,P)$, is the middle vertex between $u$ and $v$ in $P$, and let $\tau_u(P)$ and $\tau_v(P)$ be the borderline vertices.
    If $h_1(\tau) > d(u,v)/2+0.5\ODD$  then it follows from Lemma~\ref{L-h-rho-big} that $\tau\in C(u)\cap C(v)$.
    From Property~\ref{P-Intersection-With-P-Returns-d(u,v)} we get that $\Intersection(u,v,C(u),C(v))=d(u,v)$ and the claim holds.

    Otherwise, $h_1(\tau) \leq  d(u,v)/2+0.5\ODD$. 
    We divide the rest of the proof into two cases. 
    The case that $p_1(\tau) \in B_1(u)\cap B_1(v)$ and the case that $p_1(\tau) \notin B_1(u)\cap B_1(v)$.

    If $p_1(\tau) \in B_1(u)\cap B_1(v)$ then $\Intersection(u,v,B_1(u),B_1(v))$ returns a value bounded by $d(u,p_1(\tau))+d(p_1(\tau),v)$. 
    From Lemma~\ref{L-tau-In-B1uB1v} it follows that $d(u,p_1(\tau))+d(p_1(\tau),v)\leq d(u,v)+2h_1(\tau)$. 
    Since $h_1(\tau) \leq d(u,v)/2+0.5\ODD$ we get that $d(u,v)+2h_1(\tau)\le d(u,v) + 2(d(u,v)/2+0.5\ODD)=2d(u,v)+1\ODD$, 
    and $\hat{d}(u,v)\leq d(u,v) + 2(d(u,v)/2+0.5\ODD)=2d(u,v)+1\ODD$, as required.

    Otherwise, if $p_1(\tau) \notin B_1(u)\cap B_1(v)$. It follows from Lemma~\ref{L-tau-Not-In-B1uB1v}
    that either $h_2(u) \leq d(u,v)/2+0.5\ODD+h_1(\tau)$ or $h_2(v) \leq d(u,v)/2+0.5\ODD+h_1(\tau)$. Wlog, we assume that $h_2(u) \leq  d(u,v)/2+0.5\ODD+h_1(\tau)$.

    Next, we bound $h_{2}(v)$ using the following two claims:
    \begin{claim}\label{C-2k-1-4c-H1-UW}
        If $h_1(\tau'_u) > d(\tau'_u,v)$ then $\hat{d}(u,v) = d(u,v)$
    \end{claim}
    \begin{proof}
        If $h_1(\tau'_u) > d(\tau'_u, v)$ then we know that $v\in B(\tau'_u)$ and therefore $\tau'_u\in C(v)$. By definition, $\tau'_u\in C(u)$. Therefore, we get that $\tau'_u \in C(v) \cap C(u)$ and from Lemma~\ref{P-Intersection-With-P-Returns-d(u,v)} we get that\\ $\hat{d}(u,v) \le \Intersection(u,v,C^*(u,A_c), C^*(v,A_c)) = d(u,v)$, as required.
    \end{proof}
    
    \begin{claim}\label{C-2k-1-4c-H-UW}
        Either $h_{2}(v) \le d(u,v)+2$ or $\hat{d}(u,v)\leq 3d(u,v)$.
    \end{claim}
    \begin{proof}
    We divide the proof into two cases. The case that $p_c(\tau'_u)\in B_{1}(v)$ and the case that $p_c(\tau'_u)\not\in B_{1}(v)$.
    Consider the case that $p_1(\tau'_u)\in B_{1}(v)$. 
    In the query algorithm, we encounter $\tau'_u$ because
    $\tau'_u\in C(u)$. 
    Since $p_1(\tau'_u) \in B_1(v)$, we have that $\MTZQuery(\tau'_u,v) \le 2h_1(\tau'_u) + d(\tau'_u,v)$. From Claim~\ref{C-2k-1-4c-H1} it follows that either $h_1(\tau'_u) \le d(\tau'_u,v)$ or $\hat{d}(u,v)\le d(u,v)$ and the claim holds. Thus, we know that $h_1(\tau'_u) \le d(\tau'_u,v)$.
    Since $h_1(\tau'_u) \le d(\tau'_u,v)$
    we get that $\MTZQuery(\tau'_u,v) \le 2h_1(\tau'_u) + d(\tau'_u,v) \le 3d(\tau'_u,v) \le 3d(u,v)$, where the last inequality follows from the fact that $\tau'_u\in P$.
    
    Consider now the case that $p_1(\tau'_u)\not\in B_{c}(v)$. From Lemma~\ref{L-bound-h_{c+1}(v)-tau_u-+2} it follows that $h_{2}(v) \le d(u,v)+2$.
    \end{proof}

    From Claim~\ref{C-2k-1-4c-H-UW} it follows that either $h_{2}(v) \le d(u,v)+2$ or $\hat{d}(u,v)\le 3d(u,v)$. Therefore, it follows that $h_{2}(v) \le d(u,v)+2$ or the lemma holds. Thus, we can assume that $h_{2}(v) \le d(u,v)+2$. Recall that we assume that $h_2(u) \le d(u,v)+1\ODD$.
    Next, we show that either the lemma holds or that $\min(h_{k-3}(u), h_{k-3}(v)) \le (2k-4)d(u,v)$.
    By applying Property~\ref{P-ADO.Query-Correctness} with $i=k-3$ and $j=2$, with the fact that $\min(h_{2}(u),h_{2}(v))\le h_{2}(u) \le d(u,v)+1\ODD$ we get that either 
    \[
    \min(h_{k-3}(u), h_{k-3}(v)) \le \min(h_{2}(u), h_{2}(v))+(k-5)d(u,v) \le (k-5)d(u,v)+d(u,v)+1\ODD=(k-4)d(u,v)+1\ODD \]
    or
    \[
        \hat{d}(u,v) \le 2(k-6)d(u,v)+2\min(h_{2}(u), h_{2}(v)) \le 2(k-6)d(u,v) + d(u,v)+1\ODD \le 2(k-5)d(u,v)+1\ODD.
    \]
    Therefore, we get that either $\min(h_{k-3}(u), h_{k-3}(v)) \leq (k-2-2c)d(u,v)$ or the lemma holds. Thus, we can assume that $\min(h_{k-3}(u), h_{k-3}(v)) \leq (k-5)d(u,v)+d(u,v)+1\ODD=(k-4)d(u,v)+1\ODD$.

    In the query, we have that $\hat{d}(u,v) \le \min(h_{2}(u)+d(p_{2}(u),p_{k-3}(v))+h_{k-3}(v), h_{k-3}(u)+d(p_{k-3}(u),p_{2}(v))+h_{2}(v))$. From the triangle inequality, it follows that 
    $\hat{d}(u,v) \le d(u,v) + 2\min(h_{k-3}(u)+h_{2}(v), h_{2}(u) + h_{k-3}(v))$.
    Thus, we get:
    \begin{align*}
    \hat{d}(u,v) &\le \min(h_{2}(u)+d(p_{2}(u),p_{k-3}(v))+h_{k-3}(v), h_{k-3}(u)+d(p_{k-3}(u),p_{2}(v))+h_{2}(v))\\
    &\le d(u,v) + 2\min(h_{k-3}(u)+h_{2}(v), h_{2}(u) + h_{k-3}(v)) \\
    &\le d(u,v) + 2\min(h_{k-3}(u),h_{k-3}(v)) + 2\max(h_{2}(u),h_{2}(v)) \\
    &\le d(u,v) +2((k-5)d(u,v)+d(u,v)+1\ODD) + 2(d(u,v)+2)=(2k-5)d(u,v)+4+2\ODD,
    \end{align*}
    where the last inequality follows from the fact that we assume that $\min(h_{k-3}(u), h_{k-3}(v)) \leq (k-5)d(u,v)+d(u,v)+1\ODD$ and $\max(h_{2}(u),h_{2}(v))\le d(u,v)+2$.
\end{proof}

\begin{lemma}\label{L-Q2k-5-Time}
 The query algorithm takes $\Ot(n^{1/k})$ time.
\end{lemma}
\begin{proof}
    By Lemma~\ref{Intersection-Runtime} the runtime of $\Intersection(u,v,C(u),C(v))$  is $O(\min(|C(u)|,|C(v)|))$ 
    and the runtime of $\Intersection(u,v,B_1(u),B_1(v))$ is $O(\min(|B_1(u)|,|B_1(v)|))$. 
    Since $\min(|C(u)|,|C(v)|)=O(n^{1/k})$ and $\min(|B_1(u)|,|B_1(v)|)=O(n^{1/k})$, we get that the cost of these two steps is $O(n^{1/k})$. 
    Computing $\MTZQuery(u,v)$ takes $O(k)$ time, therefore computing $\MTZQuery(v,u')$, for every $u'\in C(u)$, and $\MTZQuery(v',u)$, for every $v'\in C(v)$, takes $O(kn^{1/k})$.
    Computing $h_{2}(u)+d(p_{2}(u),p_{k-3}(v))+h_{k-3}(v), h_{k-3}(u)+d(p_{k-3}(u),p_{2}(v))+h_{2}(v)$ takes $O(1)$ time.
    We conclude that the running time is $\Ot(n^{1/k})$.
\end{proof}

\autoref{T-DO-2k-5-unweighted} follows from~\autoref{L-Space-2k-5-UnWeighted}, ~\autoref{L-Approximation-2k-5-UnWeighted}, and~\autoref{L-Q2k-5-Time}.
% Due to the lack of space, we prove~\autoref{T-DO-3-2-Unweighted} and~\autoref{T-DO-2k-5-unweighted} in~\autoref{S-unweighted-appendix}.

\section{Distance oracles with stretch at most \texorpdfstring{$2$}{2} in unweighted graphs}\label{S-at-most-2}
In this section, we introduce the sets $S_t(\cdot)$. 
Let $A\subseteq V$, let $u,v\in V$ and let $P=P(u,v)$. We compute a set $S_t(u)$ of vertices around $u$, which takes into consideration the set $A$. The set $S_t(u)$  can be viewed as a recursive generalization of $B(u, V, A)$. We use these sets to bound for some vertex $x\in P$ the value of $h(x)$ by roughly $d(u,v)/2t$. 
We then present several new distance oracles that compute $S_t(u)$ and $S_t(v)$ in the query algorithm. 

\subsection{The set \texorpdfstring{$S_t(\cdot)$}{St(.)}}\label{S-5-1}
Let $A\subseteq V$. Recall that $B(u)=B(u,V,A)$, for every $u\in V$. 
Let $t\geq 0$ be an integer. Let $S_t(u)=\bigcup_{w\in S_{t-1}(u)}{B(u)}$, where $S_0(u)=\{u\}$.
We define a distance function $d_t(u,v)$ between $u$ and every $v\in S_t(u)$. Let $d_0(u,u)=0$ and let $d_t(u,v)$  be  the minimum value of  $d_{t-1}(u,x)+d(x,v)$, where $x\in S_{t-1}(u)$ and $v\in B(x)$, that is, $d_t(u,v)=\min_{x\in S_{t-1}(u)\cap C(v)} (d_{t-1}(u,x)+d(x,v))$.
We compute the set $S_t(u)$ and the distance function $d_t(u,\cdot)$ using  a simple recursive procedure $\GenerateSi$, presented in Algorithm~\ref{Algorithm-GenerateSi}. 
\begin{algorithm2e}[t] 
\caption{$\GenerateSi(u, t)$}\label{Algorithm-GenerateSi}
\lIf{$t=0$}{\Return $\{u\}$}
$\langle S_{t-1}(u) , d_{t-1}\rangle \gets \GenerateSi(u,t-1)$\\
\ForEach{$x\in S_{t-1}(u)$}
{
    \ForEach{$v\in B(x)$}
    {
    $d_t(u,v) = \min ( d_t(u,v), d_{t-1}(u,x)+d(x,v))$
    }
}
\Return $\langle \bigcup_{w\in S_{t-1}(u)}{B(w)}, d_t\rangle$
\end{algorithm2e}
Next, let $0<c\leq 1$ be the minimum value for which $|B(w)| \le O(n^{c})$, for every $w\in V$. We show:
\begin{lemma}\label{L-size-and-running-time-of-S_i}
     $\GenerateSi(u,t)$ computes in $O(n^{ct})$ time the set $S_t(u)$ of size $O(n^{ct})$.
\end{lemma}
\begin{proof}
    We  prove the claim by induction on $i\in [0,t]$. For $i=0$, the claim holds since $S_0(u)=\{u\}$, and  $|S_0(u)|=1$.
    Next, we assume that $|S_{i-1}(u)| = O(n^{c(i-1)})$ and that  $S_{i-1}(u)$ was computed in $O(n^{c(i-1)})$ time.
    In order to create $S_i(u)$, the algorithm $\GenerateSi(u, i)$ iterates over every $w\in S_{i-1}(u)$
    and adds $B(w)$ to $S_i(u)$ in $O(|B(w)|)=O(n^c)$ time. 
    Since $|S_{i-1}(u)| = O(n^{c(i-1)})$ the total time for computing $S_i(u)$ is $O(n^c \cdot n^{c(i-1)})=O(n^{ci})$.  
    From the same arguments we have $|S_{i}(u)| = O(n^{ci})$. 
\end{proof}

As before let $u,v\in V$ and let $P=P(u,v)$. 
Let $u=u_0$ and let $v=v_0$.
For every $1\leq i \leq t$, let $u_i \in B(u_{i-1})$ be the farthest vertex in $P$ from $u$, and let $v_i\in B(v_{i-1})$ be the farthest vertex in $P$ from $v$. Let $P_t(u)=P(u,u_t)$, that is, the portion on $P$ between $u$ and $u_t$. Notice that every vertex in $P_t(u)$ is in a ball of some vertex from $S_{t-1}(u)$. Therefore, by the definition of $S_{t}(u)$ we have that $P_t(u)\subseteq S_{t}(u)$. 

In the next  lemma we show that for the vertices of $P_t(u)$ the distance function $d_t$ is the same as the distance function $d$ of $G$. 

\begin{lemma}\label{L-Pt-Dt}
If $v\in P_t(u)$ then $d_t(u,v)=d(u,v)$.
\end{lemma}
\begin{proof}
We prove by induction on $t$. 
For the base case of $t=0$ the claim holds since $P_0(u)=\{u\}$ and $d_0(u,u)=0$. 
We assume the claim holds for every $0< i<t$ and prove the claim for $i=t$. 
Let $v\in P_t(u)$. This implies that $v\in P(u_{i-1},u_i)$, where $1\leq i\leq t$. 
Since $u_i\in P$ is the farthest vertex from $u$ in $B(u_{i-1})$ and since $v\in P(u_{i-1},u_i)$, a shortest path between $u_{i-1}$ and $u_i$, it follows that $v\in B(u_{i-1})$.
By definition of $d_t$ we have $d_t(u,v)\leq d_{t-1}(u,u_{i-1})+d(u_{i-1},v)$, since $v\in B(u_{i-1})$. 
By the induction assumption $d_{t-1}(u,u_{i-1})=d(u,u_{i-1})$, therefore, $d_t(u,v)\leq  d_{t-1}(u,u_{i-1})+d(u_{i-1},v) = d(u,u_{i-1})+d(u_{i-1},v)=d(u,v)$, where the last equality holds since $u_{i-1}\in P(u,v)$. 
\end{proof}

In the next lemmas we examine the case that $P_t(u) \cap P_t(v)  = \emptyset$. This case and the bounds proven in these lemmas are illustrated in Figure~\ref{F-S}(a).

\begin{lemma}
\label{L-No-S-Intersect}
    If $P_t(u) \cap P_t(v) = \emptyset$ then $h(u_{t-1}) + h(v_{t-1}) \leq d(u_{t-1},v_{t-1})+1$
\end{lemma}
\begin{proof}
    From the definition of $P_t$, it follows that $B(u_{t-1}) \cap P \subseteq P_t(u)$. Similarly, $B(v_{t-1}) \cap P \subseteq P_t(v)$.
    Thus, $(B(u_{t-1}) \cap P) \cap (B(v_{t-1}) \cap P)= B(u_{t-1}) \cap B(v_{t-1}) \cap P \subseteq P_t(u)\cap P_t(v) = \emptyset$. 
    Since $B(u_{t-1}) \cap B(v_{t-1}) \cap P \subseteq \emptyset$ it follows that $B(u_{t-1}) \cap B(v_{t-1}) \cap P = \emptyset$ and by
    applying Lemma~\ref{L-Intersect} we get $B(u_{t-1}) \cap B(v_{t-1})=\emptyset$. Now, by using Lemma~\ref{L-NO-Intersect} we get  $h(u_{t-1}) + h(v_{t-1}) \le d(u_{t-1},v_{t-1}) + 1$, as required.
\end{proof}

\begin{lemma}
\label{L-Bound-Sum-Of-h-in-S}
    If $P_t(u) \cap P_t(v) = \emptyset$ then $\sum_{i=0}^{t-1}h(u_i) + h(v_i) \le d(u,v)+2t-1$
\end{lemma}
\begin{proof}
    For every $0\le i \le t-1$, since $u_{i-1},u_i\in P$,  $u_i$ is the farthest vertex from $u$ in $B(u_{i-1})$, and since the graph is unweighted, it follows that $d(u_{i-1}, u_i)=h(u_{i-1})-1$. Similarly, $d(v_{i-1},v_i) = h(v_{i-1})-1$.

    From the assumption that $P_t(u) \cap P_t(v) = \emptyset$ it follows that when going from $u$ to $v$
    on $P$ we first encounter all $u_i$'s and then all $v_i$'s. Thus:
    \begin{align*} 
d(u,v)    &= d(u,u_1) + \dots +d(u_{t-2},u_{t-1})+ d(u_{t-1},v_{t-1}) + d(v_{t-1},v_{t-2})+\dots + d(v_1,v) \\ 
          &=  d(u_{t-1}, v_{t-1}) + \sum_{i=0}^{t-2} ( d(u_i,u_{i+1}) + d(v_i,v_{i+1}) )
\end{align*}

       Since $P_t(u) \cap P_t(v) = \emptyset$ it follows from Lemma~\ref{L-No-S-Intersect} that $h(u_{t-1}) + h(v_{t-1}) \leq d(u_{t-1},v_{t-1})+1 $. Recall that $d(u_{i-1},u_i) \geq h(u_{i-1})-1$ and $d(v_{i-1},v_i) \geq h(v_{i-1})-1$, for every $0\leq i \leq t$.
    Therefore, we get:
   \begin{align*} 
d(u,v)    &= d(u_{t-1}, v_{t-1}) + \sum_{i=0}^{t-2} d(u_i,u_{i+1}) + d(v_i,v_{i+1})  
          \geq h(u_{t-1}) + h(v_{t-1}) - 1 + \sum_{i=0}^{t-2}h(u_i)-1 + h(v_i)-1\\ 
          &= -2(t-1) -1 + \sum_{i=0}^{t-1}h(u_i) + h(v_i)= -2t+1 + \sum_{i=0}^{t-1}h(u_i) + h(v_i).
\end{align*}
    Adding $2t-1$ to both sides we get $\sum_{i=0}^{t-1}{h(u_i) + h(v_i)} \le d(u,v) + 2t - 1$, as required.
\end{proof}

\subsection{Almost \texorpdfstring{$1+1/t$}{1+1/t}-stretch with \texorpdfstring{$n^{2-c}$}{n\^(2-c)}-space and \texorpdfstring{$n^{ct}$}{n\^(ct)}-query}\label{S-5-t}

In this section we use the sets $S_t(\cdot)$ in the query algorithm  and prove the following: 

\Reminder{T-Unweighted-1+1/t}

\subsubsection{Storage.}
We compute a set  $A$ of size $O(n^{1-c})$ using Lemma~\ref{L-TZ-Size}.
For every $u\in V$ we save  $B(u,V,A)$. For every $u\in A$ and  $w\in V$ we save $d(u,w)$. In addition, for every  $u\in V$ we save $p(u,A)$. As before, we refer to $B(u,V,A)$ as $B(u)$, and to $p(u,A)$ as $p(u)$.

\begin{lemma}\label{L-Space-T-Unweighted-1+1/t}
    The distance oracle is constructed in $O(mn^{1-c})$-time and  uses  $O(n^{2-c})$-space.
\end{lemma}
\begin{proof}
From Lemma~\ref{L-TZ-Size} it follows that  $|B(u,V,A)|=O(n^c)$.  
In addition, for every $w\in A$ and $v\in V$, we store $d(w,v)$ at a cost of $n|A|=O(n^{2-c})$ space. 
The total space is $O(n^{2-c} + n^{1+c})$. Since  
$0 < c < 1/2$ we have that $2-c>1+c$ and the space is $O(n^{2-c})$.
From Lemma~\ref{L-TZ-Size} it follows that the set $A$ and $B(u,V,A)$, for every $u\in V$, are computed in $O(mn^c)$ time. 
Computing $p(u)$, for every $u\in V$ takes $O(m)$ time.
Computing $d(w,v)$ for every $w\in A$ and $v\in V$ takes $O(m|A|)=O(mn^{1-c})$ time.
Since $1-c>c$ we get that the running time  is $O(mn^{1-c})$. 
\end{proof}

\subsubsection{Query algorithm.}
The input to the query algorithm is two vertices $u,v\in V$.
The output of the query is an estimation $\hat{d}(u,v)$.
The algorithm first creates the sets $S_t(u)$ and $S_t(v)$ using $\GenerateSi$.  
Then, the algorithm sets $\hat{d}(u,v)$ to be the value returned by $\Intersection(u,v,S_t(u), S_t(v))$.
Next, for every vertex $w\in S_t(u) \cup S_t(v)$ the algorithm sets $\hat{d}(u,v)$ to be $\min (\hat{d}(u,v),d(u, p(w)) + d(p(w), v))$.
Finally, the algorithm returns $\hat{d}(u,v)$ as the estimated distance.
A pseudo-code for the query algorithm is given in Algorithm~\ref{Algorithm-Distance-Oracle-1+1/t}.

\begin{algorithm2e}[t] 
\caption{$\Query(u,v,t)$} \label{Algorithm-Distance-Oracle-1+1/t}
$S_t(u) \gets \GenerateSi(u,t)$\\
$S_t(v) \gets \GenerateSi(v,t)$\\
$\hat{d}(u,v) \gets \Intersection(u,v, S_t(u), S_t(v))$ \\
\ForEach{$w\in S_t(u)\cup S_t(v)$ } {
    $\hat{d}(u,v) \gets \min(\hat{d}(u,v), d(u,p(w)) + d(p(w), v))$
}
\Return $\hat{d}$
\end{algorithm2e}

Next, we bound  $\hat{d}(u,v)$. 
\begin{lemma}\label{L-S-bound-1/t}
    $\hat{d}(u,v) \le d(u,v) + 2\lceil d(u,v) / 2t \rceil$.
\end{lemma}
\begin{proof}
    Let $P$ be a shortest path from $u$ to $v$. 
    Let $u=u_0$ and let $v=v_0$. For every $i\in [1,t]$, let $u_i\in B(u_{i-1})$ be the farthest vertex from $u$ on the path $P$ and 
    let $v_i\in B(v_{i-1})$ be the farthest vertex from $v$ on the path $P$. 
    Let $P_i(u)=P(u,u_i)$, that is, the portion on $P$ between $u$ and $u_i$. 
    
    We divide the proof into two cases. The case that $P_t(u) \cap P_t(v) \neq \emptyset$ and the case 
    that $P_t(u) \cap P_t(v) = \emptyset$.
    Consider the case that $P_t(u) \cap P_t(v) \neq \emptyset$.  
    Let $x\in P_t(u) \cap P_t(v)$. From Lemma~\ref{L-Pt-Dt} it follows that $d_t(u,x)=d(u,x)$ and $d_t(v,x)=d(v,x)$.
    Since $P_t(u)\subseteq S_t(u)$ and $P_t(v)\subseteq S_t(v)$ the value returned by the call to  $\Intersection(u,v,S_t(u),S_t(v))$ is at most 
    $d_t(u,x)+d_t(v,x)=d(u,x)+d(x,v)=d(u,v)$, where the last equality follows from the fact that $x\in P$, and the claim holds. 
     
    Consider now the case that $P_t(u) \cap P_t(v) = \emptyset$, from Lemma~\ref{L-Bound-Sum-Of-h-in-S} we know that $\sum_{i=0}^{t-1}h(u_i) + h(v_i) \le d(u,v)+2t-1$.
    
    Let $w\in \bigcup_{i=0}^{t-1}\{u_i,v_i\}$ be the vertex with minimal $h(w)$ value.
    % Let $w=\arg\min_{y\in \bigcup_{i\in[t-1]}\{u_i,y^i\}}(h(y))$.
    By the the minimality of $w$ we get 
    $h(w) \le (\sum_{i=0}^{t-1}h(u_i) + h(v_i)) /2t \le (d(u,v) + 2t - 1)/2t$.   
    Since $h(w)$ is an integer it follows that $h(w) \le \lfloor (d(u,v) + 2t - 1)/2t \rfloor = \lceil d(u,v)/2t \rceil $.

The algorithm iterates over  $S_t(u)\cup S_t(v)$ and encounters $w$. Thus, $\hat{d}(u,v) \le d(u, p(w)) + d(p(w),v)$. 
By the triangle inequality $\hat{d}(u,v) \le d(u, p(w)) + d(p(w),v) \le d(u, w) + d(w,p(w)) + d(p(w),w) + d(w,v)=d(u,v) + 2h(w)$. 
    Since $h(w) \le \lceil d(u,v)/2t \rceil$, it follows that $\hat{d}(u,v) \le d(u,v) + 2h(w) \le d(u,v)+2\lceil d(u,v)/2t \rceil$, as required.
\end{proof}

\begin{lemma}\label{L-Time-1/t}
    The query algorithm takes $O(n^{ct})$ time.
\end{lemma}
\begin{proof}
    From Lemma~\ref{L-TZ-Size} we have $|B(w)| \le O(n^c)$, for every $w\in V$. Thus, it follows from Lemma~\ref{L-size-and-running-time-of-S_i} that the call to $\GenerateSi(u,t)$ and the call to $\GenerateSi(v,t)$ take $O(n^{ct})$. It also follows from Lemma~\ref{L-size-and-running-time-of-S_i}
    that  $|S_t(u)| = O(n^{ct})$ and $|S_t(v)| = O(n^{ct})$.
    Therefore, the cost of computing $d(u,p(w)) + d(p(w), v)$, for every $w\in S_t(u)\cup S_t(v)$, is $O(|S_t(u)\cup S_t(v)|)=O(n^{ct})$ time.
\end{proof}

Theorem~\ref{T-Unweighted-1+1/t} follows from Lemma~\ref{L-Space-T-Unweighted-1+1/t}, Lemma~\ref{L-S-bound-1/t} and Lemma~\ref{L-Time-1/t}.

\subsection{Almost \texorpdfstring{$1+2/t$}{1+2/t}-stretch with \texorpdfstring{$n^{2-2c}$}{n\^(2-2c)}-space and \texorpdfstring{$n^{tc}$}{n\^(tc)}-query}\label{S-5-3}

The possible distance oracles that follow from the trade-off presented in Theorem~\ref{T-Unweighted-1+1/t}, use $\Omega(n^{1.5})$-space. In the next theorem ,we use $S_t(\cdot)$ and present a distance oracle that uses $o(n^{1.5})$-space.

\subsubsection{Storage.}

We compute a set  $A$ of size $O(n^{1-c})$ using Lemma~\ref{L-TZ-Size}.
For every $u\in V$ we save  $B(u,V,A)$. For every $u,w\in A$ we save $d(u,w)$. In addition, for every  $u\in V$ we save $p(u,A)$. As before, we refer to $B(u,V,A)$ as $B(u)$, and to $p(u,A)$ as $p(u)$.

\begin{lemma}\label{L-Space-Unweighted-5/3}
    The distance oracle is constructed in $O(mn^{1-c})$-time and uses $O(n^{2-2c})$-space.
\end{lemma}
\begin{proof}
From Lemma~\ref{L-TZ-Size} it follows that  $|B(u,V,A)|=O(n^c)$.  
For every $v,w\in A$, we store $d(w,v)$ at a cost of $|A|^2=O(n^{2-2c})$ space. 
The total space is $O(n^{2-2c} + n^{1+c})$. Since  
$0 < c < 1/3$ we have $2-2c>1+c$ and the space is $O(n^{2-2c})$.
From Lemma~\ref{L-TZ-Size} it follows that the set $A$ and $B(u,V,A)$, for every $u\in V$, are computed in $O(mn^c)$ time. 
Computing $p(u)$, for every $u\in V$ takes $O(m)$ time.
Computing $d(w,v)$ for every $w\in A$ and $v\in V$ takes $O(m|A|)=O(mn^{1-c})$ time.
Since $1-c>c$ we get that the running time  is $O(mn^{1-c})$. 
\end{proof}

\subsubsection{Query algorithm.}
The input to the query algorithm is two vertices $u,v\in V$.
The output of the query is an estimation $\hat{d}(u,v)$.
The algorithm first creates the sets $S_i(u)$ and $S_i(v)$, for every $0\leq i\leq t$, using $\GenerateSi$.  
Then, the algorithm sets $\hat{d}(u,v)$ to be the value returned by $\Intersection(u,v,S_t(u), S_t(v))$.

Next, for every $0\leq i \leq t-1\}$ and  
for every $\langle u',v'\rangle \in S_i(u)\times S_{t-i-1}(v)$ the algorithm sets $\hat{d}(u,v)$ to be $\min(\hat{d}(u,v), d(u,u') + d(u',p(u')) + d(p(u'), p(v')) + d(p(v'), v') + d(v',v))$.
Finally, the algorithm returns $\hat{d}(u,v)$ as the estimated distance.
A pseudo-code for the query algorithm is given in Algorithm~\ref{Algorithm-Distance-Oracle-5/3}.
\begin{algorithm2e}[t] 
\caption{$\Query(u,v)$} \label{Algorithm-Distance-Oracle-5/3}
\lFor{$i\gets 0$ to $t$}{
    $S_i(u) \gets \GenerateSi(u,i)$, $S_i(v) \gets \GenerateSi(v,i)$
}
$\hat{d}(u,v) \gets \Intersection(u,v, S_t(u), S_t(v))$\\ 
\For{$i\gets 0$ to $t-1$} 
{
    \ForEach {$\langle u',v'\rangle \in S_i(u)\times S_{t-i-1}(v)$} 
    {
        $\hat{d}(u,v) \gets \min(\hat{d}(u,v), d(u, u') +d(u',p(u')) + d(p(u'),p(v')) + d(p(v'),v') + d(v',v))$
    }
}
\Return $\hat{d}$
\end{algorithm2e}
Next, we bound  $\hat{d}(u,v)$. 
\begin{lemma}\label{L-Bound-5/3-UW}
    $\hat{d}(u,v) \le d(u,v)  + 2\lceil d(u,v) / t \rceil + 2$.
\end{lemma}
\begin{proof}
    Let $P$ be a shortest path from $u$ to $v$. 
    Let $u=u_0$ and let $v=v_0$. For every $i\in [1,t]$, let $u_i\in B(u_{i-1})$ be the farthest vertex from $u$ on the path $P$ and 
    let $v_i\in B(v_{i-1})$ be the farthest vertex from $v$ on the path $P$. 
    Let $P_i(u)=P(u,u_i)$, that is, the portion on $P$ between $u$ and $u_i$. 
    
    We divide the proof into two cases. The case that $P_t(u) \cap P_t(v) \neq \emptyset$ and the case 
    that $P_t(u) \cap P_t(v) = \emptyset$.
    Consider the case that $P_t(u) \cap P_t(v) \neq \emptyset$.  
    Let $x\in P_t(u) \cap P_t(v)$. From Lemma~\ref{L-Pt-Dt} it follows that $d_t(u,x)=d(u,x)$ and $d_t(v,x)=d(v,x)$.
    Since $P_t(u)\subseteq S_t(u)$ and $P_t(v)\subseteq S_t(v)$ the value returned by the call to  $\Intersection(u,v,S_t(u),S_t(v))$ is at most 
    $d_t(u,x)+d_t(v,x)=d(u,x)+d(x,v)=d(u,v)$, where the last equality follows from the fact that $x\in P$, and the claim holds. 
     
    Consider now the case that $P_t(u) \cap P_t(v) = \emptyset$, from Lemma~\ref{L-Bound-Sum-Of-h-in-S} we know that $\sum_{i=0}^{t-1}h(u_i) + h(v_i) \le d(u,v)+2t-1$.
    
    Let $Q=\{\langle u_i,v_{t-i-1}\rangle \mid 0\leq i \leq t-1\}\}$. For each $q\in Q$, let $q_u$ (resp. $q_v$) be the first (resp. second) element of the pair.
    Summing over  $Q$ we get  $\sum_{q\in Q} (h(q_u) + h(q_v))=\sum_{i=0}^{t-1}(h(u_i) + h(v_i))$. Since $\sum_{i=0}^{t-1}(h(u_i) + h(v_i)) \le d(u,v)+2t-1$ we get that $\sum_{q\in Q} h(q_u) + h(q_v) = \sum_{i=0}^{t-1}(h(u_i) + h(v_i)) \le d(u,v)+2t-1$.

    Let $q\in Q$ be the pair with minimum value of $h(q_u)+h(q_v)$.
    By the minimality of $q$ we get $h(q_u)+h(q_v) \le (\sum_{q\in Q} h(q_u) + h(q_v))/t \le (d(u,v)+2t-1)/t$.
    Since $h(q_u)+h(q_v)$ is an integer, it follows that $h(q_u) + h(q_v) \le \lfloor d(u,v)+2t-1)/t \rfloor =  1 + \lceil d(u,v)/t \rceil$.
    
    The algorithm iterates over all the possible pairs of $Q$ and encounters $q$. Thus:
    \[\hat{d}(u,v) \le d(u,q_u) + d(q_u,p(q_u)) + d(p(q_u), p(q_v)) + d(p(q_v), q_v) + d(q_v, v)\]
    From the triangle inequality, we know that $d(p(q_u), p(q_v)) \le h(q_u) + d(q_u,q_v) + h(q_v)$. Thus: 
    \begin{align*}
        \hat{d}(u,v) &\le d(u,q_u) + d(q_u,p(q_u)) + d(p(q_u), p(q_v)) + d(p(q_v), q_v) + d(q_v, v) \\
                     &\le d(u,q_u) + h(q_u) + (h(q_u) + d(q_u,q_v) + h(q_v)) + h(q_v) + d(q_v, v) \\
                     &= d(u,q_u) + d(q_u,q_v) + d(q_v, v) + 2(h(q_u) + h(q_v))
    \end{align*}
    Since $P_t(u) \cap P_t(v) = \emptyset$ we get that when going from $u$ to $v$ on $P$ we first encounter all $u_i$'s and then all $v_i$'s. Thus, $d(u,q_u) + d(q_u,q_v) + d(q_v,v) = d(u,v)$. Since $d(u,q_u) + d(q_u,q_v) + d(q_v,v) = d(u,v)$ we get  $\hat{d}(u,v) \le d(u,v) + 2(h(q_u)+h(q_v))$. 
    Since $h(q_u) + h(q_v) \le 1 + \lceil d(u,v)/t \rceil$, we have $\hat{d}(u,v) \le d(u,v) + 2(h(q_u) + h(q_v)) \le d(u,v)  + 2\lceil d(u,v)/t \rceil+ 2$, as required.
\end{proof}

\begin{lemma}\label{L-Runtime-Q-5/3-UW}
    The query algorithm takes $O(n^{tc})$ time.
\end{lemma}
\begin{proof}
    From Lemma~\ref{L-TZ-Size} we have $|B(w)| \le O(n^c)$, for every $w\in V$. 
    Thus, from Lemma~\ref{L-size-and-running-time-of-S_i} it follows that the call to $\GenerateSi(u,i)$ and the call to $\GenerateSi(v,i)$ take $O(n^{ic})$, for every $1\leq i\leq t$. 
    It also follows from Lemma~\ref{L-size-and-running-time-of-S_i}
    that  $|S_i(u)| = O(n^{ic})$ and $|S_i(v)| = O(n^{ic})$.
    Therefore, the cost of computing $d(u, u') + d(u',p(u')) + d(p(u'),p(v')) + d(p(v'),v') + d(v',v)$, for every $\langle u',v' \rangle \in S_i(u)\times S_{t-1-i}(v)$, is $O(|S_i(u)| \cdot |S_{t-1-i}(v)|)=O(n^{(t-1)c})$ time.
    Since $|S_t(u)| = O(n^{tc})$, it follows from Lemma~\ref{Intersection-Runtime} that the running time of $\Intersection(u,v, S_t(u), S_t(v))$ is $O(n^{tc})$.
\end{proof}
Theorem~\ref{T-Unweighted-5/3} follows from Lemma\ref{L-Space-Unweighted-5/3}, Lemma\ref{L-Bound-5/3-UW} and Lemma\ref{L-Runtime-Q-5/3-UW}.

\section{Distance oracles in weighted graphs with stretch at most \texorpdfstring{$2$}{2}}\label{S-at-most-2weighted}
In this section, we extend our second technique, of computing $S_t(u)$ and $S_t(v)$ by the query algorithm, to weighted undirected graphs with non-negative real edge weights. We start this section by providing preliminaries related to weighted graphs. 
\subsection{Toolbox}
Let $G=(V,E)$ be a weighted undirected graph with real non-negative edges.
Let $\mu = 2m/n$, be the average degree of the graph.
For every $u\in V$, let $B^*(u)$  be the set of all the vertices adjacent to a vertex in $B(u)$, that is, $B^*(u)=N(B(u))\cup B(u)$.

Agarwal, Godfrey, and Har-Peled~\cite{DBLP:journals/corr/abs-1201-2703} showed that in the case of distance oracles in weighted graphs that use $\Omega(m)$ space, it suffices to consider graphs with maximum degree $\mu$. 
Thus, for the rest of this section, we assume that the maximum degree of $G$ is at most  $\mu$.
The next lemma follows from the discussion above.
\begin{lemma}\label{L-Bound-B^*-size}
    Let $u\in V$. $|B^*(u)| \le (1+\mu)|B(u)|$.
\end{lemma}
\begin{proof}
    Since the maximum degree of $G$ is bounded by $\mu$, we know that $|N(B(u))| \le \mu \cdot |B(u)|$. From the definition of $B^*(u)$, it follows that $|B^*(u)| = | B(u) \cup N(B(u))| \le |B(u)| + |N(B(u))|$. Since $|N(B(u))|\le \mu |B(u)|$, we get that $|B^*(u)| \le |B(u)| + |N(B(u))| \le |B(u)| + \mu |B(u)| = (1+\mu)|B(u)|$, as required.
\end{proof}

Let $u,v\in V$, and let $P=P(u,v)$. Let $u'\in B^*(u)$ be the farthest vertex from $u$ in $P$.
We show: 

\begin{lemma}\label{L-distance-from-farthest-vertex-on-path-weighted}
$d(u,u') \ge h(u)$.
\end{lemma}
\begin{proof}
    Assume towards a contradiction, that $d(u, u') < h(u)$ thus  $u'\in B(u)$. The  vertex after $u'$ in $P$  is in $N(u')$ and thus  in $B^*(u)$, contradicting the fact that $u'$ is the farthest  from $u$ in $P$ which is also in  $B^*(u)$.
\end{proof}

\begin{lemma}
    \label{L-Intersect-Weighted}
    If $h(u)+h(v)> d(u,v)$ then $P \cap (B^*(u) \cap B^*(v)) \neq \emptyset$
\end{lemma}
\begin{proof}
    If $v'$ appears before $u'$ when going from $u$ to $v$ in $P$ then $P \cap (B^*(u) \cap B^*(v)) \neq \emptyset$ and the claim holds.
    Otherwise, we know that $d(u,v) = d(u,u') + d(u',v') +d(v',v)$.
    From Lemma~\ref{L-distance-from-farthest-vertex-on-path-weighted} it follows that $d(u,u')\ge h(u)$ and $d(v,v')\ge h(v)$. Thus, $d(u,v) = d(u,u') + d(u',v') + d(v,v') \ge h(u) + h(v) +  d(u',v')$. 
    Since $d(u',v')\geq 0$ we get that  $d(u,v) \geq h(u) + h(v)$, contradicting the assumption that $d(u,v)<h(u)+h(v)$.
\end{proof}
\begin{lemma} \label{L-no-intersect-weighted}
    If $B^*(u)\cap B^*(v) = \emptyset$ then $h(u) + h(v) \le d(u,v)$
\end{lemma}
\begin{proof}
Let $u'\in B^*(u)$ be the farthest vertex from $u$ in $P$. From lemma~\ref{L-distance-from-farthest-vertex-on-path-weighted} it follows that $d(u,u')\ge h(u)$. 
Assume for the sake of contradiction that $h(u)+h(v) > d(u,v)$. 
This implies that $h(v) > d(u,v)-h(u)$. 

Since $u'$ is in $P$, we know that $d(u,v) = d(u,u') + d(u',v)$ and since $d(u,u') \ge h(u)$ we get that $d(u,v)\ge h(u) + d(u',v)$, thus $d(u',v)\le d(u,v)-h(u)$. 
Now since $h(v) > d(u,v)-h(u)$ we get that $h(v) > d(u',v)$. This implies that $u'\in B^*(v)$, contradiction to the fact that $B^*(u)\cap B^*(v) = \emptyset$.
\end{proof}
\subsection{The set \texorpdfstring{$S_t(\cdot)$}{S_t(\cdot)} in the weighted case}

Let $t\geq 0$ be an integer. Let $S_t(u)=\bigcup_{w\in S_{t-1}(u)}{B^*(u)}$, where $S_0(u)=\{u\}$.
We compute the set $S_t(u)$ using  a simple recursive procedure $\GenerateSi$, presented in Algorithm~\ref{Algorithm-GenerateSi-Weighted}. 

\begin{algorithm2e}[t] 
\caption{$\GenerateSi(u, t)$}\label{Algorithm-GenerateSi-Weighted}
\lIf{$t=0$}{\Return $\{u\}$}
$\langle S_{t-1}(u) , d_{t-1}\rangle \gets \GenerateSi(u,t-1)$\\
\ForEach{$x\in S_{t-1}(u)$}
{
    \ForEach{$v\in B^*(x)$}
    {
    $d_t(u,v) = \min ( d_t(u,v), d_{t-1}(u,x)+d(x,v))$
    }
}
\Return $\langle \bigcup_{w\in S_{t-1}(u)}{B^*(w)}, d_t\rangle$
\end{algorithm2e}

Next, let $0<c\leq 1$ be the minimum value for which $|B^*(w)| \le O(\mu n^{c})$, for every $w\in V$. We show:
\begin{lemma}\label{L-size-and-running-time-of-S_i-weighted}
     $\GenerateSi(u,t)$ computes in $O(\mu^t n^{ct})$ time the set $S_t(u)$ of size $O(\mu^tn^{ct})$.
\end{lemma}
\begin{proof}
    We prove the claim by induction on $i\in [0,t]$. For $i=0$, the claim holds since $S_0(u)=\{u\}$, and  $|S_0(u)|=1$.
    Next, assume that $|S_{i-1}(u)| = O(\mu^{i-1} n^{c(i-1)})$ and that  $S_{i-1}(u)$ was computed in $O(\mu^{i-1}n^{c(i-1)})$ time.
    In order to create $S_i(u)$, the algorithm $\GenerateSi(u, i)$ iterates over every $w\in S_{i-1}(u)$
    and adds $B^*(w)$ to $S_i(u)$ in $O(|B^*(w)|)=O(\mu n^c)$ time. 
    Since $|S_{i-1}(u)| = O(\mu^{i-1} n^{c(i-1)})$ the total time for computing $S_i(u)$ is $O(\mu n^c \cdot \mu^{i-1} n^{c(i-1)})=O(\mu^ {i} n^{ci})$.  
    From the same arguments we have $|S_{i}(u)| =O(\mu^ {i} n^{ci})$. 
\end{proof}

For every $1\leq i \leq t$, let $u_i \in B^*(u_{i-1})$ be the farthest vertex in $P$ from $u$, and let $v_i\in B^*(v_{i-1})$ be the farthest vertex in $P$ from $v$. Let $P_t(u)=P(u,u_t)$, that is, the portion on $P$ between $u$ and $u_t$. Notice that every vertex in $P_t(u)$ is in a $B^*$ of some vertex from $S_{t-1}(u)$. Therefore, by the definition of $S_{t}(u)$ we have that $P_t(u)\subseteq S_{t}(u)$. 

In the next lemma, we show that for the vertices of $P_t(u)$ the distance function $d_t$ is the same as the distance function $d$ of $G$. 
\begin{lemma}\label{L-Pt-Dt-W}
If $v\in P_t(u)$ then $d_t(u,v)=d(u,v)$.
\end{lemma}
\begin{proof}
We prove by induction on $t$. 
For the base case of $t=0$ the claim holds since $P_0(u)=\{u\}$ and $d_0(u,u)=0$. 
We assume the claim holds for every $0< i < t$ and prove the claim for $i=t$. 
Let $v\in P_t(u)$. This implies that $v\in P(u_{i-1},u_i)$, where $1\leq i\leq t$. 
Since $u_i\in P$ is the farthest vertex from $u$ in $B^*(u_{i-1})$ and since $v\in P(u_{i-1},u_i)$, a shortest path between $u_{i-1}$ and $u_i$, it follows that $v\in B^*(u_{i-1})$.
By definition of $d_t$ we have $d_t(u,v)\leq d_{t-1}(u,u_{i-1})+d(u_{i-1},v)$, since $v\in B^*(u_{i-1})$. 
By the induction assumption $d_{t-1}(u,u_{i-1})=d(u,u_{i-1})$, therefore, $d_t(u,v)\leq  d_{t-1}(u,u_{i-1})+d(u_{i-1},v) = d(u,u_{i-1})+d(u_{i-1},v)=d(u,v)$, where the last equality holds since $u_{i-1}\in P(u,v)$. 
\end{proof}

In the next lemmas, we examine the case that $P_t(u) \cap P_t(v)  = \emptyset$. 

\begin{lemma}
\label{L-No-S-Intersect-Weighted}
    If $P_t(u) \cap P_t(v) = \emptyset$ then $h(u_{t-1}) + h(v_{t-1}) \leq d(u_{t-1},v_{t-1})$.
\end{lemma}
\begin{proof}
    From the definition of $P_t$, it follows that $B^*(u_{t-1}) \cap P \subseteq P_t(u)$. Similarly, $B^*(v_{t-1}) \cap P \subseteq P_t(v)$. Thus, $(B^*(u_{t-1}) \cap P) \cap (B^*(v_{t-1}) \cap P)= B^*(u_{t-1}) \cap B^*(v_{t-1}) \cap P \subseteq P_t(u)\cap P_t(v) = \emptyset$. 
    Since $B^*(u_{t-1}) \cap B^*(v_{t-1}) \cap P \subseteq \emptyset$ it follows that $B^*(u_{t-1}) \cap B^*(v_{t-1}) \cap P = \emptyset$ and by
    applying Lemma~\ref{L-Intersect-Weighted} we get $B^*(u_{t-1}) \cap B^*(v_{t-1})=\emptyset$.
    Now, by using Lemma~\ref{L-no-intersect-weighted} we get $h(u_{t-1}) + h(v_{t-1}) \le d(u_{t-1},v_{t-1})$, as required.
\end{proof}

\begin{lemma}
\label{L-Bound-Sum-Of-H-in-S-Weighted}
    If $P_t(u) \cap P_t(v) = \emptyset$ then $\sum_{i=0}^{t-1}h(u^i) + h(v^i) \le d(u,v)$.
\end{lemma}
\begin{proof}
    Since $u_{i-1},u_i\in P$, $u_i$ is the farthest vertex from $u$ in $B^*(u_{i-1})$, it follows from Lemma~\ref{L-distance-from-farthest-vertex-on-path-weighted} that $d(u_{i-1}, u_i)=h(u_{i-1})$ for every $0< i \le t$. Similarly, $d(v_{i-1},v_i) = h(v_{i-1})$.
    
    % We first show that $d(u_{i-1},u_{i})\ge h(u_{i-1})$, for every $0\leq i\leq  t$.
    % Let $y\in B^*(u_{i-1})$ be the farthest vertex from $u_{i-1}$ in $P$. From Lemma~\ref{L-distance-from-farthest-vertex-on-path-weighted} we have that $d(u_{i-1},y) \ge h(u_{i-1})$. Since $S_i(u)=\bigcup_{w\in S_{i-1}(u)}{B^*(u)}$ and $u_{i-1}\in S_{i-1}(u)$ it follows that $B^*(u_{i-1}) \subseteq S_i(u)$. Recall that $u_i\in P$ is the farthest vertex from $u$ in $S_i(u)$. 
    % As $B^*(u_{i-1}) \subseteq S_i(u)$ it follows that $d(u,u_i) \ge d(u,y)$. 
    % Since $P$ is a shortest path and $u,u_{i-1},u_i,y\in P$ we have 
    % $d(u,u_i) = d(u,u_{i-1})+d(u_{i-1},u_i)$
    % and $d(u,y) = d(u,u_{i-1})+d(u_{i-1},y)$. Thus,
    % $d(u_{i-1},u_i) \ge d(u_{i-1},y) \ge h(u_{i-1})$, as wanted.
    % From the same arguments as above, it holds that $d(v_{i-1},v_i) \geq h(v_{i-1})$, for every $0\leq i \leq t$.

    From the assumption that $P_t(u) \cap P_t(v) = \emptyset$ it follows that when going from $u$ to $v$
    on $P$ we first encounter all $u_i$'s and then all $v_i$'s. Thus:
    \begin{align*} 
d(u,v)    &= d(u,u_1) + \dots +d(u_{t-2},u_{t-1})+ d(u_{t-1},v_{t-1}) + d(v_{t-1},v_{t-2})+\dots + d(v_1,v) \\ 
          &=  d(u_{t-1}, v_{t-1}) + \sum_{i=0}^{t-2} ( d(u_i,u_{i+1}) + d(v_i,v_{i+1}) )
\end{align*}
       Since $P_t(u) \cap P_t(v) = \emptyset$ it follows from Lemma~\ref{L-No-S-Intersect-Weighted} that $h(u_{t-1}) + h(v_{t-1}) \leq d(u_{t-1},v_{t-1})$. Recall that $d(u_{i-1},u_i) \geq h(u_{i-1})$ and $d(v_{i-1},v_i) \geq h(v_{i-1})$, for every $0\leq i \leq t$.
    Therefore, we get:
\[  
d(u,v)    = d(u_{t-1}, v_{t-1}) + \sum_{i=0}^{t-2} (d(u_i,u_{i+1}) + d(v_i,v_{i+1}) )
          \geq h(u_{t-1}) + h(v_{t-1}) + \sum_{i=0}^{t-2}(h(u_i) + h(v_i))
          = \sum_{i=0}^{t-1}(h(u_i) + h(v_i))
\] as required.
\end{proof}

\subsection{\texorpdfstring{$1+2/t$}{1+2/t}-stretch with \texorpdfstring{$n^{2-2c}$}{n\^(2-2c)}-space and \texorpdfstring{$\mu^t n^{tc}$}{(m/n)\^3 n\^(tc)}-query}\label{S-5-3-weighted}
Similar to the unweighted case, the $(1+1/t)$-stretch presented by Agarwal~\cite{DBLP:conf/esa/Agarwal14} uses $\Omega(n^{1.5})$-space. In the next theorem, we use $S_t(\cdot)$ in a similar manner to Section~\ref{S-5-3} and present a distance oracle that uses $m+o(n^{1.5})$-space.

\Reminder{T-Weighted-5/3}

\subsubsection{Storage.}

We compute a set $A$ of size $O(n^{1-c})$ using Lemma~\ref{L-TZ-Size}.
For every $u\in V$ we save  $B(u)$. For every $u,w\in A$ we save $d(u,w)$. For every  $u\in V$ we save $p(u)$. 
We also save the graph $G$.
\begin{lemma}\label{L-Space-T-Weighted-5/3}
    The distance oracle is constructed in $\Ot(mn^{1-c})$-time and  uses  $O(m+n^{2-2c})$-space.
\end{lemma}
\begin{proof}
From Lemma~\ref{L-TZ-Size} it follows that  $|B(u)|=O(n^c)$.  
In addition, for every $w,v\in A$, we store $d(w,v)$ at a cost of $|A|^2=O(n^{2-2c})$ space. 
Saving $G$ requires $O(m)$ space.
The total space is $O(m + n^{2-2c} + n^{1+c})$. Since 
$0 < c < 1/3$ we have $2-2c>1+c$ and the space is $O(m + n^{2-2c})$.
From Lemma~\ref{L-TZ-Size} it follows that for every $u\in V$ the set $B(u)$ and $A$ are computed in $O(mn^c)$ time. 
Computing $p(u)$, for every $u\in V$ takes $\Ot(m)$ time.
Computing $d(w,v)$ for every $w\in A$ and $v\in V$ takes $\Ot(m|A|)=\Ot(mn^{1-c})$ time.
Since $1-c>c$ we get that the running time  is $\Ot(mn^{1-c})$. 
\end{proof}

\subsubsection{Query algorithm.}
The query algorithm is identical to the query algorithm from Section~\ref{S-5-3}.
The only difference is that $S_t(\cdot)$ is computed using the weighted version of $\GenerateSi$. 
Next, we bound  $\hat{d}(u,v)$. 
\begin{lemma}\label{L-Bound-5/3-Weighted}
    $\hat{d}(u,v) \le (1+2/t)d(u,v)$.
\end{lemma}
\begin{proof}
    Let $P=P(u,v)$. 
    Let $u=u_0$ (resp. $v=v_0$). For every $1 \leq i \leq t$, let $u_i\in B^*(u_{i-1})$ (resp. $v_i\in B^*(v_{i-1})$ be the farthest vertex from $u$ (resp. $v$) on the path $P$.
    Let $P_i(u)=P(u,u_i)$, that is, the portion on $P$ between $u$ and $u_i$. 
    
    We divide the proof into two cases. The case that $P_t(u) \cap P_t(v) \neq \emptyset$ and the case 
    that $P_t(u) \cap P_t(v) = \emptyset$.
    Consider the case that $P_t(u) \cap P_t(v) \neq \emptyset$.  
    Let $x\in P_t(u) \cap P_t(v)$. From Lemma~\ref{L-Pt-Dt-W} it follows that $d_t(u,x)=d(u,x)$ and $d_t(v,x)=d(v,x)$.
    Since $P_t(u)\subseteq S_t(u)$ and $P_t(v)\subseteq S_t(v)$ the value returned by the call to $\Intersection(u,v,S_t(u),S_t(v))$ is at most 
    $d_t(u,x)+d_t(v,x)=d(u,x)+d(x,v)=d(u,v)$, where the last equality follows from the fact that $x\in P$, and the claim holds. 
     
    Consider now the case that $P_t(u) \cap P_t(v) = \emptyset$. From Lemma~\ref{L-Bound-Sum-Of-H-in-S-Weighted} it follows that $\sum_{i=0}^{t-1}h(u_i) + h(v_i) \le d(u,v)$.
    
    Let $Q=\{\langle u_i,v_{t-1-i}\rangle \mid 0\leq i \leq t-1\}$. For each $q\in Q$, let $q_u$ (resp. $q_v$) be the first (resp. second) element of the pair.
    By summing over the set $Q$ we get that $\sum_{q\in Q} (h(q_u) + h(q_v))=\sum_{i=0}^{t-1}(h(u_i) + h(v_i))$. Since $\sum_{i=0}^{t-1}(h(u_i) + h(v_i)) \le d(u,v)$ we get that $\sum_{q\in Q} h(q_u) + h(q_v) = \sum_{i=0}^{t-1}(h(u_i) + h(v_i)) \le d(u,v)$.

    Let $q\in Q$ be the pair with minimum value of $h(q_u)+h(q_v)$.
    By the minimality of $q$ we get $h(q_u)+h(q_v) \le (\sum_{q\in Q} h(q_u) + h(q_v))/t \le d(u,v)/t$.

    The algorithm iterates over all the possible pairs of $Q$ and encounters $q$. Thus:
    \[\hat{d}(u,v) \le d(u,q_u) + d(q_u,p(q_u)) + d(p(q_u), p(q_v)) + d(p(q_v), q_v) + d(q_v, v)\]
    From the triangle inequality, we know that $d(p(q_u), p(q_v)) \le h(q_u) + d(q_u,q_v) + h(q_v)$. Thus: 
\begin{align*}
    \hat{d}(u,v) &\le d(u,q_u) + d(q_u,p(q_u)) + d(p(q_u), p(q_v)) + d(p(q_v), q_v) + d(q_v, v) \\
                 &\le d(u,q_u) + h(q_u) + (h(q_u) + d(q_u,q_v) + h(q_v)) + h(q_v) + d(q_v, v) \\
                 &= d(u,q_u) + d(q_u,q_v) + d(q_v, v) + 2(h(q_u) + h(q_v))
\end{align*}
Since $P_t(u) \cap P_t(v) = \emptyset$ we get that when going from $u$ to $v$ on $P$ we first encounter all $u_i$'s and then all $v_i$'s. Thus we get $d(u,q_u) + d(q_u,q_v) + d(q_v,v) = d(u,v)$. Since $d(u,q_u) + d(q_u,q_v) + d(q_v,v) = d(u,v)$ we get that $\hat{d}(u,v) \le d(u,v) + 2(h(q_u)+h(q_v))$. 

Since $h(q_u) + h(q_v) \le d(u,v)/t$, it follows that $\hat{d}(u,v) \le d(u,v) + 2(h(q_u) + h(q_v)) \le (1+2/t)d(u,v)$, as required.        
\end{proof}

\begin{lemma}\label{Runtime-Q-5/3-Weighted}
    The query algorithm takes $O(\mu^t n^{tc})$ time.
\end{lemma}
\begin{proof}
    From Lemma~\ref{L-TZ-Size} we have $|B(w)| \le O(n^c)$, for every $w\in V$. 
    Thus, from Lemma~\ref{L-size-and-running-time-of-S_i-weighted} it follows that the call to $\GenerateSi(u,i)$ and the call to $\GenerateSi(v,i)$ take $O(\mu^i n^{ic})$, for every $1\leq i\leq t$. 
    It also follows from Lemma~\ref{L-size-and-running-time-of-S_i-weighted}
    that $|S_i(u)| = O(\mu^i n^{ic})$ and $|S_i(v)| = O(\mu^i n^{ic})$.
    Therefore, the cost of computing $d(u, u') +d(u',p(u')) + d(p(u'),p(v')) + d(p(v'),v') + d(v',v)$, for every $\langle u',v' \rangle \in S_i(u)\times S_{t-i}(v)$, is $O(|S_i(u)| \cdot |S_{t-i}(v)|)=O(\mu^{t-1} n^{(t-1)c})$ time.
    Since $|S_t(u)| = O(\mu^t n^{tc})$, it follows from Lemma~\ref{Intersection-Runtime} that the running time of $\Intersection(u,v, S_t(u), S_t(v))$ is $O(\mu^t n^{tc})$.
\end{proof}

Theorem~\ref{T-Weighted-5/3} follows from Lemma~\ref{L-Space-T-Weighted-5/3}, Lemma~\ref{L-Bound-5/3-Weighted} and Lemma~\ref{Runtime-Q-5/3-Weighted}.

\section{New applications - $n-PSP$ and $ANSC$}\label{S-apps}
In this section, we describe several applications of our new distance oracles that improve upon previous results for the $n-PSP$ and $ANSC$ problems.

\subsection{Improved algorithms for $n$-pairs shortest paths problem ($n-PSP$)}\label{S-apps-npsp}
All of our algorithms for the $n-PSP$ problem follow the following framework. We construct a distance oracle and then query it $n$ times. Therefore, the runtime of the resulting algorithm is the construction time of the distance oracle plus $n$ times the query time of the distance oracle.

For unweighted graphs, we use the distance oracle of \autoref{T-Unweighted-1+1/t} to achieve the following theorem.
\Reminder{T-NPSP-1+1/k-Unweighted}

\begin{proof}
    We construct the distance oracle of \autoref{T-Unweighted-1+1/t} and query it for every $(s_i,t_i)$.
    From \autoref{T-Unweighted-1+1/t} constructing the distance oracle takes $\Ot(mn^{1-c})$ time. Querying the distance oracle $n$ times takes $\Ot(n\cdot n^{ck})=\Ot(n^{1+ck})$ time. 
    Therefore, we get a running time of $\Ot(mn^{1-c}+n^{1+ck})$, and by setting $c=\frac{\log_n{m}}{k+1}$ we get that $\Ot(mn^{1-c}+n^{1+ck})=\Ot(m^{1-\frac{1}{k+1}}n+m^{1-\frac{1}{k+1}}n)=\Ot(m^{1-\frac{1}{k+1}}n)$, as required.
    From \autoref{T-Unweighted-1+1/t} we know that $\hat{d}(s_i,t_i) \le d(s_i,t_i) + 2\ceil{d(s_i,t_i)/(kt)}$, as required.
\end{proof}

In addition, the following simple theorem follows from using the spanner of~\cite{DBLP:journals/siamcomp/BaswanaK10} and the distance oracle of~\autoref{T-2k-3-Weighted}. 
\Reminder{T-NPSP-2k-1-2k-3}

\begin{proof}
    We construct a $(2k-1)$-spanner $H$ in $O(km)$ time (\cite{DBLP:journals/siamcomp/BaswanaK10}), with $O(n^{1+1/k})$ edges. Then, we construct the distance oracle of~\autoref{T-2k-3-Weighted} for the graph $H$.
    Constructing the distance oracle takes $\Ot(|E(H)|n^{\frac{1}{k}})=\Ot(n^{1+\frac{1}{k}}n^{\frac{1}{k}})=\Ot(n^{1+2/k})$ time. 
    Querying the distance oracle $n$ times takes $\Ot(n\cdot \frac{|E(H)|}{n}n^{\frac{1}{k}})=\Ot(n^{1+\frac{2}{k}})$. Therefore, we conclude that the running time is $\Ot(n^{1+\frac{2}{k}})$.

    From~\autoref{T-2k-3-Weighted} we have that $\hat{d}(s_i,t_i) \le (2k-3)d_H(s_i,t_i)$. From the fact that $H$ is a $(2k-1)$-stretch spanner, we know that $d_H(s_i,t_i)\le (2k-1)d(u,v)$. Overall, we get that 
    $\hat{d}(s_i,t_i) \le (2k-3)d_H(s_i,t_i)\le (2k-3)(2k-1)d(u,v)$, as required.
\end{proof}

\subsection{Improved algorithms for all nodes shortest cycles ($ANSC$)}\label{S-ANSC}
% \input{tables/ansc}
% This section presents our new algorithms using the previous distance oracles that improve upon the best-known algorithms of~\cite{DBLP:conf/focs/DalirrooyfardJW22}.
We prove the following two theorems, one for unweighted graphs and one for weighted graphs. 
Both theorems follow the same framework with a different distance oracle.
For unweighted graphs, we use the distance oracle of \autoref{T-Unweighted-1+1/t}, and for weighted graphs, we use the distance oracle of~\cite{DBLP:conf/esa/Agarwal14}. For completeness, in the next lemma we provide the properties of the distance oracle used from~\cite{DBLP:conf/esa/Agarwal14}.

\begin{lemma}[\cite{DBLP:conf/esa/Agarwal14}]\label{L-oracle-of-agarwal14}
Let $k\geq 1$ be an integer and let $0 < c < 1/2$ be a real constant.
There is a $\Ot(m+n^{2-c})$-space distance oracle that given  two query vertices $u,v\in V$ computes in $\Ot(\mu^kn^{c\cdot k})$-time 
a distance estimation $\hat{d}(u,v)$ that satisfies 
$d(u,v) \leq \hat{d}(u,v) \le d(u,v) + (1+1/k)d(u,v)$. The distance oracle is constructed in $\Ot(mn^{1-c})$ time. 
\end{lemma}

\subsubsection{Algorithm.}
Let $k>1$ be an integer parameter.
Let $0 < c \le 1/2$ be a constant to be determined later.
Using the oracle of \autoref{T-Unweighted-1+1/t} (Lemma~\ref{L-oracle-of-agarwal14}), we construct a distance oracle with parameters $c$ and $k$.
Notice, that both of the query algorithms first grow a ball around $u$ and $v$, thus, we can easily change the implementation to ignore the edge $(u,v)$ and to approximate the shortest path that is longer than a single edge, we denote this query by $\Query_{>1}$.
For every vertex $u\in V$ we set $\hat{c}_u$ to $\min_{v\in N(u)}(\Query_{>1}(v,u) + w(u,v))$.
Pseudo-code for this algorithm exists in Algorithm~\ref{A-ansc-1+1/k}.

Let $u\in V$. Let $Cy$ be the shortest cycle going through $u$. Next we bound $\hat{c}_u$.
\begin{lemma} \label{L-Correctness-1+1/k-ANSC}
    $\hat{c}_u \le 1+2\lceil SC(u)/2k \rceil$ ($\hat{c}_u \le (1+1/k)SC(u)$)
\end{lemma}
\begin{proof}
    Let $v\in N(u)\cap Cy$. Therefore, $SC(u)=|Cy|=w(u,v)+d_{>1}(v,u)$.
    During our algorithm we encounter the edge $(u,v)$ and therefore we get that $\hat{c}_u \le w(u,v)+\Query_{>1}(v,u)$.
    From the correctness of \autoref{T-Unweighted-1+1/t} (Lemma~\ref{L-oracle-of-agarwal14}), we know that $\Query_{>1}(v,u) \le d_{>1}(v,u) + 2\lceil \frac{d_{>1}(v,u)}{2k} \rceil$ ($\Query_{>1}(v,u) \le (1+1/k)d_{>1}(v,u)$).
    Next, we divide the proof into the case of unweighted and weighted graphs.
    In unweighted graphs, we get that:
    \[\hat{c}_u \le w(u,v)+\Query_{>1}(v,u) \le  1+d_{>1}(v,u) + 2\lceil \frac{d_{>1}(v,u)}{2k} \rceil = SC(u) + 2\lceil \frac{SC(u)-1}{2k} \rceil \le SC(u)+2\lceil \frac{SC(u)}{2k} \rceil,\] as required.
    In weighted graphs, we get that:
    \[\hat{c}_u \le w(u,v)+\Query_{>1}(v,u) \le  w(u,v)+(1+1/k)d_{>1}(v,u) \le SC(u) + (1/k)(SC(u)-w(u,v))\le (1+1/k)SC(u),\] as required.
\end{proof}

In the following two lemmas, we bound the running time of $\ApproximateANSC$ in unweighted and weighted graphs, respectively.
\begin{lemma}\label{L-Runtime-Unweighted-ANSC}
    The running time of $\ApproximateANSC$ in unweighted graphs is $O(mn^{1-1/(k+1)})$
\end{lemma}
\begin{proof}
    From \autoref{T-Unweighted-1+1/t}, the construction time of the distance oracle takes $\Ot(mn^{1-c})$ time.
    From \autoref{T-Unweighted-1+1/t}, calling $\Query_{>1}(u,v)$ for every $(u,v)\in E$, takes $O(m\cdot n^{kc})$-time.
    Therefore, the running time is:
    $$\Ot(mn^{1-c} + m \cdot n^{kc}) = \Ot(mn^{1-1/(k+1)} + mn^{k/(k+1)}=mn^{1-1/(k+1)}),$$
    where the first equality follows from setting $c={1/(k+1)}$.
\end{proof}

\begin{lemma}\label{L-Runtime-Weighted-ANSC}
    The running time of $\ApproximateANSC$ in weighted graphs is $O(mn^{1-1/(k+1)})$
\end{lemma}
\begin{proof}
    From Lemma~\ref{L-oracle-of-agarwal14}, the construction time of the distance oracle is $\Ot(mn^{1-c})$.
    From Lemma~\ref{L-oracle-of-agarwal14} calling $\Query_{>1}(u,v)$ for every $(u,v)\in E$, takes $O(m\cdot \mu^kn^{kc})$-time.
    Therefore, the running time is:
    \begin{align*}
        \Ot(mn^{1-c} + m\mu^kn^{kc}) &= \Ot(mn^{1-(1-(1-\frac{1}{k+1})\log_n(m))} + m(m/n)^kn^{k(1-(1-\frac{1}{k+1})\log_n(m))}) \\
        &= \Ot(m^{2-\frac{1}{k+1}}+m(m/n)^kn^{k-(k-\frac{k}{k+1})\log_n(m)}) = \Ot(m^{2-\frac{1}{k+1}} + m^{2-\frac{1}{k+1}}) = \Ot(m^{2-\frac{1}{k+1}})
    \end{align*}
    where the first equality follows from setting $c=1-(1-\frac{1}{k+1})\log_n(m)$, as required.
\end{proof}

\begin{algorithm2e}[t] 
\caption{$\ApproximateANSC(G,k)$}\label{A-ansc-1+1/k}
$ADO \gets \codestyle{DistanceOracle}(G,k)$ (\autoref{T-Unweighted-1+1/t} / ~\cite{DBLP:conf/esa/Agarwal14}) \\
$\hat{c} \gets [\infty, \dots, \infty]$\\
\ForEach{$(u,v)\in E$} {
    $\hat{c}_u = \min(\hat{c}_u, ADO.\Query_{>1}(u,v) + w(u,v))$ \\
    $\hat{c}_v = \min(\hat{c}_v, ADO.\Query_{>1}(u,v) + w(u,v))$
}
\Return $\hat{c}$
\end{algorithm2e}

The following theorem follows from Lemma~\ref{L-Correctness-1+1/k-ANSC} and Lemma~\ref{L-Runtime-Unweighted-ANSC}.
\Reminder{T-ANSC-1+1/k-approx-unweighted}

The following Theorem follows from Lemma~\ref{L-Correctness-1+1/k-ANSC} and Lemma~\ref{L-Runtime-Weighted-ANSC}
\Reminder{T-ANSC-1+1/k-approx-weighted}

\bibliographystyle{alpha}
\bibliography{bibliography}
\pagebreak
\appendix 
\section{Tables of our distance oracles}\label{S-Tables}
\subsection{Weighted graphs with stretch $\ge 2$}
\begin{table}[H]
    \centering
    \begin{tabular}{|c|c|c|c|c|c|}
    \hline
     Construction & Space & Time  & Stretch & Ref. & Comments \\ \hline\hline
    $\Ot(mn^{1/k})$ & $\Ot(n^{1+1/k})$ & $\Ot(1)$ & $2k-1$ & ~\cite{DBLP:journals/jacm/ThorupZ05} & \\ \hline
    $\Ot(mn^{3/k})$ & $\Ot(m+n^{1+1/k})$ & $\Ot(\mu n^{c/k})$ & $2k-1-4c$ & Th\ref{T-DO-2k-1-4c-Weighted} & $k\ge5$, $1\leq c<k/2$\\ \hline
    $\Ot(mn^{3/k})$ & $\Ot(m+n^{1+1/k})$ & $\Ot(\mu n^{r})$ & $2k(1-2r)-1$ & Th\ref{T-DO-2k-1-4c-Weighted} & $r = c/k$, $r<1/2$\\ \hline
    \Xhline{3\arrayrulewidth}
    $\Ot(mn^{1/k})$ & $\Ot(m+n^{1+1/k})$ & $\Ot(\mu n^{1/k})$ & $2k-2$ & \cite{DBLP:conf/focs/DalirrooyfardJW22} & $k\ge 3$, Unweighted\\ \hline 
    $\Ot(mn^{1/k})$ & $\Ot(m+n^{1+1/k})$ & $\Ot(\mu n^{1/k})$ & $2k-3$ & Th\ref{T-2k-3-Weighted} & $k\ge 3$\\ \hline 
    % $\Ot(mn^{2/k})$ & $\Ot(m+n^{1+1/k})$ & $\Ot(\mu n^{1/k})$ & $2k-3$ & Extension of \cite{DBLP:journals/corr/abs-1201-2703} & $k\ge 3$\\ \hline 
    $\Ot(mn^{2/k})$ & $\Ot(m+n^{1+1/k})$ & $\Ot(\mu n^{1/k})$ & $2k-4$ & Th\ref{T-2k-4c-Weighted} & $k\ge 4$\\ \hline 
    $\Ot(mn^{3/k})$ & $\Ot(m+n^{1+1/k})$ & $\Ot(\mu n^{1/k})$ & $2k-5$ & Th\ref{T-DO-2k-1-4c-Weighted} & $k\ge 5$\\
    \Xhline{3\arrayrulewidth}
    $\Ot(mn^{2/3})$ & $\Ot(m+n^{4/3})$ & $\Ot(\mu n^{1/3})$ & $3$ & ~\cite{DBLP:journals/corr/abs-1201-2703} & \\ \hline
    $\Ot(mn^{1/3})$ & $\Ot(m+n^{4/3})$ & $\Ot(\mu n^{1/3})$ & $3$ & 
    Th\ref{T-2k-3-Weighted} &  \\ \hline
    $\Ot(mn^{1/2})$ & $\Ot(m+n^{5/4})$ & $\Ot(\mu n^{1/4})$ & $4$ & 
    Th\ref{T-2k-4c-Weighted} & \\ \hline
    $\Ot(mn^{3/5})$ & $\Ot(m+n^{6/5})$ & $\Ot(\mu n^{1/5})$ & $5$ & Th\ref{T-DO-2k-1-4c-Weighted} & \\ 
    \Xhline{3\arrayrulewidth}
    $\Ot(mn^{1-c})$ & $\Ot(m+n^{2-c})$ & $\Ot(\mu^t n^{tc})$  & $ 1+1/t$& ~\cite{DBLP:conf/esa/Agarwal14} & $c<1/2$ \\ \hline
    $\Ot(mn^{1-c})$ & $\Ot(m+n^{2-c})$ & $\Ot(\mu^t n^{(t+1)c})$  & $ 1+1/(t+0.5)$& ~\cite{DBLP:conf/esa/Agarwal14} & $c<1/2$ \\ \hline
    $\Ot(mn^{1-c})$ & $\Ot(m+n^{2-2c})$ & $\Ot(\mu^t n^{tc})$ & $1+2/t$ & Th\ref{T-Weighted-5/3} & $c<1/3$\\ \hline
    % % \Xhline{3\arrayrulewidth}
    % $\Ot(mn^{1/2+\eps})$ & $\Ot(m+n^{3/2+\epsilon})$ & $\Ot(\mu n^{1-2\epsilon})$ & $5/3$& ~\cite{DBLP:conf/esa/Agarwal14} & $\eps>0$ \\ \hline
    % $\Ot(mn^{2/3+\eps})$ & $\Ot(m+n^{4/3+2\epsilon})$ & $\Ot(\mu^3 n^{1-3\epsilon})$ & $5/3$& Theorem~\ref{T-Weighted-5/3} & $\eps>0$ \\ 
    % \Xhline{3\arrayrulewidth}
    % $\Ot(mn^{1/k})$ & $\Ot(m+n^{(1-c)(1+1/k)})$ & $\Ot(\mu n^{c})$ &  $4k-1$  & ~\cite{DBLP:journals/corr/abs-1201-2703} & $c<\frac{1}{k-1})$\\ \hline 
    % $\Ot(mn^{(1-c)(1/k)})$ & $\Ot(m+n^{(1-c)(1+)1/k)})$ & $\Ot(\mu^2 n^{2c})$ & $\frac{4}{3}(2k-1)$  & Theorem~\ref{T-Oracle-4/3(2k-1)d(u,v)-approx-Weighted} & $c<\frac{1}{2k+4})$ \\ \hline 
    \end{tabular}
    \caption{Distance oracles in weighted graphs.}
    \label{tab:weighted-new}
\end{table}

\subsection{$n$-PSP and $ANSC$ results}
\begin{table}[H]
    \centering
    \begin{tabular}{|c|c|c|c|}
    \hline
    Time  & Approximation & Ref. & Comment\\
     \hline\hline
     $\Ot(m+n^{1+2/k})$ & $(2k-1)(2k-2)\delta+4k-2$ & \cite{DBLP:conf/focs/DalirrooyfardJW22} & Unweighted, $k\ge 2$\\ \hline
     $\Ot(m+n^{1+2/k})$ & $(2k-1)(2k-3)\delta$ & Theorem~\ref{T-NPSP-2k-1-2k-3} & $k\ge 2$\\
    \Xhline{3\arrayrulewidth}
    $\Omega(m^{2-\frac{2}{k+1}} n^{\frac{1}{k+1}-o(1)})$ & $(1+1/k)\delta$ & ~\cite{DBLP:conf/focs/DalirrooyfardJW22} & \\ \hline
    $\Ot(m^{2-\frac{2}{k+1}} n^{\frac{1}{k+1}})$ & $(1+1/k)\delta$ & ~\cite{DBLP:conf/esa/Agarwal14}\footnotemark & \\ \hline
    $\Ot(m^{1-\frac{1}{k+1}}n)$ & $\delta + 2\lceil \delta / 2k \rceil$ &  Theorem~\ref{T-NPSP-1+1/k-Unweighted} & Unweighted, $k\ge 2$\\
    \hline
    \end{tabular}
    \caption{$nPSP$ in weighted graphs unless stated otherwise. $\delta=d(u,v)$. }
    \label{tab:npsp}
\end{table}
\footnotetext{This result is obtained by constructing the distance oracle of Agarwal~\cite{DBLP:conf/esa/Agarwal14} and querying it $n$ times.}

\begin{table}[H]
    \centering
    \begin{tabular}{|c|c|c|c|}
    \hline
    Time  & Approximation & Ref. & Comment\\
     \hline\hline
    $\Ot(m^{2-2/k}n^{1/k})$ & $SC(u)+2\lceil SC(u)/2(k-1)\rceil$ & \cite{DBLP:conf/focs/DalirrooyfardJW22} & Unweighted, $k\ge 2$\\ \hline
    $\Ot(mn^{1-1/k})$ & $SC(u)+2\lceil SC(u)/2(k-1) \rceil$ & Theorem~\ref{T-ANSC-1+1/k-approx-unweighted} & Unweighted, $k\ge 2$\\ \hline
    $\Ot(m^{2-\frac{1}{k}})$ & $(1+\frac{1}{k-1})SC(u)$ & Theorem~\ref{T-ANSC-1+1/k-approx-weighted} & Weighted, $k\ge 2$\\ 
    \hline
    \end{tabular}
    \caption{$ANSC$ algorithm improvements in undirected graphs. $SC(u)$ is the shortest cycle for $u$. }
    \label{tab:ansc}
\end{table}

\section{Simple lower bound for non-adjacent vertex pairs}\label{S-LB-appendix}
In this section, we prove the following theorem.
\Reminder{T-LB-2}

\begin{proof}
    Let $G$ be a graph with girth $\geq 2k+2$ and $\Omega(n^{1+1/k})$ edges, given from the Erd\H{o}s conjecture. 
    For every $v\in V$, we add a vertex $v'$ and an edge $(v,v')$ to $G$. 
    Any distance oracle with $o(n^{1+1/k})$ on $G$ must have at an edge $(u,v)\in E$ that is not preserved, and therefore $\hat{d}(u,v)>1$.
    
    Let $d_{>1}(u,v)$ be the second shortest distance between $u$ and $v$ in $G$.
    Since the girth $\geq 2k+2$ we know that $d_{>1}(u,v)\ge 2k+1$. 
    Since $\hat{d}(u,v)>1$ and the distance oracle returns as an estimation a length of a path in $G$ we get that $\hat{d}(u,v) \ge d_{>1}(u,v) \ge 2k+1$. 
    Moreover, $d_{>1}(u',v) = 1+d_{>1}(u',v) \ge 2k+2$, and therefore $\hat{d}(u',v) \ge 2k+2$, as required.
\end{proof}

\end{document}